\documentclass[%
 reprint,
 amsmath,amssymb,
 aps,
 prx,
]{revtex4-2}

\usepackage{graphicx}
\usepackage{dcolumn}
\usepackage{bm}
\usepackage[colorlinks,allcolors=blue]{hyperref}
\usepackage[capitalize]{cleveref}  

\usepackage{amsthm}
\usepackage{braket}
\usepackage[english]{babel}
\newtheorem{theorem}{Theorem}
\newtheorem{observation}[theorem]{Observation}

\newtheorem{lemma}[theorem]{Lemma}
\newtheorem{corollary}[theorem]{Corollary}
\usepackage[export]{adjustbox}
\usepackage{enumerate}

\newcommand{\vtwo}[1]{{\color{black} #1}}

\DeclareMathOperator{\Tr}{Tr}
\DeclareMathSymbol{\shortminus}{\mathbin}{AMSa}{"39}

\usepackage{tikz}
\usetikzlibrary{decorations.pathreplacing,calligraphy,decorations.markings}
\usetikzlibrary {arrows.meta}

\definecolor{tensorcolor}{rgb}{0.65,0.77,0.95}
\definecolor{btensorcolor}{rgb}{0.65,0.50,0.69}
\definecolor{whitetensorcolor}{rgb}{0.93,0.93,0.93}
\definecolor{gtensorcolor}{rgb}{0.6,0.8,0.5}

\newcommand\singledx{1.8}

\newcommand{\HollowDot}[1]{
	\begin{scope}[shift={(#1)}]
        \filldraw (0,0) circle (6pt);
        \filldraw[color=white] (0,0) circle (3pt);
	\end{scope}
}

\newcommand{\HollowRedDot}[1]{
	\begin{scope}[shift={(#1)}]
        \filldraw[color=red] (0,0) circle (6pt);
        \filldraw[color=white] (0,0) circle (3pt);
	\end{scope}
}

\newcommand{\Arrow}[3]{
	\begin{scope}[shift={(#1)}]
        \ifnum#3=0
            \draw[very thick] (-#2,-#2) -- (0,0);
            \draw[very thick] (#2,-#2) -- (0,0);
        \fi
        \ifnum#3=1
            \draw[very thick] (-#2,#2) -- (0,0);
            \draw[very thick] (-#2,-#2) -- (0,0);
        \fi
        \ifnum#3=2
            \draw[very thick] (#2,#2) -- (0,0);
            \draw[very thick] (-#2,#2) -- (0,0);
        \fi
        \ifnum#3=3
            \draw[very thick] (#2,-#2) -- (0,0);
            \draw[very thick] (#2,#2) -- (0,0);
        \fi
	\end{scope}
}

\newcommand{\RedArrow}[3]{
	\begin{scope}[shift={(#1)}]
        \ifnum#3=0
            \draw[very thick, draw=red] (-#2,-#2) -- (0,0);
            \draw[very thick, draw=red] (#2,-#2) -- (0,0);
        \fi
        \ifnum#3=1
            \draw[very thick, draw=red] (-#2,#2) -- (0,0);
            \draw[very thick, draw=red] (-#2,-#2) -- (0,0);
        \fi
        \ifnum#3=2
            \draw[very thick, draw=red] (#2,#2) -- (0,0);
            \draw[very thick, draw=red] (-#2,#2) -- (0,0);
        \fi
        \ifnum#3=3
            \draw[very thick, draw=red] (#2,-#2) -- (0,0);
            \draw[very thick, draw=red] (#2,#2) -- (0,0);
        \fi
	\end{scope}
}

\newcommand{\GTensor}[5]{
	\begin{scope}[shift={(#1)}]
    \ifnum#5=0
		\draw[very thick, draw=red] (-#2,0) -- (#2,0);
            \draw[very thick] (0,#2) -- (0,0);
    \fi
    \ifnum#5=-1
		\draw[very thick] (0,0) -- (#2,0);
            \draw[very thick] (0,#2) -- (0,0);
    \fi
    \ifnum#5=1
		\draw[very thick] (-#2,0) -- (0,0);
            \draw[very thick] (0,#2) -- (0,0);
    \fi

    \ifnum#5=2
		\draw[very thick,draw=red] (-#2,0) -- (#2,0);
    \fi
    \ifnum#5=3
		\draw[very thick] (0,-#2) -- (0,#2);
    \fi
    \ifnum#5=4
		\draw[very thick] (-#2,0) -- (#2,0);
    \fi
    \ifnum#5=5
		\draw[very thick, draw=red] (-#2,0) -- (#2,0);
		\draw[very thick] (0,#2) -- (0,-#2);
    \fi
    \ifnum#5=6
		\draw[very thick] (-#2,0) -- (#2,0);
            \draw[very thick] (0,#2) -- (0,0);
    \fi
    \ifnum#5=7
		\draw[very thick, draw=red] (-#2,0) -- (#2,0);
            \draw[very thick] (0,-#2) -- (0,0);
    \fi
    \ifnum#5=8
		\draw[very thick] (-#2,0) -- (#2,0);
    \fi
        \draw[ thick, fill=tensorcolor, rounded corners=2pt] (-#3,-#3) rectangle (#3,#3);
		\draw (0,0) node {\scriptsize #4};
	\end{scope}
}

\newcommand{\GFTensor}[5]{
	\begin{scope}[shift={(#1)}]
    \ifnum#5=0
		\draw[very thick, draw=red] (-#2,0) -- (\singledx+#2,0);
            \draw[very thick, draw=red] (0,#2) -- (0,0);
            \draw[very thick, draw=red] (\singledx,#2) -- (\singledx,0);
    \fi
    \ifnum#5=1
		\draw[very thick, draw=red] (-#2,0) -- (\singledx+#2,0);
            \draw[very thick, draw=red] (0,-#2) -- (0,0);
            \draw[very thick, draw=red] (\singledx,-#2) -- (\singledx,0);
    \fi
    \ifnum#5=2
		\draw[very thick, draw=red] (0,-#2) -- (0,#2);
            \draw[very thick, draw=red] (\singledx,-#2) -- (\singledx,#2);
    \fi
    \ifnum#5=3
		\draw[very thick] (0,-#2) -- (0,#2);
    \fi
    \ifnum#5=4
		\draw[very thick] (-#2,0) -- (#2,0);
    \fi
        \draw[ thick, fill=tensorcolor, rounded corners=2pt] (-#3,-#3) rectangle (#3+\singledx,#3);
		\draw (0+0.5*\singledx,0) node {\scriptsize #4};
	\end{scope}
}

\newcommand{\UFTensor}[5]{
	\begin{scope}[shift={(#1)}]
    \ifnum#5=0
            \draw[very thick] (0,#2) -- (0,-#2);
            \draw[very thick] (\singledx,#2) -- (\singledx,-#2);
    \fi
    \ifnum#5=1
		\draw[very thick, draw=red] (0,-#2) -- (0,0);
            \draw[very thick, draw=red] (\singledx,-#2) -- (\singledx,0);
            \draw[very thick] (0.5*\singledx,#2) -- (0.5*\singledx,0);
    \fi

    \ifnum#5=2
		\draw[very thick, draw=red] (0,#2) -- (0,0);
            \draw[very thick, draw=red] (\singledx,#2) -- (\singledx,0);
            \draw[very thick] (0.5*\singledx,-#2) -- (0.5*\singledx,0);
    \fi
    \ifnum#5=3
		\draw[very thick, draw=red] (0,#2) -- (0,-#2);
            \draw[very thick, draw=red] (\singledx,#2) -- (\singledx,-#2);
    \fi
    \ifnum#5=4
		\draw[very thick] (-#2,0) -- (#2,0);
    \fi
        \draw[ thick, fill=whitetensorcolor, rounded corners=2pt] (-#3,-#3) rectangle (#3+\singledx,#3);
		\draw (0+0.5*\singledx,0) node {\scriptsize #4};
	\end{scope}
}

\newcommand{\GPTensor}[5]{
	\begin{scope}[shift={(#1)}]
    \ifnum#5=0
		\draw[very thick, draw=red] (-#2,0) -- (#2,0);
            \draw[very thick, draw=red] (0,-#2) -- (0,#2);
    \fi
    \ifnum#5=1
		\draw[very thick, draw=red] (-#2,0.2) -- (#2,0.2);
            \draw[very thick, draw=red] (-#2,-0.2) -- (#2,-0.2);
            \draw[very thick, draw=red] (0.2,-#2) -- (0.2,#2);
            \draw[very thick, draw=red] (-0.2,-#2) -- (-0.2,#2);
    \fi

    \ifnum#5=2
		\draw[very thick, draw=red] (-#2,0) -- (#2,0);
            \draw[very thick, draw=red] (0,-#2) -- (0,#2);
    \fi

    \ifnum#5=3
		\draw[very thick] (0,-#2) -- (0,#2);
    \fi
        \draw[ thick, fill=tensorcolor, rounded corners=2pt] (-#3,-#3) rectangle (#3,#3);
	\draw (0,0) node {\scriptsize #4};
    \ifnum#5=0
		\draw[very thick] (#3/2,#3/2) -- (#2,#2);
    \fi
	\end{scope}
}

\newcommand{\UPTensor}[5]{
	\begin{scope}[shift={(#1)}]
    \ifnum#5=0
		\draw[very thick, draw=red] (-#2,0) -- (#2,0);
            \draw[very thick, draw=red] (0,-#2) -- (0,#2);
    \fi
    \ifnum#5=1
		\draw[very thick, draw=red] (-#2,0.2) -- (#2,0.2);
            \draw[very thick, draw=red] (-#2,-0.2) -- (#2,-0.2);
            \draw[very thick, draw=red] (0.2,-#2) -- (0.2,#2);
            \draw[very thick, draw=red] (-0.2,-#2) -- (-0.2,#2);
    \fi

    \ifnum#5=2
		\draw[very thick] (-#2,0) -- (#2,0);
            \draw[very thick] (0,-#2) -- (0,#2);
    \fi

    \ifnum#5=3
		\draw[very thick, draw=red] (-#2,0) -- (#2,0);
            \draw[very thick, draw=red] (0,-#2) -- (0,#2);
    \fi
        \draw[ thick, fill=whitetensorcolor, rounded corners=2pt] (-#3,-#3) rectangle (#3,#3);
	\draw (0,0) node {\scriptsize #4};
    \ifnum#5=0
		\draw[very thick] (#3/2,#3/2) -- (#2,#2);
    \fi
    \ifnum#5=2
		\draw[very thick] (#3/2,#3/2) -- (0.75*#2,0.75*#2);
    \fi
    \ifnum#5=3
		\draw[very thick] (#3/2,#3/2) -- (0.75*#2,0.75*#2);
    \fi
	\end{scope}
}

\newcommand{\EPTensor}[5]{
	\begin{scope}[shift={(#1)}]
    \ifnum#5=0
		\draw[very thick, draw=red] (-#2,0.2) -- (#2,0.2);
            \draw[very thick, draw=red] (-#2,-0.2) -- (#2,-0.2);
            \draw[very thick, draw=red] (0.2,-#2) -- (0.2,#2);
            \draw[very thick, draw=red] (-0.2,-#2) -- (-0.2,#2);
    \fi
    \ifnum#5=1
		\draw[very thick, draw=red] (-#2,0.2) -- (#2,0.2);
            \draw[very thick, draw=red] (-#2,-0.2) -- (#2,-0.2);
            \draw[very thick, draw=red] (0.2,-#2) -- (0.2,#2);
            \draw[very thick, draw=red] (-0.2,-#2) -- (-0.2,#2);
    \fi

    \ifnum#5=2
		\draw[very thick,draw=red] (-#2,0) -- (#2,0);
    \fi

    \ifnum#5=3
		\draw[very thick] (0,-#2) -- (0,#2);
    \fi
        \draw[ thick, fill=tensorcolor, rounded corners=2pt] (-#3,-#3) rectangle (#3,#3);
	\draw (0,0) node {\scriptsize #4};
	\end{scope}
}

\newcommand{\Unitary}[5]{
	\begin{scope}[shift={(#1)}]
    \ifnum#5=0
		\draw[very thick, draw=red] (-#2,0) -- (#2,0);
            \draw[very thick] (-#2,2*#3) -- (#2,2*#3);
    \fi
    \ifnum#5=-1
		\draw[very thick] (0,0) -- (#2,0);
            \draw[very thick] (0,#2) -- (0,0);
    \fi
    \ifnum#5=1
		\draw[very thick, draw=red] (-#2,0) -- (#2,0);
            \draw[very thick] (0,2*#3) -- (#2,2*#3);
    \fi

    \ifnum#5=2
		\draw[very thick, draw=red] (-#2,0) -- (#2,0);
            \draw[very thick, draw=red] (-#2,2*#3) -- (#2,2*#3);
    \fi

    \ifnum#5=3
		\draw[very thick] (0,-#2) -- (0,#2);
    \fi
        \draw[ thick, fill=tensorcolor, rounded corners=2pt] (-#3,-#3) rectangle (#3,3*#3);
		\draw (0,#3) node {\scriptsize #4};
	\end{scope}
}

\newcommand{\FUnitary}[5]{
	\begin{scope}[shift={(#1)}]
    \ifnum#5=0
		\draw[very thick, draw=red] (-#2,0) -- (#2,0);
            \draw[very thick, draw=red] (-#2,2*#3) -- (#2,2*#3);
            \draw[very thick, draw=red] (-#2,4*#3) -- (#2,4*#3);
    \fi
    \ifnum#5=1
		\draw[very thick, draw=red] (-#2,0) -- (#2,0);
            \draw[very thick, draw=red] (0,2*#3) -- (#2,2*#3);
            \draw[very thick, draw=red] (0,4*#3) -- (#2,4*#3);
    \fi
    \ifnum#5=2
		\draw[very thick,draw=red] (-#2,0) -- (#2,0);
    \fi

    \ifnum#5=3
		\draw[very thick] (0,-#2) -- (0,#2);
    \fi
        \draw[ thick, fill=tensorcolor, rounded corners=2pt] (-#3,-#3) rectangle (#3,5*#3);
		\draw (0,2*#3) node {\scriptsize #4};
	\end{scope}
}

\newcommand{\PTensor}[5]{
	\begin{scope}[shift={(#1)}]
    \ifnum#5=0
		\draw[very thick, draw = red] (0,-#2) -- (0,0);
            \draw[very thick, draw = red] (0,0) -- (0,#2);
    \fi
    \ifnum#5=1
		\draw[very thick, draw = red] (-#2,0) -- (0,0);
            \draw[very thick, draw = red] (0,0) -- (#2,0);
    \fi
    \ifnum#5=2
		\draw[very thick, draw = red] (0,-#2) -- (0,0);
            \draw[very thick, draw = red] (0,0) -- (0,#2);
    \fi
    \ifnum#5=3
		\draw[very thick, draw = red] (-#2,0) -- (0,0);
            \draw[very thick, draw = red] (0,0) -- (#2,0);
    \fi
        \draw[ thick, fill=tensorcolor, rounded corners=2pt] (-#3,-#3) rectangle (#3,#3);
		\draw (0,0) node {\scriptsize #4};
	\end{scope}
}

\newcommand{\BTensor}[5]{
	\begin{scope}[shift={(#1)}]
    \ifnum#5=0
		\draw[very thick, draw=red] (-#2,0) -- (#2,0);
    \fi
    \ifnum#5=1
		\draw[very thick,draw=red] (0,#2) -- (0,-#2);
    \fi
    \ifnum#5=2
		\draw[very thick] (-#2,0) -- (#2,0);
    \fi
    \ifnum#5=3
		\draw[very thick] (0,#2) -- (0,-#2);
    \fi

        \draw[ thick, fill=tensorcolor, rounded corners=2pt] (-#3,-#3) rectangle (#3,#3);
		\draw (0,0) node {\scriptsize #4};
	\end{scope}
}

\newcommand{\DTensor}[5]{
	\begin{scope}[shift={(#1)}]
    \ifnum#5=0
		\draw[very thick, draw=red] (-#2,0) -- (#2,0);
    \fi
    \ifnum#5=1
		\draw[very thick,draw=red] (0,#2) -- (0,-#2);
    \fi
    \ifnum#5=2
		\draw[very thick] (-#2,0) -- (#2,0);
    \fi
    \ifnum#5=3
		\draw[very thick] (0,#2) -- (0,-#2);
    \fi

        \draw[ thick, fill=whitetensorcolor, rounded corners=2pt] (-#3,-#3) rectangle (#3,#3);
		\draw (0,0) node {\scriptsize #4};
	\end{scope}
}

\newcommand{\UTensor}[5]{
	\begin{scope}[shift={(#1)}]
    \ifnum#5=0
		\draw[very thick, draw=red] (-#2,0) -- (#2,0);
    \fi
    \ifnum#5=1
		\draw[very thick,draw=red] (0,#2) -- (0,-#2);
    \fi
    \ifnum#5=2
		\draw[very thick] (-#2,0) -- (#2,0);
    \fi
    \ifnum#5=3
		\draw[very thick] (0,#2) -- (0,-#2);
    \fi
        \draw[ thick, fill=gtensorcolor, rounded corners=2pt] (-#3,-#3) rectangle (#3,#3);
		\draw (0,0) node {\scriptsize #4};
	\end{scope}
}

\newcommand{\RTensor}[5]{
	\begin{scope}[shift={(#1)}]
    \ifnum#5=0
		\draw[very thick, draw=red] (-#2,0) -- (#2,0);
    \fi
    \ifnum#5=1
		\draw[very thick,draw=red] (0,#2) -- (0,-#2);
    \fi
    \ifnum#5=2
		\draw[very thick] (-#2,0) -- (#2,0);
    \fi
    \ifnum#5=3
		\draw[very thick] (0,#2) -- (0,-#2);
    \fi
        \draw[ thick, fill=violet, rounded corners=2pt] (-#3,-#3) rectangle (#3,#3);
		\draw (0,0) node {\scriptsize #4};
	\end{scope}
}

\newcommand{\GDTensor}[5]{
	\begin{scope}[shift={(#1)}]
    \ifnum#5=0
		\draw[very thick] (-#2,0) -- (#2,0);
		\draw[very thick] (0,#2) -- (0,-#2);
    \fi
    \ifnum#5=-1
		\draw[very thick] (0,0) -- (#2,0);
		\draw[very thick] (0,#2) -- (0,-#2);
    \fi
    \ifnum#5=1
		\draw[very thick] (-#2,0) -- (0,0);
		\draw[very thick] (0,#2) -- (0,-#2);
    \fi
        \draw[ thick, fill=tensorcolor, rounded corners=2pt] (-#3,-#3) rectangle (#3,#3);
    \def\dx{#3/3};
	\draw [thick]  (-#3+\dx, \dx) -- (- \dx,#3-\dx);
	\draw [thick] (-#3+1.5*\dx,-#3+1.5*\dx) -- (#3-1.5*\dx,#3-1.5*\dx);
	\draw [thick]  ( \dx, -#3 + \dx) -- (#3 - \dx,-\dx);
	\draw (0,0) node {\scriptsize #4};
	\end{scope}
}

\newcommand{\ATensor}[3]{
    \GTensor{#1}{1}{.6}{#2}{#3};
}

\newcommand{\BendUpRight}[2]{
	\begin{scope}[shift={(#1)}]
        \draw [very thick] (0,0) arc (270:360:0.4);
        \draw [very thick] (0.4,0.4) to (0.4,#2-0.4);
	\end{scope}
}

\newcommand{\BendUpLeftRed}[2]{
	\begin{scope}[shift={(#1)}]
        \draw [very thick, draw=red] (0,0) arc (270:180:0.4);
        \draw [very thick, draw=red] (-0.4,0.4) to (-0.4,#2-0.4);
	\end{scope}
}

\newcommand{\BendUpRightRed}[2]{
	\begin{scope}[shift={(#1)}]
        \draw [very thick, draw=red] (0,0) arc (270:360:0.4);
        \draw [very thick, draw=red] (0.4,0.4) to (0.4,#2-0.4);
	\end{scope}
}

\newcommand{\CNOT}[2]{
	\begin{scope}[shift={(#1)}]
        \filldraw (0,0) circle (3pt);
        \draw[thick] (0,0) -- (0,-#2-0.4);
        \draw[thick] (0,-#2) circle [radius=0.4]; 
	\end{scope}
}

\newcommand{\Gate}[2]{
	\begin{scope}[shift={(#1)}]
        \draw[ thick, fill=white, rounded corners=2pt] (-0.6,-0.6) rectangle (0.6,0.6);
		\draw (0,0) node {\scriptsize #2};
	\end{scope}
}

\newcommand{\Measurement}[1]{
	\begin{scope}[shift={(#1)}]
        \draw[ thick, fill=white, rounded corners=2pt] (-0.6,-0.6) rectangle (0.6,0.6);
	\draw[thick] (0.35,0.2) arc (30:150:0.4);
        \draw[thick] (0,-0.4) -- (0.3,0.5);
	\end{scope}
}

\newcommand\subsetsim{\mathrel{%
  \ooalign{\raise0.2ex\hbox{$\subset$}\cr\hidewidth\raise-0.8ex\hbox{\scalebox{0.9}{$\sim$}}\hidewidth\cr}}}

\begin{document}


\title{Characterizing MPS and PEPS Preparable via Measurement and Feedback}

\author{Yifan Zhang}
\email{yz4281@princeton.edu}
 \affiliation{Department of Electrical and Computer Engineering, Princeton University}

 \author{Sarang Gopalakrishnan}
\email{sarang.gopalakrishnan@princeton.edu}
 \affiliation{Department of Electrical and Computer Engineering, Princeton University}
 
\author{Georgios Styliaris}%
 \email{georgios.styliaris@mpq.mpg.de}
\affiliation{Max-Planck-Institut für Quantenoptik, Hans-Kopfermann-Straße 1, 85748 Garching, Germany}%
 \affiliation{Munich Center for Quantum Science and Technology (MCQST), Schellingstr. 4, D-80799 München, Germany}
\date{\today}

\begin{abstract}
Preparing long-range entangled states poses significant challenges for near-term quantum devices. It is known that measurement and feedback (MF) can aid this task by allowing the preparation of certain paradigmatic long-range entangled states with only constant circuit depth. Here we systematically explore the structure of states that can be prepared using constant-depth local circuits and a single MF round. Using the framework of tensor networks, the preparability under MF translates to tensor symmetries. We detail the structure of matrix-product states (MPS) and projected entangled-pair states (PEPS) that can be prepared using MF, revealing the coexistence of Clifford-like properties and magic.
\vtwo{In one dimension, we show that states with abelian symmetry protected topological order are a restricted class of MF-preparable states. In two dimensions, we parameterize a subset of states with abelian topological order that are MF-preparable.} 
Finally, we discuss the analogous implementation of operators via MF, providing a structural theorem that connects to the well-known Clifford teleportation.
\end{abstract}

\maketitle







\section{Introduction}

The recent advancement in building large-scale quantum computers and simulators opens up new possibilities in studying strongly-correlated many-body quantum states in the engineered quantum platforms.
As the number of available qubits continues to scale, an important challenge is to prepare such many-body quantum states from a native initial configuration, such as the product state.
This becomes increasingly challenging for larger system sizes because physically interesting, strongly-correlated states often require high-depth circuits.
Realizing the latter can be prohibitive in the absence of fault tolerance, since the accumulation of errors scales unfavorably with the circuit depth~\cite{preskill2018quantum}.
Moreover, locality of the accessible interactions constrains even further the states that can be prepared with shallow circuits because long-range correlated states, even with a small amount of entanglement, require depth at least linear in the system size~\cite{bravyi2006lieb}.
Some examples include states in the symmetry-broken phase exhibiting long-range correlations~\cite{aharonov1998quantum} (e.g., the GHZ state~\cite{nielsen2002quantum}), and error-correcting code states with topological order~\cite{kitaev2003fault,bravyi2006lieb,aharonov2018quantum}.

%

One shortcut to avoid the high circuit depth is via utilizing measurement and feedback (MF). By performing projective measurements and applying error correction (unitary) operators depending on the measurement outcome, one can deterministically prepare states in a constant circuit depth that would otherwise requires a deep circuit using local unitary interactions.
\vtwo{One way to understand the power of MF is via locality. In a constant-depth local unitary circuit, correlations cannot propagate beyond the light cone of the quantum circuit \cite{bravyi2006lieb}. On the other hand, in MF the (classical) measurement outcome is shared non-locally to determine the error correction operators, so one can spread correlations outside the circuit light cone.}
%
%
 In other words, the MF paradigm can be considered as a many-body, multi-party version of local operation and classical communication (LOCC)~\cite{horodecki2009quantum}, which is enhanced by allowing shallow circuits. 
%

Several recent works have discussed preparing many-body states using MF, both in theory \cite{piroli2021quantum,tantivasadakarn2021long,bravyi2022adaptive,lu2022measurement,tantivasadakarn2023shortest,tantivasadakarn2023hierarchy,smith2023deterministic,buhrman2023state,gunn2023phases,zhu2023nishimori,li2023symmetry,malz2024preparation} and in experiments \cite{foss2302experimental,iqbal2024non,chen2023realizing,baumer2023efficient}. However, in the particularly relevant case of constant-depth MF, most of the studies so far focus on specific examples, and a systematic understanding of states preparable using MF is still lacking.
%
%
Here we address the following question: What is the structure of states preparable from product states using constant-depth circuits and a single round of MF? We refer to such states as MF states and we will primarily discuss them in 1D and 2D.
Specifically, we focus on the following two aspects to address their generality.

First, an obvious class of MF states is the stabilizer states with small, local stabilizer generators. For example, the well-understood toric code belongs to this category~\cite{raussendorf2005long,aguado2008creation}. Nevertheless, we also know states preparable using MF that are not stabilizer states such as the AKLT state \cite{smith2023deterministic} and the double-semion order \cite{lu2022measurement,tantivasadakarn2021long}. These states exhibit Clifford-like properties akin to the stabilizer states, so one would wonder if there are any connections. We show that there is indeed always a Clifford-like structure underlying MF states, and moreover provide a way to understand where the magic (non-stabilizerness) originates from and why it does not interfere with the Clifford-like properties. To give a taste of the Clifford+magic structure, under certain conditions specified later, if a matrix-product state (MPS) is preparable with MF, then its tensors admit the following decomposition:
\begin{align}
    \begin{tikzpicture}[scale=0.5,baseline={([yshift=-0.65ex] current bounding box.center) }]
        \GTensor{0,0}{1.2}{0.6}{\small $A$}{0};
\end{tikzpicture}
    \; \propto \;
    \begin{tikzpicture}[scale=0.5,baseline={([yshift=-2.65ex] current bounding box.center) }]
        \draw[very thick, draw=red](-3,0)--(3,0);
        \FUnitary{0,0}{1.2}{.6}{\small $U_C$}{0};
        \draw[ thick, rounded corners=2pt] (-1.2-1.2,1.2-0.6) rectangle (-1.2,3.6-0.6);
        \draw (-1.2-0.6,2.4-0.6) node {\small $\ket{\psi}$};
        \draw[very thick, draw=red](1.2,2.4) -- (1.2,3.6);
        \draw[very thick, draw=red](1.2,1.2) -- (1.2+\singledx,1.2) -- (1.2+\singledx,3.6);
        \UFTensor{0.3+0.5*\singledx,4.0}{1.2}{.6}{\small $V$}{1};
\end{tikzpicture}
\;.
\end{align}
Here \(U_C\) is a Clifford unitary, \(\ket{\psi}\) is some (possibly magic) initial state, and \(V\) is an inverse isometry to account for local rotations and potential non-injectivity. Similar structure exists in two dimensions as well.

Second, we provide analytic solutions to the states preparable using MF under some physically-motivated symmetries. In 1D, we analytically solve MF states under the symmetry resembling the symmetry protected topological (SPT) order~\cite{chen2011classification,schuch2011classifying}, while in 2D, we analytically solve MF states under the symmetry resembling the topological order and discuss their relevant properties. Importantly, the analytic solution parameterizes a continuous family of states, again highlighting the fact that MF states can possess magic.

The major tool we employ is the language of tensor networks~\cite{cirac2021matrix}. In this language, the MF property translates to the symmetries of the MPS tensor in 1D and the symmetries of the projected entangled pair state (PEPS) tensor in 2D. We characterize these tensors (MF MPS and MF PEPS for short) by providing a structural theorem for each. Then, we write down the analytic solutions for MF MPS and MF PEPS when the symmetry resembles the SPT order in 1D and the topological order in 2D. In addition, the language of tensor network enables us the extend the discussion to not only preparing states but also applying unitary operators using MF.


This manuscript is organized as follows. In Section~\ref{section_mps} we discuss MF MPS. We begin by introducing the setting and protocol. We then prove the structural theorem to characterize MF MPS and to understand the source of magic. Finally, we provide an analytic solution of the MF MPS under the SPT-type symmetry.
In Section~\ref{section_peps} we extend the discussion to MF PEPS where we prove a similar structural theorem and provide the analytic solution under the symmetry resembling topological order. To probe the topological order, we further develop statements regarding the spectrum of the transfer matrix where its degeneracy constitutes a signature of the topological order.
Lastly, we discuss how to implement a class of unitaries using MF in Section~\ref{section_mpu} and provide a characterization.

\section{MPS}\label{section_mps}
\subsection{Overview of the Protocol}

We begin by analyzing a protocol for deterministically preparing MPS using MF. Similar protocols have been proposed in Refs.~\cite{smith2023deterministic,gunn2023phases}, and more recently in \cite{smith2024constant,sahay2024classifying,sahay2024finite,stephen2024preparing}.
The protocol begins by preparing unentangled tripartite states, each in $\mathbb C^D \otimes \mathbb C^d \otimes \mathbb C^D$, corresponding to the MPS tensors:
\begin{align}
    \dotso \quad
    \begin{tikzpicture}[scale=0.5,baseline={([yshift=-0.65ex] current bounding box.center) }]
        \GTensor{0,0}{1.2}{.6}{\small $A_{k}$}{0};
        \BendUpRightRed{1.2,0}{1.4};
        \BendUpLeftRed{-1.2,0}{1.4};
        \GTensor{0+4,0}{1.2}{.6}{\small $A_{k\!+\!1}$}{0};
        \BendUpRightRed{1.2+4,0}{1.4};
        \BendUpLeftRed{-1.2+4,0}{1.4};
        \GTensor{0-4,0}{1.2}{.6}{\small $A_{k\!-\!1}$}{0};
        \BendUpRightRed{1.2-4,0}{1.4};
        \BendUpLeftRed{-1.2-4,0}{1.4};
\end{tikzpicture}
    \quad \dotso
\end{align}
Here $d$ denotes the physical and $D$ the virtual dimension. Note that the virtual space is thus here treated as a physical degree of freedom.
To ``glue'' together the individual tensors to the desired MPS, we perform a bipartite projective measurement over the adjacent virtual qubits:
%
%
\vtwo{
\begin{equation}\label{mps_glue}
\begin{split}
&\textcolor{red}{...}\;
\begin{tikzpicture}[scale=0.5,baseline={([yshift=-0.65ex] current bounding box.center) }]
        \GTensor{0,0}{1.2}{.6}{\small $A_{k\!-\!1}$}{0};
        \BendUpRightRed{1.2,0}{1.4};
        \BendUpLeftRed{-1.2,0}{1.4};
        \GTensor{0+4,0}{1.2}{.6}{\small $A_{k}$}{0};
        \BendUpRightRed{1.2+4,0}{1.4};
        \BendUpLeftRed{-1.2+4,0}{1.4};
        \GTensor{0+8,0}{1.2}{.6}{\small $A_{k\!+\!1}$}{0};
        \BendUpRightRed{1.2+8,0}{1.4};
        \BendUpLeftRed{-1.2+8,0}{1.4};
        \Measurement{2,1.4};
        \Measurement{6,1.4};
\end{tikzpicture}
\;\textcolor{red}{...}\\
\;=\;
&\textcolor{red}{...}
\begin{tikzpicture}[scale=0.5,baseline={([yshift=-0.65ex] current bounding box.center) }]
        \GTensor{0,0}{1.2}{.6}{\small $A_{k\!-\!1}$}{0};
        \DTensor{\singledx,0}{1}{.6}{\small $P$}{0};
      \GTensor{2*\singledx,0}{1.2}{.6}{\small $A_{k}$}{0};
      \DTensor{3*\singledx,0}{1}{.6}{\small $P'$}{0};
      \GTensor{4*\singledx,0}{1.2}{.6}{\small $A_{k\!+\!1}$}{0};
\end{tikzpicture}
\textcolor{red}{...}
\end{split}
\end{equation}
}
However, in general, this results to unwanted operators \(P\) (representing the measurement outcome) intertwined within the target MPS. In certain cases, if the MPS tensors and the measurement satisfies certain compatibility conditions, the target MPS can be deterministically obtained by the action of local unitaries over the MPS physical space.

A simple example with $d = D = 2$ is given by choosing the MPS tensors to be 3-qubit stabilizer states. Then performing Bell measurements results in a MPS intertwined with Pauli operators, in a pattern that depends on the measurement outcomes. Nevertheless, due to the underlying Clifford property, the MPS is then equivalent to the target one up to local Pauli corrections.

However, recently several examples were constructed that the underlying tensors seemingly do not satisfy a (suitably generalized) Clifford property, but the MF protocol still succeeds. This notably includes the AKLT state. In the following we systematically study the structure of the MF protocol. We show that, under mild assumptions, the underlying quantum gate-teleportation~\cite{brassard1996teleportation,gottesman1999demonstrating} mechanism always relies on a Clifford property, which is possibly hidden by a part that possesses magic. 

\subsection{Setting}

We consider protocols where the gluing measurement is projective and consists of rank-1 orthogonal projectors.
In other words, the corresponding basis is a collection of states
\begin{align}
    \ket{P_i} = \; 
    \begin{tikzpicture}[scale=0.5,baseline={([yshift=-0.65ex] current bounding box.center) }]
        \BendUpLeftRed{(-1,0)}{1.2};
        \BendUpRightRed{(1,0)}{1.2};
        \draw[very thick,red] (-1,0) -- (1,0);
        \DTensor{0,0}{1.5}{.6}{\small $P_i$}{-1};
\end{tikzpicture}
\; = \; \sum_{a,b} (P_i)_{ab} \ket{a}\ket{b}
\end{align}
where \(\mathcal B = \{P_i\}\) is a collection of matrices. We refer to $\mathcal B$ as the MF basis. Different elements satisfy orthogonality 
\begin{align}
    \Tr(P_i^\dagger P_j) = \;
    \begin{tikzpicture}[scale=0.5,baseline={([yshift=-0.65ex] current bounding box.center)}]
        \draw[rounded corners,very thick, draw=red] (-1, -1) rectangle (1, 1);
        \DTensor{0,-1}{0}{.6}{\small $P^*_i$}{0};
        \DTensor{0,1}{0}{.6}{\small $P_j$}{0};
    \end{tikzpicture}
\; = \; \delta_{ij}
\end{align}
and completeness
\begin{align}
    \sum_{i=1}^{D^2} \;\;
    \begin{tikzpicture}[scale=0.5,baseline={([yshift=-0.65ex] current bounding box.center) }]
        \BendUpLeftRed{(-1,0)}{1.2};
        \BendUpRightRed{(1,0)}{1.2};
        \draw[very thick,red] (-1,0) -- (1,0);
        \DTensor{0,0}{1.5}{.6}{\small $P_i$}{-1};
        \begin{scope}[yscale=-1,xscale=1]
        \BendUpLeftRed{(-1,1.5)}{1.2};
        \BendUpRightRed{(1,1.5)}{1.2};
        \draw[very thick,red] (-1,1.5) -- (1,1.5);
        \DTensor{0,1.5}{1.5}{.6}{\small $P^*_i$}{-1};
        \end{scope}
    \end{tikzpicture}
    \;\; = \; I_{D^2} \;.
\end{align}
Both of these can be conveniently captured in the statement that the map $\ket{i} \mapsto \ket{P_i}$ is unitary. The Bell basis is the special case of qubits where (up to a constant) \(\mathcal B = \{I,X,Y,Z\}\). When generalized to qudits, the possible basis includes setting $\mathcal B$  to the Weyl-Heisenberg group which is generated by
\begin{align}
    X=\sum_{a=0}^{D-1} \ket{a+1}\!\bra{a}, \quad 
    Z=\sum_{a=0}^{D-1} e^{i 2\pi a / D} \ket{a}\!\bra{a} \;.
\end{align}
We will also invoke the notion of qudit Clifford unitaries and qudit stabilizers later on so we define it here. Similar to the qubit case, a qudit Clifford unitary is defined as a unitary operator that maps all tensor product of Weyl-Heisenberg operators to tensor product of Weyl-Heisenberg operators. A qudit stabilizer state is the simultaneous eigenstate of a (sufficiently large) subgroup of the tensor product of Weyl-Heisenberg group.
See Ref.~\cite{werner2001all} for a discussion on MF bases and Appendix~\ref{mf_basis_ex} for more examples of MF bases.


We impose an additional constraint: $\mathcal B$ forms a unitary projective representation of a group. In other words, \(\{P_i\}\) are unitaries and are closed under multiplication up to a phase. The group property imposes a non-trivial constraint because there exist constructions where \(\{P_i\}\) do not form a group~\cite{werner2001all}. See Appendix \ref{mf_basis_ex} for further discussion. On the other hand, unitarity constitutes a trivial constraint because any representation of a group can be turned into a unitary representation via a similarity transform \(G\) \cite{dresselhaus2007group}. For example, the similarity transform in the virtual bond shows up as the following:
\begin{align}
\begin{tikzpicture}[scale=0.5,baseline={([yshift=-0.65ex] current bounding box.center) }]
        \ATensor{0,0}{\small $A_1$}{0};
        \DTensor{\singledx,0}{1}{.6}{\small $G$}{0};
        \DTensor{2*\singledx,0}{1}{.6}{\small $P$}{0};
        \DTensor{3*\singledx,0}{1}{.6}{\small $G^{\shortminus 1}$}{0};
      \ATensor{4*\singledx,0}{\small $A_2$}{0};
\end{tikzpicture}
\end{align}
Meanwhile, the similarity transform corresponds to a gauge transformation of the MPS tensors so they do not modify the resulting state. An important observation which we later use extensively is that the unitary representation $\{P_i\}$ is necessarily irreducible. This follows from the completeness of the measurement.

Returning to Eq.~(\ref{mps_glue}), to remove the operator \(P\), we demand the MPS tensors to have the following symmetry for all $P \in \mathcal B$, which we call MF symmetry:
\begin{align}\label{mps_sym}
    \begin{tikzpicture}[scale=0.5,baseline={([yshift=-3.65ex] current bounding box.center) }]
        \DTensor{-\singledx,0}{1.}{.6}{\small $P$}{0};
        \UTensor{0,\singledx}{1.}{.6}{\small $U_P$}{3};
        \ATensor{0,0}{\small $A$}{0};
\end{tikzpicture}
\; = \;
    \begin{tikzpicture}[scale=0.5,baseline={([yshift=-0.65ex] current bounding box.center) }]
        \DTensor{\singledx,0}{1.}{.6}{\small $P'$}{0};
        \ATensor{0,0}{\small $A$}{0};
\end{tikzpicture}
\end{align}
This requirement allows to apply a local unitary \(U_P\) on the physical space (which depends on $P$) to move \(P\) to the right and become \(P' \in \mathcal B\). \(P'\) does not have to be equal to \(P\), nor does this map from left to right need to be bijective. This symmetry allows us to move the operator in the middle of the chain to the edge where it can be corrected:
\begin{align}
        & \;\;  \begin{tikzpicture}[scale=0.5,baseline={([yshift=-0.65ex] current bounding box.center) }]
        \BendUpLeftRed{(-1,0)}{1.};
        \UTensor{2*\singledx,\singledx}{1.}{.6}{\small $U_P$}{3};
        \ATensor{0,0}{\small $A_1$}{0};
        \GTensor{3*\singledx,0}{\singledx-0.4}{.6}{\small $A_3$}{0};
        \ATensor{2*\singledx,0}{\small $A_2$}{0};
        \UTensor{3*\singledx,\singledx}{1.}{.6}{\small $U_{P'}$}{3};
        \BendUpRightRed{(3*\singledx+\singledx-0.4,0)}{1.};
        \UTensor{4*\singledx,\singledx}{1.2}{.6}{\small $P^{''\dag}$}{1};
        \DTensor{\singledx,0}{1.}{.6}{\small $P$}{0};
        \end{tikzpicture} \\
        = & \;\; \begin{tikzpicture}[scale=0.5,baseline={([yshift=-4.65ex] current bounding box.center) }]
        \BendUpLeftRed{(-1,0)}{1.};
        \ATensor{0,0}{\small $A_1$}{0};
        \GTensor{3*\singledx,0}{\singledx-0.4}{.6}{\small $A_3$}{0};
        \ATensor{1*\singledx,0}{\small $A_2$}{0};
        \UTensor{3*\singledx,\singledx}{1.}{.6}{\small $U_{P'}$}{3};
        \BendUpRightRed{(3*\singledx+\singledx-0.4,0)}{1.};
        \UTensor{4*\singledx,\singledx}{1.2}{.6}{\small $P^{''\dag}$}{1};
        \DTensor{2*\singledx,0}{1.}{.6}{\small $P'$}{0};
        \end{tikzpicture}\\
        = & \;\; \begin{tikzpicture}[scale=0.5,baseline={([yshift=-4.65ex] current bounding box.center) }]
        \BendUpLeftRed{(-1,0)}{1.};
        \ATensor{0,0}{\small $A_1$}{0};
        \DTensor{3*\singledx,0}{1.4}{.6}{\small $P''$}{0};
        \GTensor{2*\singledx,0}{1.}{.6}{\small $A_3$}{0};
        \ATensor{1*\singledx,0}{\small $A_2$}{0};
        \BendUpRightRed{(3*\singledx+\singledx-0.4,0)}{1.};
        \UTensor{4*\singledx,\singledx}{1.2}{.6}{\small $P^{''\dag}$}{1};
        \end{tikzpicture}\\
         = & \;\; \begin{tikzpicture}[scale=0.5,baseline={([yshift=-0.65ex] current bounding box.center) }]
        \BendUpLeftRed{(-1,0)}{1.};
        \ATensor{0,0}{\small $A_1$}{0};
        \GTensor{2*\singledx,0}{1.}{.6}{\small $A_3$}{0};
        \ATensor{1*\singledx,0}{\small $A_2$}{0};
        \BendUpRightRed{(2*\singledx+1,0)}{1.};
        \end{tikzpicture}
\end{align}
%
%
Note that different MPS tensors need not be the same, nor do they need to have the same symmetry. As long as they have the symmetry that allows the transport of the operators in the virtual space, they can be glued together using MF. We refer to such states as MF MPS.

Additionally, this protocol works deterministically if the MPS has boundary conditions as above, i.e., where the virtual space on the edges is treated as physical. In the case of periodic boundary condition, one can only move all the operators onto one virtual leg, and they will all annihilate with a finite but still \(O(1)\) probability. This will be different from, say, measurement and post-selection which succeeds only with exponentially small probability.

We emphasize that preparing the product of the tripartite states and performing measurement in the MF basis only require a constant circuit depth, independent of the number of sites, signifying the advantage of MF over unitary circuits on current experimental platforms. Also, the concept of ``gluing" product resource states to perform certain tasks also appears in other quantum computing literature such as the teleportation quantum computation~\cite{brassard1996teleportation,gottesman1999demonstrating} and fusion-based quantum computation~\cite{bartolucci2023fusion}.

\subsection{Examples of MPS with MF Symmetry}\label{mps_example}
We present two examples of MPS tensors satisfying the MF symmetry. Both examples have a virtual dimension $D=2$ and physical $d=2$. In the first example, we impose the following symmetry constraint:
\begin{align}
    \begin{tikzpicture}[scale=0.5,baseline={([yshift=-3.65ex] current bounding box.center) }]
        \DTensor{-\singledx,0}{1.}{.6}{\small $X$}{0};
        \UTensor{0,\singledx}{1.}{.6}{\small $X$}{3};
        \ATensor{0,0}{\small $A$}{0};
\end{tikzpicture}
\; = \;
    \begin{tikzpicture}[scale=0.5,baseline={([yshift=-0.65ex] current bounding box.center) }]
        \DTensor{\singledx,0}{1.}{.6}{\small $X$}{0};
        \ATensor{0,0}{\small $A$}{0};
\end{tikzpicture} \\
\begin{tikzpicture}[scale=0.5,baseline={([yshift=-3.65ex] current bounding box.center) }]
        \DTensor{-\singledx,0}{1.}{.6}{\small $Z$}{0};
        \UTensor{0,\singledx}{1.}{.6}{\small $I$}{3};
        \ATensor{0,0}{\small $A$}{0};
\end{tikzpicture}
\; = \;
    \begin{tikzpicture}[scale=0.5,baseline={([yshift=-0.65ex] current bounding box.center) }]
        \DTensor{\singledx,0}{1.}{.6}{\small $Z$}{0};
        \ATensor{0,0}{\small $A$}{0};
\end{tikzpicture}
\end{align}
%
These are linear equation with respect to \(A\) so one can solve it directly. The result is, up to normalization,
\begin{align}
    \begin{tikzpicture}[scale=0.5,baseline={([yshift=-0.65ex] current bounding box.center) }]
        \ATensor{0,0}{\small $A$}{0};
\end{tikzpicture}
\; = \;
    \begin{tikzpicture}[scale=0.5,baseline={([yshift=-1.3ex] current bounding box.center) }]
        \draw[very thick,red] (-1,0) -- (1,0);
        \draw[very thick,black] (0,0) -- (0,1);
		\fill[color=black] (0,0) circle (.15);
\end{tikzpicture}
+ \alpha \;\;
    \begin{tikzpicture}[scale=0.5,baseline={([yshift=-1.3ex] current bounding box.center) }]
        \draw[very thick,red] (-1,0) -- (1,0);
        \draw[very thick,black] (0,.35) -- (0,1);
        \HollowDot{0,0.45};
    	\draw (0.75,0.65) node {\small $\ket{+}$};
\end{tikzpicture}
\;,
\end{align}
%
%
where the first term is the copy tensor.
Importantly, this is a continuous family satisfying the MF symmetry and includes non-stabilizer states.

We provide another example to demonstrate that the map from \(P\) to \(P'\) in Eq.~\eqref{mps_sym} need not be either identity or bijective. We demand the following symmetry constrains:
\begin{align}\label{mps_ex_nonbijective}
    \begin{tikzpicture}[scale=0.5,baseline={([yshift=-3.65ex] current bounding box.center) }]
        \DTensor{-\singledx,0}{1.}{.6}{\small $X$}{0};
        \UTensor{0,\singledx}{1.}{.6}{\small $X$}{3};
        \ATensor{0,0}{\small $A$}{0};
\end{tikzpicture}
\; = \;
    \begin{tikzpicture}[scale=0.5,baseline={([yshift=-0.65ex] current bounding box.center) }]
        \DTensor{\singledx,0}{1.}{.6}{\small $Z$}{0};
        \ATensor{0,0}{\small $A$}{0};
\end{tikzpicture} \\
\begin{tikzpicture}[scale=0.5,baseline={([yshift=-3.65ex] current bounding box.center) }]
        \DTensor{-\singledx,0}{1.}{.6}{\small $Z$}{0};
        \UTensor{0,\singledx}{1.}{.6}{\small $Z$}{3};
        \ATensor{0,0}{\small $A$}{0};
\end{tikzpicture}
\; = \;
    \begin{tikzpicture}[scale=0.5,baseline={([yshift=-0.65ex] current bounding box.center) }]
        \DTensor{\singledx,0}{1.}{.6}{\small $I$}{0};
        \ATensor{0,0}{\small $A$}{0};
\end{tikzpicture}
\end{align}
%
This set of constrains admits the following family of solutions:
\begin{align}
    \begin{tikzpicture}[scale=0.5,baseline={([yshift=-0.65ex] current bounding box.center) }]
        \ATensor{0,0}{\small $A$}{0};
\end{tikzpicture}
\; = \;
    \begin{tikzpicture}[scale=0.5,baseline={([yshift=-0.6ex] current bounding box.center) }]
        \draw[very thick,red] (-1,0) -- (2,0);
        \draw[very thick,black] (0,0) -- (0,1);
        \DTensor{1.5,0}{1.}{.6}{\small $H$}{0};
		\fill[color=black] (0,0) circle (.15);
\end{tikzpicture}
+ \alpha \;\;
    \begin{tikzpicture}[scale=0.5,baseline={([yshift=-1.3ex] current bounding box.center) }]
        \BendUpRight{0,0}{1.5}
        \draw[very thick,red] (-.75,0) -- (0,0);
        \draw[very thick,red] (1,0) -- (1.6,0);
        \HollowRedDot{1,0};
		\draw (1.2,0.65) node {\small $\ket{0}$};
\end{tikzpicture}
\; .
\end{align}

Crucially, in both examples, the symmetry constraints admit a continuous family of solution parametrized by \(\alpha\). Given that stabilzer states are a discrete set of states, most of the MF MPS cannot be stabilizer states. Therefore, we would like to understand why MF MPS can be continuously parametrized (signaling the non-stabilizerness) even though they possess symmetries similar to the stabilizer states.
%
%
%
%
%

\subsection{Structure of MPS with MF symmetry}

We give a structural theorem of all MPS preparable using MF. In particular, this reveals an underlying Clifford-like structure as well as the source of non-stabilizerness. We start by showing that MF MPS are automatically in the MPS canonical form.  The arrows throughout indicate how a tensor should be interpreted as an operator.
\begin{lemma} \label{lemma_cf}
    If $A$ satisfies the MF symmetry [Eq.~\eqref{mps_sym}], then
    \begin{align}
    \begin{tikzpicture}[scale=0.5,baseline={([yshift=-0.65ex] current bounding box.center) }]
        \draw[very thick,red] (0,0) -- (1.2,0) -- (1.2,\singledx) -- (0,\singledx);
        \GTensor{0,0}{1.2}{0.6}{\small $A$}{0};
        \GTensor{0,\singledx}{1.2}{0.6}{\small $A^\dag$}{7};
        \RedArrow{-0.8,0}{0.2}{1};
        \RedArrow{1.,0}{0.2}{1};
        \RedArrow{-1,\singledx}{0.2}{3};
        \Arrow{0,1.}{0.2}{0};
\end{tikzpicture}
\; \propto \;
    \begin{tikzpicture}[scale=0.5,baseline={([yshift=-1.65] current bounding box.center) }]
        \draw[very thick,red] (-1.2,0) -- (0.,0) -- (0.,\singledx) -- (-1.2,\singledx);
        \RedArrow{-0.8,0}{0.2}{1};
        \RedArrow{-1,\singledx}{0.2}{3};
\end{tikzpicture}
\; ,
\end{align}
\end{lemma}
i.e., $A$ is proportional to an isometry $V:\mathbb C^D \to \mathbb C^{dD}$
    \begin{align}
    \begin{tikzpicture}[scale=0.5,baseline={([yshift=-0.65ex] current bounding box.center) }]
        \ATensor{0,0}{\small $A$}{0};
        \draw[very thick,black] (0,1) -- (1,1);
\end{tikzpicture}
\; \propto \;\;
    \begin{tikzpicture}[scale=0.5,baseline={([yshift=-5.65] current bounding box.center) }]
        \Unitary{0,0}{1}{.6}{\small $V$}{1};
\end{tikzpicture}
\; .
\end{align}
\begin{proof}
The MF symmetry Eq.~\eqref{mps_sym} directly implies, by the unitarity of $P$ and $U_P$,
    \begin{align}
        \begin{tikzpicture}[scale=0.5,baseline={([yshift=-0.65ex] current bounding box.center) }]
        \draw[very thick,red] (0,0) -- (1.2,0) -- (1.2,\singledx) -- (0,\singledx);
        \GTensor{0,0}{1.2}{0.6}{\small $A$}{0};
        \begin{scope}[yscale=-1,xscale=1]
        \GTensor{0,-\singledx}{1.2}{0.6}{\small $A^\dag$}{0};
        \DTensor{-\singledx,0}{1.}{.6}{\small $P$}{0};
        \DTensor{-\singledx,-\singledx}{1.}{.6}{\small $P^\dag$}{0};        
        \end{scope}
        \RedArrow{-0.8,0}{0.2}{1};
        \RedArrow{1.,0}{0.2}{1};
        \RedArrow{-1,\singledx}{0.2}{3};
        \Arrow{0,1.}{0.2}{0};
\end{tikzpicture}
    \; = \;
        \begin{tikzpicture}[scale=0.5,baseline={([yshift=-0.65ex] current bounding box.center) }]
        \draw[very thick,red] (0,0) -- (1.2,0) -- (1.2,\singledx) -- (0,\singledx);
        \GTensor{0,0}{1.2}{0.6}{\small $A$}{0};
        \begin{scope}[yscale=-1,xscale=1]
        \GTensor{0,-\singledx}{1.2}{0.6}{\small $A^\dag$}{0};
        \end{scope}
        \RedArrow{-0.8,0}{0.2}{1};
        \RedArrow{1.,0}{0.2}{1};
        \RedArrow{-1,\singledx}{0.2}{3};
        \Arrow{0,1.}{0.2}{0};
        \end{tikzpicture}
    \end{align}
where $P$ is any element of the MF group.
%
%
The result follows using Schur's lemma (since $\mathcal B$ is irreducible) by interpreting the above equation as the adjoint action
\begin{align}
    P \left( \sum_i A^i A^{i \dagger} \right) P^\dagger = \left( \sum_i A^i A^{i \dagger} \right) \,.
\end{align}
\end{proof}

We now prove our main theorems which characterize the structure of the MF MPS tensors. There are two parts of of the theorem. The first part allows us to perform a local isometry and replace \(U_{P_i}\) with \((P_i^\dag \otimes P_i^T)\). The second part states that after the local isometry the MPS possesses a Clifford-like structure and characterizes the magic.

\subsubsection{Local Isometry}
We begin by performing a polar decomposition to the MPS. This reveals the local isometry that could allow us to replace \(U_{P_i}\) with \((P_i^\dag \otimes P_i^T)\).

\begin{theorem}\label{structural_theorem_1}
    Let \(A\) be a MPS tensor satisfying $\forall P_i \in \mathcal B$ the MF symmetry
    \begin{align}
    \begin{tikzpicture}[scale=0.5,baseline={([yshift=-0.65ex] current bounding box.center) }]
        \ATensor{0,0}{\small $A$}{0};
\end{tikzpicture}
    \; = \;\begin{tikzpicture}[scale=0.5,baseline={([yshift=-3.65ex] current bounding box.center) }]
        \DTensor{-\singledx,0}{1.}{.6}{\small $P_i$}{0};
        \UTensor{0,\singledx}{1.}{.6}{\small $U_{P_i}$}{3};
        \ATensor{0,0}{\small $A$}{0};
        \DTensor{\singledx,0}{1.}{.6}{\small $P_i'^{\dag}$}{0};
        \RedArrow{-0.8,0}{0.2}{1};
        \RedArrow{1,0}{0.2}{1};
\end{tikzpicture}
\;.
\end{align}
Performing a polar decomposition $A = V Q$ we can write
    \begin{align}\label{mps_polar_decomp}
    \begin{tikzpicture}[scale=0.5,baseline={([yshift=-0.65ex] current bounding box.center) }]
        \GTensor{0,0}{1.2}{0.6}{\small $A$}{0};
        \RedArrow{-0.8,0}{0.2}{1};
        \RedArrow{0.8,0}{0.2}{3};
        \Arrow{0,1.0}{0.2}{0};
\end{tikzpicture}
    \; = \;\begin{tikzpicture}[scale=0.5,baseline={([yshift=-3.65ex] current bounding box.center) }]
        \GFTensor{0,0}{1.2}{.6}{\small $Q$}{0};
        \UFTensor{0,\singledx}{1.2}{.6}{\small $V$}{1};
        \RedArrow{-0.8,0}{0.2}{1};
        \RedArrow{0.8+\singledx,0}{0.2}{3};
        \RedArrow{0,1.}{0.2}{0};
        \RedArrow{\singledx,1.}{0.2}{0};
        \Arrow{0.5*\singledx,1.0+\singledx}{0.2}{0};
\end{tikzpicture}
\end{align}
where \(Q: \mathbb{C}^{D^2} \rightarrow \mathbb{C}^{D^2}\) is a positive semi-definite Hermitian matrix, satisfying:
\begin{enumerate}[(i)]
    \item \(Q\) has the same null-space as \(V\), i.e.,
    \begin{align}
    \begin{tikzpicture}[scale=0.5,baseline={([yshift=-3.65ex] current bounding box.center) }]
        \GTensor{0,0}{1.2}{0.6}{\small $A$}{0};
        \UFTensor{-0.5*\singledx,\singledx}{1.2}{.6}{\small $V^\dag$}{2};
        \RedArrow{-0.8,0}{0.2}{1};
        \RedArrow{0.8,0}{0.2}{3};
        \RedArrow{-0.5*\singledx,\singledx+1.0}{0.2}{0};
        \RedArrow{0.5*\singledx,\singledx+1.0}{0.2}{0};
\end{tikzpicture}
    =\begin{tikzpicture}[scale=0.5,baseline={([yshift=-0.8ex] current bounding box.center) }]
        \GFTensor{0,0}{1.2}{.6}{\small $Q$}{0};
        \RedArrow{-0.8,0}{0.2}{1};
        \RedArrow{0.8+\singledx,0}{0.2}{3};
        \RedArrow{0,1.}{0.2}{0};
        \RedArrow{\singledx,1.}{0.2}{0};
\end{tikzpicture}
\end{align}
and $V: \mathbb{C}^{D^2} \to \mathbb C^d$ is an isometry when restricted to the range of $Q$.
    \item $[ Q , P_i \otimes  P'^*_i] = 0$, i.e., $Q$ has the symmetry
\begin{align} \label{eq:symmetry_Q}
    \begin{tikzpicture}[scale=0.5,baseline={([yshift=-0.65ex] current bounding box.center) }]
        \GFTensor{0,0}{1.2}{.6}{\small $Q$}{0};
        \RedArrow{-0.8,0}{0.2}{1};
        \RedArrow{0.8+\singledx,0}{0.2}{3};
        \RedArrow{0,1.}{0.2}{0};
        \RedArrow{\singledx,1.}{0.2}{0};
\end{tikzpicture}
    \; = \;
    \begin{tikzpicture}[scale=0.5,baseline={([yshift=-3.25ex] current bounding box.center) }]
        \DTensor{-\singledx,0}{1.}{.6}{\small $P_i$}{0};
        \DTensor{0,\singledx}{1.}{.6}{\small $P_i^\dag$}{1};
        \DTensor{\singledx,\singledx}{1.}{.6}{\small $P_i'^T$}{1};
        \GFTensor{0,0}{1.}{.6}{\small $Q$}{0};
        \DTensor{2*\singledx,0}{1.}{.6}{\small $P_i'^{*}$}{0};
        \RedArrow{-0.8,0}{0.2}{1};
        \RedArrow{0.8+\singledx,0}{0.2}{3};
        \RedArrow{0,1.}{0.2}{0};
        \RedArrow{\singledx,1.}{0.2}{0};
\end{tikzpicture}
\;.
\end{align}
\end{enumerate}
\end{theorem}

\begin{proof}
    The proof is constructive. We apply the polar decomposition to \(A\), interpreted as an operator according to the arrows:
\begin{align}
    \begin{tikzpicture}[scale=0.5,baseline={([yshift=-0.65ex] current bounding box.center) }]
        \GTensor{0,0}{1.2}{0.6}{\small $A$}{0};
        \RedArrow{-0.8,0}{0.2}{1};
        \RedArrow{0.8,0}{0.2}{3};
        \Arrow{0,1.0}{0.2}{0};
\end{tikzpicture}
    \; = \;\begin{tikzpicture}[scale=0.5,baseline={([yshift=-3.65ex] current bounding box.center) }]
        \GFTensor{0,0}{1.2}{.6}{\small $Q$}{0};
        \UFTensor{0,\singledx}{1.2}{.6}{\small $V$}{1};
        \RedArrow{-0.8,0}{0.2}{1};
        \RedArrow{0.8+\singledx,0}{0.2}{3};
        \RedArrow{0,1.}{0.2}{0};
        \RedArrow{\singledx,1.}{0.2}{0};
        \Arrow{0.5*\singledx,1.0+\singledx}{0.2}{0};
\end{tikzpicture}
\end{align}
Here \(Q\) is a positive semidefinite Hermitian matrix, and \(V\) is an isometry when restricted to the image of $Q$.
Now we evaluate the transfer matrix under the polar decomposition
\begin{align}
    \begin{tikzpicture}[scale=0.5,baseline={([yshift=-0.65ex] current bounding box.center) }]
        \GTensor{0,0}{1.2}{.6}{\small $A$}{0};
        \RedArrow{-0.8,0}{0.2}{1};
        \RedArrow{0.8,0}{0.2}{3};
        \GTensor{0,\singledx}{1.2}{.6}{\small $A^\dag$}{7};
        \RedArrow{-1,\singledx}{0.2}{3};
        \RedArrow{1,\singledx}{0.2}{1};
\end{tikzpicture}
=\begin{tikzpicture}[scale=0.5,baseline={([yshift=-0.65ex] current bounding box.center) }]
        \GFTensor{0,0}{1.2}{.6}{\small $Q$}{0};
        \RedArrow{-0.8,0}{0.2}{1};
        \RedArrow{0.8+\singledx,0}{0.2}{3};
        \RedArrow{-1,3*\singledx}{0.2}{3};
        \RedArrow{1+\singledx,3*\singledx}{0.2}{1};
        \UFTensor{0,\singledx}{1.}{.6}{\small $V$}{1};
        \UFTensor{0,2*\singledx}{1.}{.6}{\small $V^\dag$}{2};
        \GFTensor{0,3*\singledx}{1.2}{.6}{\small $Q^\dag$}{1};
\end{tikzpicture}
=\begin{tikzpicture}[scale=0.5,baseline={([yshift=-0.65ex] current bounding box.center) }]
        \GFTensor{0,0}{1.2}{.6}{\small $Q$}{0};
        \RedArrow{-0.8,0}{0.2}{1};
        \RedArrow{0.8+\singledx,0}{0.2}{3};
        \RedArrow{-1,\singledx}{0.2}{3};
        \RedArrow{1+\singledx,\singledx}{0.2}{1};
        \GFTensor{0,\singledx}{1.2}{.6}{\small $Q^\dag$}{1};
\end{tikzpicture}
=\begin{tikzpicture}[scale=0.5,baseline={([yshift=-0.65ex] current bounding box.center) }]
        \GFTensor{0,0}{1.2}{.6}{\small $Q^2$}{2};
        \RedArrow{0,-0.8}{0.2}{0};
        \RedArrow{\singledx,-0.8}{0.2}{0};
        \RedArrow{0,1}{0.2}{0};
        \RedArrow{\singledx,1}{0.2}{0};
\end{tikzpicture}
\,.
\end{align}
In the third equality we use the fact that \(Q\) and \(V\) share the same null space to cancel out \(VV^\dag\); in the fourth equality we use the Hermitian property \(Q=Q^\dag\). We now apply the symmetry constraint to the transfer matrix:
\begin{align}
\begin{tikzpicture}[scale=0.5,baseline={([yshift=-0.65ex] current bounding box.center) }]
        \GTensor{0,0}{1.2}{.6}{\small $A$}{0};
        \RedArrow{-0.8,0}{0.2}{1};
        \RedArrow{0.8,0}{0.2}{3};
        \GTensor{0,\singledx}{1.2}{.6}{\small $A^\dag$}{7};
        \RedArrow{-1,\singledx}{0.2}{3};
        \RedArrow{1,\singledx}{0.2}{1};
\end{tikzpicture}
    =\begin{tikzpicture}[scale=0.5,baseline={([yshift=-0.65ex] current bounding box.center) }]
        \GTensor{0,0}{1.2}{.6}{\small $A$}{0};
        \RedArrow{-0.8,0}{0.2}{1};
        \RedArrow{0.8,0}{0.2}{3};
        \RedArrow{-1,\singledx}{0.2}{3};
        \RedArrow{1,\singledx}{0.2}{1};
        \GTensor{0,\singledx}{1.2}{.6}{\small $A^\dag$}{7};
        \DTensor{-\singledx,0}{1.}{.6}{\small $P_i$}{0};
        \DTensor{\singledx,0}{1.}{.6}{\small $P_i'^{*}$}{0};
        \DTensor{-\singledx,\singledx}{1.}{.6}{\small $P_i^\dag$}{0};
        \DTensor{\singledx,\singledx}{1.}{.6}{\small $P_i'^{T}$}{0};
\end{tikzpicture}
=\begin{tikzpicture}[scale=0.5,baseline={([yshift=-0.65ex] current bounding box.center) }]
        \GFTensor{0,\singledx}{1.2}{.6}{\small $Q^2$}{2};
        \RedArrow{0,\singledx-0.8}{0.2}{0};
        \RedArrow{\singledx,\singledx-0.8}{0.2}{0};
        \RedArrow{0,\singledx+1}{0.2}{0};
        \RedArrow{\singledx,\singledx+1}{0.2}{0};
        \DTensor{0,0}{1.}{.6}{\small $P_i$}{1};
        \DTensor{\singledx,0}{1.}{.6}{\small $P_i'^{*}$}{1};
        \DTensor{0,2*\singledx}{1.}{.6}{\small $P_i^\dag$}{1};
        \DTensor{\singledx,2*\singledx}{1.}{.6}{\small $P_i'^{T}$}{1};
\end{tikzpicture}
\end{align}
In other words, \(Q^2\) commutes with \(P_i \otimes P_i'^*\). This implies that there exist a basis such that \(Q^2\) and \(P_i \otimes P_i'^*\) are simultaneously diagonalizable, and by definition this basis is one possible eigenbasis of \(Q^2\). Importantly, because \(Q\) is positive semidefinite, \(Q\) and \(Q^2\) share the same eigen-structures, that means for each eigenvalue, the associated (possibly degenerate) eigen-subspaces are identical. As a result, the basis we have chosen to simultaneously diagonalize \(Q^2\) and \(P_i \otimes P_i'^*\) can also diagonalize \(Q\), so \(Q\) and \(P_i \otimes P_i'^*\) commute as well.
\end{proof}

The above theorem constructs a new MPS tensor \(Q\) by applying an isometry \(V^\dag\) to the original MPS tensor such that \(Q\) has MF symmetry but with \(P_i^\dag \otimes P_i'^T\) as the error-correcting operators. A natural question then is how the original error-correcting operators \(U_{P_i}\) are related to \(P_i^\dag \otimes P_i'^T\). We show that they agree in the range of \(Q\).
\begin{corollary}\label{Up_PP_relation}
    Let \(R=V^\dag V\) be the projector onto the domain of \(V\) (equivalently the range of \(Q\)). Then
    \begin{align}
        V^\dag U_{P_i} V = (P_i^\dag \otimes P_i^T) R \;.
    \end{align}
    Pictorially,
    \begin{align}
    \begin{tikzpicture}[scale=0.5,baseline={([yshift=-0.25ex] current bounding box.center) }]
        \UFTensor{0,0}{1.}{.6}{\small $V$}{1};
        \UTensor{0.5*\singledx,\singledx}{1.}{.6}{\small $U_{P_i}$}{3};
        \UFTensor{0,2*\singledx}{1.}{.6}{\small $V^\dag$}{2};
        \end{tikzpicture}
        \;= \;
        \begin{tikzpicture}[scale=0.5,baseline={([yshift=-3.25ex] current bounding box.center) }]
        \DTensor{0,0}{1.}{.6}{\small $P_i^\dag$}{1};
        \DTensor{\singledx,0}{1.}{.6}{\small $P_i'^T$}{1};
        \UFTensor{0,-\singledx}{1.}{.6}{\small $R$}{3};
        \end{tikzpicture}
        \;.
    \end{align}
\end{corollary}

\begin{proof}
We write the MPS tensor in the form of Eq.~(\ref{mps_polar_decomp}) and then apply the symmetry
\begin{align}
    \begin{tikzpicture}[scale=0.5,baseline={([yshift=-0.65ex] current bounding box.center) }]
        \GFTensor{0,0}{1.2}{.6}{\small $Q$}{0};
        \UFTensor{0,\singledx}{1.}{.6}{\small $V$}{1};
        \RedArrow{-0.8,0}{0.2}{1};
        \RedArrow{0.8+\singledx,0}{0.2}{3};
        \RedArrow{0,1.}{0.2}{0};
        \RedArrow{\singledx,1.}{0.2}{0};
\end{tikzpicture}
=\begin{tikzpicture}[scale=0.5,baseline={([yshift=-3.65ex] current bounding box.center) }]
        \GFTensor{0,0}{1.2}{.6}{\small $Q$}{0};
        \UFTensor{0,\singledx}{1.}{.6}{\small $V$}{1};
        \RedArrow{-0.8,0}{0.2}{1};
        \RedArrow{0.8+\singledx,0}{0.2}{3};
        \RedArrow{0,1.}{0.2}{0};
        \RedArrow{\singledx,1.}{0.2}{0};
        \DTensor{-\singledx,0}{1.}{.6}{\small $P_i$}{0};
        \UTensor{0.5*\singledx,2*\singledx}{1.}{.6}{\small $U_{P_i}$}{3};
        \DTensor{2*\singledx,0}{1.}{.6}{\small $P_i'^{*}$}{0};
\end{tikzpicture}
\;.
\end{align}
We will read the above equation from bottom to top so that
\begin{align}
    V Q = U_{P_i} V Q (P_i \otimes P_i'^{*})
\end{align}
We now apply the pseudo-inverse of \(Q\) from the right
\begin{align}
    V Q Q^{-1} &= U_{P_i} V Q (P_i \otimes P_i'^{*}) Q^{-1} \\
    V Q Q^{-1} &= U_{P_i} V Q Q^{-1} (P_i \otimes P_i'^{*}) \\
    V R &= U_{P_i} V R (P_i \otimes P_i'^{*})
\end{align}
Where in the second line, we use the fact that \(Q\) and \((P_i \otimes P_i'^{*})\) commute; in the third line we realize that \(Q Q^{-1}\) is the orthogonal projector to the domain of \(Q\) which by definition is \(R\). Lastly, we multiply \(V^\dag\) from the left and \((P_i^\dag \otimes P_i'^{T})\) from the right.
\begin{align}
    V^\dag V R (P_i^\dag \otimes P_i'^{T}) &= V^\dag U_{P_i} V R \\
    (P_i^\dag \otimes P_i'^{T}) R &= V^\dag U_{P_i} V
\end{align}
Where in the second line, \(V^\dag V=R\) since \(V\) and \(Q\) share the same null space, and \(V R=V V^\dag V=V\). Finally, \(R\) and \((P_i^\dag \otimes P_i'^{T})\) commute so we can swap their order.
\end{proof}

In the case where the tensor \(A\) is injective, the projector \(R\) is trivial, and thus one can immediately conclude that \(U_{P_i}\) and \((P^\dag \otimes P^T)\) are isometrically equivalent to each other. On the other hand, the non-injective case is richer, since the \(V^\dag U_{P_i} V \) and \((P^\dag \otimes P^T)\) only need to agree on the range of \(R\). As an example, we apply Corollary~\ref{Up_PP_relation} to the AKLT state where $D=2$ and $d=3$, hence the tensor is non-injective. The correction operators \(U_{P_i}\) are:
\begin{align}
    U_Y &= \ket{+1}\bra{-1} - \ket{0}\bra{0} + \ket{-1}\bra{+1}\\
    U_X &= \ket{+1}\bra{-1} + \ket{0}\bra{0} + \ket{-1}\bra{+1}\\
    U_Z &= \ket{+1}\bra{+1} - \ket{0}\bra{0} + \ket{-1}\bra{-1}\\
    U_I &= \ket{+1}\bra{+1} + \ket{0}\bra{0} + \ket{-1}\bra{-1}
\end{align}
The AKLT state can also be represented as triplet projectors acting on the maximally entangled state~\cite{affleck2004rigorous}, so \(V\) has the following form:
\begin{align}
    V = \ket{+1}\bra{11} + \ket{0}\bra{T} + \ket{-1}\bra{00}
\end{align}
where we denote the triplet state \(\ket{T}=\frac{1}{\sqrt{2}}(\ket{01}+\ket{10})\) and the singlet state \(\ket{S}=\frac{1}{\sqrt{2}}(\ket{01}-\ket{10})\). We explicitly write down
\begin{align}
    V^\dag U_{Z} V &= \ket{00}\bra{00} + \ket{11}\bra{11} - \ket{T}\bra{T}\\
    Z \otimes Z & = \ket{00}\bra{00} + \ket{11}\bra{11}  - \ket{T}\bra{T} - \ket{S}\bra{S}\\
    V^\dag U_{X} V &= \ket{00}\bra{11} + \ket{11}\bra{00} + \ket{T}\bra{T}\\
    X \otimes X &= \ket{00}\bra{11} + \ket{11}\bra{00} + \ket{T}\bra{T} - \ket{S}\bra{S}
\end{align}
from which it becomes clear that \(V^\dag U_{Z} V\) and \(Z \otimes Z\), \(V^\dag U_{X} V\) and \(X \otimes X\) indeed agree in the domain of \(V\), i.e., \(\mathrm{Span} \{\ket{00}, \ket{T}, \ket{11}\}\), but disagree on \(\ket{S}\). Lastly, we note that similar results are reported in Ref.~\cite{smith2024constant}.

\subsubsection{Clifford Structure}

We now show that when the MF group is the Weyl-Heisenberg group, then \(Q\) possesses a Clifford-like structure. We begin with the following lemma.
\begin{lemma}\label{lemma_clifford_purification}
    Suppose an isometry \(V: \mathcal{C}^D \rightarrow \mathcal{C}^{nD}\) has the following symmetry which maps the \(D\)-dimensional Weyl-Heisenberg operators \(P_i\) to a multiqudit Weyl-Heisenberg operator \((P_{i1} \otimes P_{i2} \otimes ... \otimes P_{in})\):
    \begin{align}
        V = P_{i}^\dag V (P_{i1} \otimes P_{i2} \otimes ... \otimes P_{in}), \quad \forall P_i \,.
  \end{align}
  Then there exist a Clifford unitary \(U: \mathcal{C}^{nD} \rightarrow \mathcal{C}^{nD}\) that maps \(I \otimes I \otimes ... \otimes P_i\) to \((P_{i1} \otimes P_{i2} \otimes ... \otimes P_{in})\).
\end{lemma}
We defer the proof to Appendix \ref{existance_clifford_purification}. The above lemma states that for any isometry with Clifford-like symmetries, one can always find a Clifford unitary with the same symmetries. With that, we are ready to unveil the Clifford structure of the MF MPS.
\begin{theorem}\label{structural_theorem_2}
    In the setting of \cref{structural_theorem_1}, when \(\{P_i\}\) is the \(D\)-dimensional Weyl-Heisenberg group, \(Q\) can be written in the following Clifford form when interpreted sideways:
\begin{align}\label{clifford_magic}
    \begin{tikzpicture}[scale=0.5,baseline={([yshift=-0.65ex] current bounding box.center) }]
        \draw[very thick, draw=red] (0,0) -- (0,2.4);
        \draw[very thick, draw=red] (\singledx,0) -- (\singledx,1.2);
        \draw[very thick, draw=red] (0,2.4) -- (\singledx+1.2,2.4);
        \draw[very thick, draw=red] (\singledx,1.2) -- (\singledx+1.2,1.2);
        \GFTensor{0,0}{1.2}{.6}{\small $Q$}{0};
        \RedArrow{-0.8,0}{0.2}{1};
        \RedArrow{1.+\singledx,0}{0.2}{1};
        \RedArrow{0,1.}{0.2}{0};
        \RedArrow{\singledx,1.}{0.2}{0};
\end{tikzpicture}
    \; \propto \;
    \begin{tikzpicture}[scale=0.5,baseline={([yshift=-0.65ex] current bounding box.center) }]
        \FUnitary{0,0}{1.2}{.6}{\small $U_C$}{0};
        \draw[ thick, rounded corners=2pt] (-1.2-1.2,1.2-0.6) rectangle (-1.2,3.6-0.6);
        \draw (-1.2-0.6,2.4-0.6) node {\small $\ket{\psi}$};
        \RedArrow{-0.8,0}{0.2}{1};
        \RedArrow{1,0}{0.2}{1};
        \RedArrow{1,1.2}{0.2}{1};
        \RedArrow{1,2.4}{0.2}{1};
\end{tikzpicture}
\;.
\end{align}
Here \(U_C: \mathbb{C}^{D^3} \rightarrow \mathbb{C}^{D^3}\) is a Clifford unitary satisfying the same symmetry constrains as $Q$ in Eq.~\eqref{eq:symmetry_Q} and \(\ket{\psi}\) is a (possibly magic) pure state.
\end{theorem}

\begin{proof}
We begin with showing that \(Q\) is also in the MPS canonical form
\begin{align}
\begin{tikzpicture}[scale=0.5,baseline={([yshift=-0.65ex] current bounding box.center) }]
        \draw[very thick,red] (0,0) -- (1.2,0) -- (1.2,\singledx) -- (0,\singledx);
        \GTensor{0,0}{1.2}{0.6}{\small $A$}{0};
        \begin{scope}[yscale=-1,xscale=1]
        \GTensor{0,-\singledx}{1.2}{0.6}{\small $A^\dag$}{0};
        \end{scope}
         \RedArrow{-0.8,0}{0.2}{1};
        \RedArrow{-1,\singledx}{0.2}{3};
\end{tikzpicture}
=\begin{tikzpicture}[scale=0.5,baseline={([yshift=-0.65ex] current bounding box.center) }]
        \GFTensor{0,0}{1.2}{.6}{\small $Q$}{0};
        \RedArrow{-0.8,0}{0.2}{1};
        \RedArrow{1.+\singledx,0}{0.2}{1};
        \RedArrow{-1,3*\singledx}{0.2}{3};
        \RedArrow{0,1.}{0.2}{0};
        \RedArrow{\singledx,1.}{0.2}{0};
        \UFTensor{0,\singledx}{1.}{.6}{\small $V$}{1};
        \UFTensor{0,2*\singledx}{1.}{.6}{\small $V^\dag$}{2};
        \GFTensor{0,3*\singledx}{1.2}{.6}{\small $Q^\dag$}{1};
        \draw[very thick, draw=red](\singledx+1.2,0) -- (\singledx+1.2,3*\singledx);
\end{tikzpicture}
=\begin{tikzpicture}[scale=0.5,baseline={([yshift=-0.65ex] current bounding box.center) }]
        \GFTensor{0,0}{1.2}{.6}{\small $Q$}{0};
        \RedArrow{-0.8,0}{0.2}{1};
        \RedArrow{1.+\singledx,0}{0.2}{1};
        \RedArrow{-1,\singledx}{0.2}{3};
        \RedArrow{0,1.}{0.2}{0};
        \RedArrow{\singledx,1.}{0.2}{0};
        \GFTensor{0,\singledx}{1.2}{.6}{\small $Q^\dag$}{1};
        \draw[very thick, draw=red](\singledx+1.2,0) -- (\singledx+1.2,\singledx);
\end{tikzpicture}
\propto\begin{tikzpicture}[scale=0.5,baseline={([yshift=-0.65ex] current bounding box.center) }]
        \draw[very thick, draw=red](0,0) -- (0,\singledx);
        \draw[very thick, draw=red](-1.2,0) -- (0,0);
        \draw[very thick, draw=red](-1.2,\singledx) -- (0,\singledx);
        \RedArrow{-0.8,0}{0.2}{1};
        \RedArrow{-1,\singledx}{0.2}{3};
\end{tikzpicture}
\end{align}
We have again used the fact that \(Q\) and \(V\) share the same null space and the last equality follows from \cref{lemma_cf}. As a result, \(Q\) is also an isometry when interpreted sideways
\begin{align} \label{eq:def_V_Q}
    \begin{tikzpicture}[scale=0.5,baseline={([yshift=-0.65ex] current bounding box.center) }]
        \draw[very thick, draw=red] (0,0) -- (0,2.4);
        \draw[very thick, draw=red] (\singledx,0) -- (\singledx,1.2);
        \draw[very thick, draw=red] (0,2.4) -- (\singledx+1.2,2.4);
        \draw[very thick, draw=red] (\singledx,1.2) -- (\singledx+1.2,1.2);
        \GFTensor{0,0}{1.2}{.6}{\small $Q$}{0};
        \RedArrow{-0.8,0}{0.2}{1};
        \RedArrow{1.+\singledx,0}{0.2}{1};
        \RedArrow{0,1.}{0.2}{0};
        \RedArrow{\singledx,1.}{0.2}{0};
\end{tikzpicture}
    \; \propto \;
    \begin{tikzpicture}[scale=0.5,baseline={([yshift=-0.65ex] current bounding box.center) }]
        \FUnitary{0,0}{1.2}{.6}{\small $V_Q$}{1};
        \RedArrow{-0.8,0}{0.2}{1};
        \RedArrow{1,0}{0.2}{1};
        \RedArrow{1,1.2}{0.2}{1};
        \RedArrow{1,2.4}{0.2}{1};
\end{tikzpicture}
\end{align}
Because of Eqs. (\ref{eq:symmetry_Q}), \(V_Q\) has symmetries mapping \(P_i\) to \((P_i^\dag \otimes P_i'^T \otimes P_i'^\dagger)\).
\begin{align}
    \begin{tikzpicture}[scale=0.5,baseline={([yshift=-0.65ex] current bounding box.center) }]
        \FUnitary{0,0}{1.2}{.6}{\small $V_Q$}{1};
\end{tikzpicture}
=\begin{tikzpicture}[scale=0.5,baseline={([yshift=-0.65ex] current bounding box.center) }]
        \FUnitary{0,0}{1.2}{.6}{\small $V_Q$}{1};
        \DTensor{-\singledx,0}{1.2}{.6}{\small $P_i^\dag$}{0};
        \DTensor{\singledx,0}{1.2}{.5}{\small $P_i'^\dag$}{0};
        \DTensor{\singledx,1.2}{1.2}{.5}{\small $P_i'^T$}{0};
        \DTensor{\singledx,2.4}{1.2}{.5}{\small $P_i^\dag$}{0};
\end{tikzpicture}
\end{align}
Note that the Weyl-Heisenberg group is closed under taking transpose or complex conjugation. Now we use Lemma \ref{lemma_clifford_purification} to find a Clifford unitary \(U_C\) that also maps under adjoint action \(I \otimes I \otimes P_i\) to \(P_i^\dag \otimes P_i'^T \otimes P'^\dagger\). Contracting \(U_C\) with \(V_Q\).
\begin{align}
    \begin{tikzpicture}[scale=0.5,baseline={([yshift=-0.65ex] current bounding box.center) }]
        \FUnitary{0,0}{1.2}{.6}{\small $V_Q$}{1};
        \FUnitary{\singledx,0}{1.2}{.6}{\small $U_C^\dag$}{0};
        \RedArrow{-0.8,0}{0.2}{1};
        \RedArrow{1+\singledx,0}{0.2}{1};
        \RedArrow{1+\singledx,1.2}{0.2}{1};
        \RedArrow{1+\singledx,2.4}{0.2}{1};
\end{tikzpicture}
\end{align}
We can again apply the symmetries of \(V_Q\) and \(U_C\) to show that any \(P_i\) can be pushed through the virtual leg:
\begin{align}
    \begin{tikzpicture}[scale=0.5,baseline={([yshift=-0.65ex] current bounding box.center) }]
        \DTensor{-\singledx,0}{1.2}{0.6}{\small $P_i$}{0}
        \FUnitary{0,0}{1.2}{.6}{\small $V_Q$}{1};
        \FUnitary{\singledx,0}{1.2}{.6}{\small $U_C^\dag$}{0};
        \RedArrow{-0.8,0}{0.2}{1};
        \RedArrow{1+\singledx,0}{0.2}{1};
        \RedArrow{1+\singledx,1.2}{0.2}{1};
        \RedArrow{1+\singledx,2.4}{0.2}{1};
\end{tikzpicture}
=\begin{tikzpicture}[scale=0.5,baseline={([yshift=-0.65ex] current bounding box.center) }]
        \DTensor{2*\singledx,0}{1.2}{0.6}{\small $P_i$}{0}
        \FUnitary{0,0}{1.2}{.6}{\small $V_Q$}{1};
        \FUnitary{\singledx,0}{1.2}{.6}{\small $U_C^\dag$}{0};
        \RedArrow{-0.8,0}{0.2}{1};
        \RedArrow{1+\singledx,0}{0.2}{1};
        \RedArrow{1+\singledx,1.2}{0.2}{1};
        \RedArrow{1+\singledx,2.4}{0.2}{1};
\end{tikzpicture}
\end{align}
Then by Schur's lemma, we can conclude that the contraction results in an identify on the virtual leg.
\begin{align}
    \begin{tikzpicture}[scale=0.5,baseline={([yshift=-0.65ex] current bounding box.center) }]
        \FUnitary{0,0}{1.2}{.6}{\small $V_Q$}{1};
        \FUnitary{\singledx,0}{1.2}{.6}{\small $U_C^\dag$}{0};
        \RedArrow{-0.8,0}{0.2}{1};
        \RedArrow{1+\singledx,0}{0.2}{1};
        \RedArrow{1+\singledx,1.2}{0.2}{1};
        \RedArrow{1+\singledx,2.4}{0.2}{1};
\end{tikzpicture}
=\begin{tikzpicture}[scale=0.5,baseline={([yshift=-0.65ex] current bounding box.center) }]
        \draw[very thick,draw=red] (-1.2,0)--(1.2,0);
        \draw[very thick,draw=red] (0,1.2)--(1.2,1.2);
        \draw[very thick,draw=red] (0,2.4)--(1.2,2.4);
        \draw[ thick, fill=white, rounded corners=2pt] (-0.6,1.2-0.6) rectangle (0.6,3.6-0.6);
        \draw (0,2.4-0.6) node {\small $\ket{\psi}$};
        \RedArrow{-0.8,0}{0.2}{1};
        \RedArrow{1,0}{0.2}{1};
        \RedArrow{1,1.2}{0.2}{1};
        \RedArrow{1,2.4}{0.2}{1};
\end{tikzpicture}
\end{align}
The final equation is equivalent to \cref{clifford_magic}.
\end{proof}

This result has several interesting consequences. First, it shows that, when $\mathcal B$ the Weyl-Heisenberg group, the mechanism that allows gluing together the different tensors is a Clifford property. In that sense, all MF protocols closely resemble the parallelization of stabilizer state preparation via gate teleportation~\cite{gottesman1999demonstrating}. Nevertheless, as we have seen the resulting MF MPS are not necessarily stabilizer states. According to \cref{structural_theorem_1} and \cref{structural_theorem_2} this can happen due to (a) $\ket{\psi}$ being nonstabilizer, and (b) $V^\dagger$ not having the Clifford symmetry.

Importantly, (a) and (b) are, in general, separate sources of magic.
To see that, consider applying a rotation \(\tilde{U}\) to \(\ket{\psi}\) in Eq.~(\ref{clifford_magic}). We push \(\tilde{U}\) through \(U_C\). For a generic rotation, this results in a global unitary \(\tilde{U}'\) supported on both the physical legs (top two legs) and the virtual leg.
\begin{align}
    \begin{tikzpicture}[scale=0.5,baseline={([yshift=-0.65ex] current bounding box.center) }]
        \FUnitary{0,0}{1.2}{.6}{\small $U_C$}{0};
        \Unitary{-\singledx,1.2}{1.2}{.6}{\small $\tilde{U}$}{2};
        \draw[ thick, rounded corners=2pt] (-1.2-1.2-\singledx,1.2-0.6) rectangle (-1.2-\singledx,3.6-0.6);
        \draw (-1.2-0.6-\singledx,2.4-0.6) node {\small $\ket{\psi}$};
        \RedArrow{-0.8,0}{0.2}{1};
        \RedArrow{1,0}{0.2}{1};
        \RedArrow{1,1.2}{0.2}{1};
        \RedArrow{1,2.4}{0.2}{1};
\end{tikzpicture}
=\begin{tikzpicture}[scale=0.5,baseline={([yshift=-0.65ex] current bounding box.center) }]
        \FUnitary{0,0}{1.2}{.6}{\small $U_C$}{0};
        \FUnitary{\singledx,0}{1.2}{.6}{\small $\tilde{U}'$}{0};
        \draw[ thick, rounded corners=2pt] (-1.2-1.2,1.2-0.6) rectangle (-1.2,3.6-0.6);
        \draw (-1.2-0.6,2.4-0.6) node {\small $\ket{\psi}$};
        \RedArrow{-0.8,0}{0.2}{1};
        \RedArrow{1+\singledx,0}{0.2}{1};
        \RedArrow{1+\singledx,1.2}{0.2}{1};
        \RedArrow{1+\singledx,2.4}{0.2}{1};
\end{tikzpicture}
\end{align}
On the other hand, a local rotation results in a unitary only on the physical legs, so the magic originating from $\ket{\psi}$ cannot be represented as a local rotation.

We provide an example to illustrate the magic generated by (a). We consider the case where \( (P,P') \in \{ (X,X), (Z,Z)\}\). In this case, \(V_Q\) [Eq.~\eqref{eq:def_V_Q}] maps via adjoint action \(X \mapsto X^{\otimes 3}\) and \(Z \mapsto Z^{\otimes 3}\), so we can take the following \(U_C\):
\begin{align}
\begin{tikzpicture}[scale=0.5,baseline={([yshift=-0.65ex] current bounding box.center) }]
        \FUnitary{0,0}{1.2}{.6}{\small $U_C$}{0};
        \draw[ thick, rounded corners=2pt] (-1.2-1.2,1.2-0.6) rectangle (-1.2,3.6-0.6);
        \draw (-1.2-0.6,2.4-0.6) node {\small $\ket{\psi}$};
        \RedArrow{-0.8,0}{0.2}{1};
        \RedArrow{1,0}{0.2}{1};
        \RedArrow{1,1.2}{0.2}{1};
        \RedArrow{1,2.4}{0.2}{1};
\end{tikzpicture}
\; =
    \begin{tikzpicture}[scale=0.4,baseline={([yshift=-0.25ex] current bounding box.center) }]
        \draw[thick] (-1,0) -- (0,0) -- (0, 2*\singledx) -- (6*\singledx+1.2, 2*\singledx);
        \draw[thick] (\singledx, \singledx) -- (6*\singledx+1.2, \singledx);
        \draw[thick] (\singledx, 0) -- (6*\singledx+1.2, 0);
        \draw[thick, fill=white, rounded corners=2pt] (\singledx-0.6,-0.6) rectangle (\singledx+0.6,\singledx+0.6);
        \draw (\singledx,0.5*\singledx) node {\small $\ket{\psi}$};
        \CNOT{2*\singledx,2*\singledx}{\singledx};
        \Gate{3*\singledx,2*\singledx}{\small $H$};
        \CNOT{4*\singledx,2*\singledx}{2*\singledx};
        \Gate{5*\singledx,2*\singledx}{\small $H$};
        \Gate{5*\singledx,0*\singledx}{\small $H$};
        \CNOT{6*\singledx,1*\singledx}{\singledx};
        \end{tikzpicture}
\end{align}
One can verify that the above Clifford circuit indeed maps \(X\) and \(Z\) as stated before, regardless of $\ket{\psi}$. When \(\ket{\psi}=\ket{+}\ket{+}\), it generates the cluster state (up to a local unitary).
\begin{align}
    \begin{tikzpicture}[scale=0.4,baseline={([yshift=-0.25ex] current bounding box.center) }]
        \draw[thick] (-1,0) -- (0,0) -- (0, 2*\singledx) -- (6*\singledx+1.2, 2*\singledx);
        \draw[thick] (\singledx, \singledx) -- (6*\singledx+1.2, \singledx);
        \draw[thick] (\singledx, 0) -- (6*\singledx+1.2, 0);
        \HollowDot{\singledx,1*\singledx};
        \HollowDot{\singledx,0*\singledx};
        \draw (\singledx-1,1*\singledx) node {\small $\ket{+}$};
        \draw (\singledx-1,0*\singledx) node {\small $\ket{+}$};
        \CNOT{2*\singledx,2*\singledx}{\singledx};
        \Gate{3*\singledx,2*\singledx}{\small $H$};
        \CNOT{4*\singledx,2*\singledx}{2*\singledx};
        \Gate{5*\singledx,2*\singledx}{\small $H$};
        \Gate{5*\singledx,0*\singledx}{\small $H$};
        \CNOT{6*\singledx,1*\singledx}{\singledx};
        \end{tikzpicture}
        =\begin{tikzpicture}[scale=0.3,baseline={([yshift=-0.25ex] current bounding box.center) }]
        \draw[thick] (-1,0) -- (0,0) -- (0, 2*\singledx) -- (\singledx+1.2, 2*\singledx);
        \draw[thick] (\singledx, \singledx) -- (\singledx+1.2, \singledx);
        \draw[thick] (\singledx, 0) -- (\singledx+1.2, 0);
        \draw[thick] (\singledx, 0) -- (\singledx, \singledx);
        \end{tikzpicture}
\end{align}
On the other hand, one can choose an arbitrary \(\ket{\psi}\) and the resulting MPS still preserves the MF symmetry. 
\vtwo{For example, for AKLT state, $\ket{\psi}$ is the following non-stabilizer state.
\begin{align}
    \ket{\psi} &= \frac{1}{4\sqrt{3}}(3\ket{+}\ket{+} - \ket{-}\ket{+} -\ket{+}\ket{-} - \ket{-}\ket{-})\\
    &= \frac{1}{\sqrt{3}}(\ket{01}+\ket{10}+\ket{11})
\end{align}
This example illustrates how the tensor \(Q\) can contain magic even though it possesses a Clifford-like symmetry.
}

\vtwo{In addition, we point out a property of the MF MPS regarding the calculation of expectation values. After applying the local isometry $V^\dagger$, the evaluation of the Weyl-Heisenberg observable is simple even though the state could possess high magic. More precisely, consider the MPS after applying $V^\dagger$ to each tensor; recall that \(V^\dag A=Q\), and the resulting state can be sequentially generated using Clifford unitaries and product of magic initial states \cite{lami2024quantum}.}

\begin{align}
&\begin{tikzpicture}[scale=0.5,baseline={([yshift=-3.65ex] current bounding box.center) }]
        \draw[very thick, draw=red] (-1.2,0) -- (4*\singledx+1.2,0);
        \GTensor{0,0}{1.2}{0.6}{\small $A$}{0};
        \UFTensor{-0.5*\singledx,\singledx}{1.2}{.6}{\small $V^\dag$}{2};
        \GTensor{2*\singledx,0}{1.2}{0.6}{\small $A$}{0};
        \UFTensor{1.5*\singledx,\singledx}{1.2}{.6}{\small $V^\dag$}{2};
        \GTensor{4*\singledx,0}{1.2}{0.6}{\small $A$}{0};
        \UFTensor{3.5*\singledx,\singledx}{1.2}{.6}{\small $V^\dag$}{2};
\end{tikzpicture}\\
    = & \; \begin{tikzpicture}[scale=0.5,baseline={([yshift=-0.65ex] current bounding box.center) }]
        \GFTensor{0,0}{1.2}{.6}{\small $Q$}{0};
        \GFTensor{2*\singledx,0}{1.2}{.6}{\small $Q$}{0};
        \GFTensor{4*\singledx,0}{1.2}{.6}{\small $Q$}{0};
\end{tikzpicture}\\
    \; \propto & \;
    \begin{tikzpicture}[scale=0.4,baseline={([yshift=-0.65ex] current bounding box.center) }]
    \draw[very thick, draw=red] (-2.4,0) -- (6*\singledx+2.4,0);
    \begin{scope}[shift={(0,0)}]
        \FUnitary{0,0}{1.2}{.6}{\small $U_C$}{0};
        \draw[ thick, rounded corners=2pt] (-1.2-1.2,1.2-0.6) rectangle (-1.2,3.6-0.6);
        \draw (-1.2-0.6,2.4-0.6) node {\small $\ket{\psi}$};
        \draw[very thick, draw=red](1.2,2.4) -- (1.2,3.6);
        \draw[very thick, draw=red](1.2,1.2) -- (1.2+1.2,1.2) -- (1.2+1.2,3.6);
    \end{scope}
    \begin{scope}[shift={(3*\singledx,0)}]
        \FUnitary{0,0}{1.2}{.6}{\small $U_C$}{0};
        \draw[ thick, rounded corners=2pt] (-1.2-1.2,1.2-0.6) rectangle (-1.2,3.6-0.6);
        \draw (-1.2-0.6,2.4-0.6) node {\small $\ket{\psi}$};
        \draw[very thick, draw=red](1.2,2.4) -- (1.2,3.6);
        \draw[very thick, draw=red](1.2,1.2) -- (1.2+1.2,1.2) -- (1.2+1.2,3.6);
    \end{scope}
    \begin{scope}[shift={(6*\singledx,0)}]
        \FUnitary{0,0}{1.2}{.6}{\small $U_C$}{0};
        \draw[ thick, rounded corners=2pt] (-1.2-1.2,1.2-0.6) rectangle (-1.2,3.6-0.6);
        \draw (-1.2-0.6,2.4-0.6) node {\small $\ket{\psi}$};
        \draw[very thick, draw=red](1.2,2.4) -- (1.2,3.6);
        \draw[very thick, draw=red](1.2,1.2) -- (1.2+1.2,1.2) -- (1.2+1.2,3.6);
    \end{scope}
\end{tikzpicture}
\end{align}

\vtwo{After removing the inverse isometry, the sequential structure renders the computation of Weyl-Heisenberg expectation values efficient despite the fact that the resulting state potentially has high magic. This is because one can push Pauli strings through the sequential Clifford unitaries as shown below.}

\begin{align}\label{mps_op_ex1}
    &\begin{tikzpicture}[scale=0.4,baseline={([yshift=-0.65ex] current bounding box.center) }]
    \draw[very thick, draw=red] (-2.4,0) -- (6*\singledx+2.4,0);
    \begin{scope}[shift={(0,0)}]
        \FUnitary{0,0}{1.2}{.6}{\small $U_C$}{0};
        \draw[ thick, rounded corners=2pt] (-1.2-1.2,1.2-0.6) rectangle (-1.2,3.6-0.6);
        \draw (-1.2-0.6,2.4-0.6) node {\small $\ket{\psi}$};
        \draw[very thick, draw=red](1.2,2.4) -- (1.2,3.6);
        \draw[very thick, draw=red](1.2,1.2) -- (1.2+1.2,1.2) -- (1.2+1.2,3.6);
        \DTensor{1.2,4}{1}{.5}{}{1};
        \DTensor{2.4,4}{1}{.5}{}{1};
        \draw [very thick,->] (3.5,4.5) -- (3.5,3.5);
    \end{scope}
    \begin{scope}[shift={(3*\singledx,0)}]
        \FUnitary{0,0}{1.2}{.6}{\small $U_C$}{0};
        \draw[ thick, rounded corners=2pt] (-1.2-1.2,1.2-0.6) rectangle (-1.2,3.6-0.6);
        \draw (-1.2-0.6,2.4-0.6) node {\small $\ket{\psi}$};
        \draw[very thick, draw=red](1.2,2.4) -- (1.2,3.6);
        \draw[very thick, draw=red](1.2,1.2) -- (1.2+1.2,1.2) -- (1.2+1.2,3.6);
        \DTensor{1.2,4}{1}{.5}{}{1};
        \DTensor{2.4,4}{1}{.5}{}{1};
        \draw [very thick,->] (3.5,4.5) -- (3.5,3.5);
    \end{scope}
    \begin{scope}[shift={(6*\singledx,0)}]
        \FUnitary{0,0}{1.2}{.6}{\small $U_C$}{0};
        \draw[ thick, rounded corners=2pt] (-1.2-1.2,1.2-0.6) rectangle (-1.2,3.6-0.6);
        \draw (-1.2-0.6,2.4-0.6) node {\small $\ket{\psi}$};
        \draw[very thick, draw=red](1.2,2.4) -- (1.2,3.6);
        \draw[very thick, draw=red](1.2,1.2) -- (1.2+1.2,1.2) -- (1.2+1.2,3.6);
        \DTensor{1.2,4}{1}{.5}{}{1};
        \DTensor{2.4,4}{1}{.5}{}{1};
        \draw [very thick,->] (3.5,4.5) -- (3.5,3.5);
    \end{scope}
\end{tikzpicture}\\ \label{mps_op_ex2}
&=\begin{tikzpicture}[scale=0.4,baseline={([yshift=-0.65ex] current bounding box.center) }]
    \draw[very thick, draw=red] (-3.4,0) -- (6*\singledx+4.4,0);
    \begin{scope}[shift={(0,0)}]
        \FUnitary{0,0}{1.2}{.6}{\small $U_C$}{0};
        \draw[ thick, rounded corners=2pt] (-1.2-1.2-1.,1.2-0.6) rectangle (-1.2-1.,3.6-0.6);
        \draw (-1.2-0.6-1.,2.4-0.6) node {\small $\ket{\psi}$};
        \draw[very thick, draw=red](1.2,2.4) -- (1.2,3.6);
        \draw[very thick, draw=red](1.2,1.2) -- (1.2+1.2,1.2) -- (1.2+1.2,3.6);
        \DTensor{-1.2-0.2,0}{0.8}{.5}{}{0};
        \DTensor{-1.2-0.2,1.2}{0.8}{.5}{}{0};
        \DTensor{-1.2-0.2,2.4}{0.8}{.5}{}{0};
    \end{scope}
    \begin{scope}[shift={(3*\singledx+1,0)}]
        \FUnitary{0,0}{1.2}{.6}{\small $U_C$}{0};
        \draw[ thick, rounded corners=2pt] (-1.2-1.2-1.,1.2-0.6) rectangle (-1.2-1.,3.6-0.6);
        \draw (-1.2-0.6-1.,2.4-0.6) node {\small $\ket{\psi}$};
        \draw[very thick, draw=red](1.2,2.4) -- (1.2,3.6);
        \draw[very thick, draw=red](1.2,1.2) -- (1.2+1.2,1.2) -- (1.2+1.2,3.6);
        \DTensor{-1.2-0.2,0}{0.8}{.5}{}{0};
        \DTensor{-1.2-0.2,1.2}{0.8}{.5}{}{0};
        \DTensor{-1.2-0.2,2.4}{0.8}{.5}{}{0};
        \draw [very thick,->] (-1.2-0.2+0.5,-1) -- (-1.2-0.2-0.5,-1);
    \end{scope}
    \begin{scope}[shift={(6*\singledx+2,0)}]
        \FUnitary{0,0}{1.2}{.6}{\small $U_C$}{0};
        \draw[ thick, rounded corners=2pt] (-1.2-1.2-1.,1.2-0.6) rectangle (-1.2-1.,3.6-0.6);
        \draw (-1.2-0.6-1.,2.4-0.6) node {\small $\ket{\psi}$};
        \draw[very thick, draw=red](1.2,2.4) -- (1.2,3.6);
        \draw[very thick, draw=red](1.2,1.2) -- (1.2+1.2,1.2) -- (1.2+1.2,3.6);
        \DTensor{-1.2-0.2,0}{0.8}{.5}{}{0};
        \DTensor{-1.2-0.2,1.2}{0.8}{.5}{}{0};
        \DTensor{-1.2-0.2,2.4}{0.8}{.5}{}{0};
        \draw [very thick,->] (-1.2-0.2+0.5,-1) -- (-1.2-0.2-0.5,-1);
    \end{scope}
\end{tikzpicture}\\ \label{mps_op_ex3}
&=\begin{tikzpicture}[scale=0.4,baseline={([yshift=-0.65ex] current bounding box.center) }]
    \draw[very thick, draw=red] (-3.4,0) -- (6*\singledx+4.4,0);
    \begin{scope}[shift={(0,0)}]
        \FUnitary{0,0}{1.2}{.6}{\small $U_C$}{0};
        \draw[ thick, rounded corners=2pt] (-1.2-1.2-1.,1.2-0.6) rectangle (-1.2-1.,3.6-0.6);
        \draw (-1.2-0.6-1.,2.4-0.6) node {\small $\ket{\psi}$};
        \draw[very thick, draw=red](1.2,2.4) -- (1.2,3.6);
        \draw[very thick, draw=red](1.2,1.2) -- (1.2+1.2,1.2) -- (1.2+1.2,3.6);
        \DTensor{-1.2-0.2,0}{0.8}{.5}{}{0};
        \DTensor{-1.2-0.2,1.2}{0.8}{.5}{}{0};
        \DTensor{-1.2-0.2,2.4}{0.8}{.5}{}{0};
    \end{scope}
    \begin{scope}[shift={(3*\singledx+1,0)}]
        \FUnitary{0,0}{1.2}{.6}{\small $U_C$}{0};
        \draw[ thick, rounded corners=2pt] (-1.2-1.2-1.,1.2-0.6) rectangle (-1.2-1.,3.6-0.6);
        \draw (-1.2-0.6-1.,2.4-0.6) node {\small $\ket{\psi}$};
        \draw[very thick, draw=red](1.2,2.4) -- (1.2,3.6);
        \draw[very thick, draw=red](1.2,1.2) -- (1.2+1.2,1.2) -- (1.2+1.2,3.6);
        \DTensor{-1.2-0.2,1.2}{0.8}{.5}{}{0};
        \DTensor{-1.2-0.2,2.4}{0.8}{.5}{}{0};
    \end{scope}
    \begin{scope}[shift={(6*\singledx+2,0)}]
        \FUnitary{0,0}{1.2}{.6}{\small $U_C$}{0};
        \draw[ thick, rounded corners=2pt] (-1.2-1.2-1.,1.2-0.6) rectangle (-1.2-1.,3.6-0.6);
        \draw (-1.2-0.6-1.,2.4-0.6) node {\small $\ket{\psi}$};
        \draw[very thick, draw=red](1.2,2.4) -- (1.2,3.6);
        \draw[very thick, draw=red](1.2,1.2) -- (1.2+1.2,1.2) -- (1.2+1.2,3.6);
        \DTensor{-1.2-0.2,1.2}{0.8}{.5}{}{0};
        \DTensor{-1.2-0.2,2.4}{0.8}{.5}{}{0};
    \end{scope}
\end{tikzpicture}
\end{align}

\vtwo{In the second line above we push the operators from the physical legs to the virtual legs and the purification legs; in the third line we push the operators on the virtual legs to the left. After pushing back, the operator expectation values factorizes into terms supported on product initial states, rendering the computation to be efficient even though the state may have high magic. While we are evaluating the Pauli observables on a possibly expanded space after the isometries, they can be related to the observables on the physical leg. In the case of injective MPS, the Pauli operators simply relates to some operators on the physical leg via a rotation $V$; in the case of non-injective MPS, one can still relate $U_p$ to some Pauli operators via Corollary \ref{Up_PP_relation}.

We note that computing the expectation value of any product of local observable is computationally efficient with MPS, so the sequential generation structure does not bring computational advantage in one dimension. Nevertheless, we will see later that MF PEPS has the similar sequential generation structure which allows for efficient computation of the Weyl-Heisenberg observable. This property does bring a computational advantage in two dimensions.}

Lastly, one would still expect a similar structure when the MF group is not Weyl-Heisenberg, but it is not clear what is the analog of Clifford operators there. Specifically, one would need to find a set of unitaries that do not generate operator entanglement in a generic MF basis (one can always find a channel that does that, but it is not clear if such a channel is a unitary). We leave this question to future work.

\subsection{Connection to SPT}


The MF symmetry of MPS resembles the symmetry of the SPT order~\cite{hong2023quantum}, but there are subtle differences. In SPT the fractionalized operators can always be chosen to be identical $P'_i = P_i$, whereas the symmetry of the MF MPS need not be. In fact, we have seen that the map \(M: \mathcal B \to \mathcal B\) mapping \(P_i \mapsto P_i'\) in Eq.~\eqref{mps_sym} need not be bijective as well. 
%

As an example, we discuss the MF MPS with the following symmetry constraints:
\begin{align}
   \begin{tikzpicture}[scale=0.5,baseline={([yshift=-3.65ex] current bounding box.center) }]
       \DTensor{-\singledx,0}{1.}{.6}{\small $X$}{0};
       \UTensor{0,\singledx}{1.}{.6}{\small $X$}{3};
       \ATensor{0,0}{\small $A$}{0};
\end{tikzpicture}
   \; = \;\begin{tikzpicture}[scale=0.5,baseline={([yshift=-0.65ex] current bounding box.center) }]
       \DTensor{\singledx,0}{1.}{.6}{\small $Y$}{0};
       \ATensor{0,0}{\small $A$}{0};
\end{tikzpicture} \\
\begin{tikzpicture}[scale=0.5,baseline={([yshift=-3.65ex] current bounding box.center) }]
       \DTensor{-\singledx,0}{1.}{.6}{\small $Z$}{0};
       \UTensor{0,\singledx}{1.}{.6}{\small $Z$}{3};
       \ATensor{0,0}{\small $A$}{0};
\end{tikzpicture}
\; = \;\begin{tikzpicture}[scale=0.5,baseline={([yshift=-0.65ex] current bounding box.center) }]
       \DTensor{\singledx,0}{1.}{.6}{\small $I$}{0};
       \ATensor{0,0}{\small $A$}{0};
\end{tikzpicture}
\end{align}
\noindent This corresponds to the case where $M$ neither keeps \(P_i\) unchanged nor is bijective. To compare it with SPT, we block the site twice which results in the following symmetry:
\vtwo{
\begin{align}
   \begin{tikzpicture}[scale=0.5,baseline={([yshift=-3.65ex] current bounding box.center) }]
       \DTensor{-\singledx,0}{1.}{.6}{\small $Y$}{0};
       \ATensor{0,0}{\small $A$}{0};
       \ATensor{\singledx,0}{\small $A$}{0};
       \UTensor{0,\singledx}{1.}{.6}{\small $Y$}{3};
       \UTensor{\singledx,\singledx}{1.}{.6}{\small $Y$}{3};
\end{tikzpicture}
\; = \;
   \begin{tikzpicture}[scale=0.5,baseline={([yshift=-0.65ex] current bounding box.center) }]
       \ATensor{0,0}{\small $A$}{0};
       \ATensor{\singledx,0}{\small $A$}{0};
       \DTensor{2*\singledx,0}{1.}{.6}{\small $Y$}{0};
\end{tikzpicture}
\end{align}
}
We see that only \(Y\) of the virtual space can be pushed through, and all other Pauli operators result in identity. Therefore, the corresponding symmetry is not the Pauli group \(\mathbb{Z}_2 \times \mathbb{Z}_2\) but only \(\mathbb{Z}_2\). Because \(\mathbb{Z}_2\) only has trivial second group cohomology, the resulting state has the trivial type of SPT order.


When $M$ is bijective, one can reproduce the SPT-type symmetry (where $P_i' = P_i$) by blocking adjacent sites.
\begin{observation}
If the map \(M: \mathcal B \to \mathcal B\) mapping \(P_i \mapsto P_i'\) as defined by the MF symmetry [Eq.~\eqref{mps_sym}] is bijective, then after blocking a finite (i.e., system-size independent) number of times, the resulting supersite possesses the SPT-type symmetry:
\begin{align}\label{full_push_thru_sym}
    \begin{tikzpicture}[scale=0.5,baseline={([yshift=-0.65ex] current bounding box.center) }]
        \ATensor{0,0}{\small $A$}{0};
\end{tikzpicture}
    \; = \;\begin{tikzpicture}[scale=0.5,baseline={([yshift=-3.65ex] current bounding box.center) }]
        \DTensor{-\singledx,0}{1.}{.6}{\small $P_i$}{0};
        \UTensor{0,\singledx}{1.}{.6}{\small $U_{P_i}$}{3};
        \ATensor{0,0}{\small $A$}{0};
        \DTensor{\singledx,0}{1.}{.6}{\small $P_i^{\dag}$}{0};
        \RedArrow{-0.8,0}{0.2}{1};
        \RedArrow{1,0}{0.2}{1};
\end{tikzpicture}
\; , \quad \forall P_i \;.
\end{align}
\end{observation}
The above observation is simply proven by the fact that a bijective \(M\) defines a permutation in the MF basis, and any permutation always has a finite order. Hence, the specific map \(M\) does not matter critically, but the bijectiveness does matter in terms of connecting to the SPT order. For instance, consider the following symmetries:
\begin{align}
   \begin{tikzpicture}[scale=0.5,baseline={([yshift=-3.65ex] current bounding box.center) }]
       \DTensor{-\singledx,0}{1.}{.6}{\small $X$}{0};
       \UTensor{0,\singledx}{1.}{.6}{\small $X$}{3};
       \ATensor{0,0}{\small $A$}{0};
\end{tikzpicture}
   \; = \;\begin{tikzpicture}[scale=0.5,baseline={([yshift=-0.65ex] current bounding box.center) }]
       \DTensor{\singledx,0}{1.}{.6}{\small $Y$}{0};
       \ATensor{0,0}{\small $A$}{0};
\end{tikzpicture} \\
\begin{tikzpicture}[scale=0.5,baseline={([yshift=-3.65ex] current bounding box.center) }]
       \DTensor{-\singledx,0}{1.}{.6}{\small $Z$}{0};
       \UTensor{0,\singledx}{1.}{.6}{\small $I$}{3};
       \ATensor{0,0}{\small $A$}{0};
\end{tikzpicture}
\; = \;\begin{tikzpicture}[scale=0.5,baseline={([yshift=-0.65ex] current bounding box.center) }]
       \DTensor{\singledx,0}{1.}{.6}{\small $X$}{0};
       \ATensor{0,0}{\small $A$}{0};
\end{tikzpicture}
\end{align}
The symmetries define a map \(X \mapsto Y\), \(Y \mapsto Z\), \(Z \mapsto X\) which is is equivalent to $M$ being the identity after blocking three sites.

\subsection{Characterization of MF MPS with SPT-type symmetry}

We now focus on the case where the map \(M: \mathcal B \to \mathcal B\) is identity, in other words, when the MPS tensors have the SPT-type symmetry discussed previously. Such symmetries admit the following solution when we require the MF basis to form a projective representation of an Abelian group. For example, the Weyl-Heisenberg group is a projective representation of \(\mathbb{Z}_D \times \mathbb{Z}_D\). We note that all the known constructions from Ref.~\cite{werner2001all} are Abelian. Under the Abelian condition, we provide the following analytical solution to Eq.~\eqref{full_push_thru_sym}.
\begin{theorem}
Let \(\{P_i\}\) be a projective representation of an Abelian MF group, then under Eq.~\eqref{full_push_thru_sym},
\begin{align}
\begin{tikzpicture}[scale=0.5,baseline={([yshift=-0.65ex] current bounding box.center) }]
        \GFTensor{0,0}{1.2}{.6}{\small $Q$}{0};
\end{tikzpicture}
=\sum_{i} \alpha_i \begin{tikzpicture}[scale=0.5,baseline={([yshift=-0.65ex] current bounding box.center) }]
        \draw[very thick, draw=red] (-2,0) -- (0,0);
        \BTensor{0,1}{1}{.6}{\small $P_i^\dag$}{1};
        \BTensor{\singledx,1}{1}{.6}{\small $P_i$}{1};
        \draw[very thick, draw=red] (\singledx,0) -- (\singledx+2,0);
        \RedArrow{-1,0}{0.2}{1};
        \RedArrow{\singledx+1,0}{0.2}{1};
\end{tikzpicture}
\end{align}
\end{theorem}
\begin{proof}
We expand \(Q\) in the basis of the MF group (recall that the MF group forms a complete operator basis)
\begin{align}
\begin{tikzpicture}[scale=0.5,baseline={([yshift=-0.65ex] current bounding box.center) }]
        \GFTensor{0,0}{1.2}{.6}{\small $Q$}{0};
\end{tikzpicture}
=\sum_{ij}\alpha_{ij} \begin{tikzpicture}[scale=0.5,baseline={([yshift=-0.65ex] current bounding box.center) }]
        \draw[very thick, draw=red] (-2,0) -- (0,0);
        \BTensor{0,1}{1}{.6}{\small $P_i^\dag$}{1};
        \BTensor{\singledx,1}{1}{.6}{\small $P_j$}{1};
        \draw[very thick, draw=red] (\singledx,0) -- (\singledx+2,0);
        \RedArrow{-1,0}{0.2}{1};
        \RedArrow{\singledx+1,0}{0.2}{1};
\end{tikzpicture}
\;.
\end{align}
Per Theorem~\ref{structural_theorem_1}, we apply the symmetry \(P_i^{\dag} \otimes P_i\) to the tensor in the expanded basis (note that the orientation of the right leg is reversed):
\begin{widetext}
\begin{align*}
\sum_{ij}\alpha_{ij} \begin{tikzpicture}[scale=0.5,baseline={([yshift=-0.65ex] current bounding box.center) }]
        \draw[very thick, draw=red] (-2,0) -- (0,0);
        \BTensor{0,1}{1}{.6}{\small $P_i^\dag$}{1};
        \BTensor{\singledx,1}{1}{.6}{\small $P_j$}{1};
        \draw[very thick, draw=red] (\singledx,0) -- (\singledx+2,0);
        \RedArrow{-1,0}{0.2}{1};
        \RedArrow{\singledx+1,0}{0.2}{1};
\end{tikzpicture}
=\sum_{ij}\alpha_{ij} \begin{tikzpicture}[scale=0.5,baseline={([yshift=-0.65ex] current bounding box.center) }]
        \draw[very thick, draw=red] (-2,0) -- (0,0);
        \draw[very thick, draw=red] (\singledx,0) -- (\singledx+2,0);
        \BTensor{0,1}{1}{.6}{\small $P_i^{\dag}$}{1};
        \BTensor{\singledx,1}{1}{.6}{\small $P_j$}{1};
        \DTensor{-\singledx,0}{1}{.6}{\small $P_k$}{0};
        \DTensor{2*\singledx,0}{1}{.6}{\small $P_k^{\dag}$}{0};
        \UTensor{0,1+\singledx}{1}{.6}{\small $P_k^{\dag}$}{1};
        \UTensor{\singledx,1+\singledx}{1}{.6}{\small $P_k$}{1};
        \RedArrow{-0.5,0}{0.2}{1};
        \RedArrow{\singledx+0.5,0}{0.2}{1};
\end{tikzpicture}
=\sum_{ij}\alpha_{ij} \omega(i,k)^{-1}\omega(j,k) \begin{tikzpicture}[scale=0.5,baseline={([yshift=-0.65ex] current bounding box.center) }]
        \draw[very thick, draw=red] (-2,0) -- (0,0);
        \draw[very thick, draw=red] (\singledx,0) -- (\singledx+2,0);
        \UTensor{0,1}{1}{.6}{\small $P_k^{\dag}$}{1};
        \UTensor{\singledx,1}{1}{.6}{\small $P_k$}{1};
        \DTensor{-\singledx,0}{1}{.6}{\small $P_k$}{0};
        \DTensor{2*\singledx,0}{1}{.6}{\small $P_k^{\dag}$}{0};
        \BTensor{0,1+\singledx}{1}{.6}{\small $P_i^{\dag}$}{1};
        \BTensor{\singledx,1+\singledx}{1}{.6}{\small $P_j$}{1};
        \RedArrow{-0.5,0}{0.2}{1};
        \RedArrow{\singledx+0.5,0}{0.2}{1};
\end{tikzpicture}
, \forall k
\end{align*}
\end{widetext}
where\(\omega(j,k)\) is defined as the phase after commuting \(P_j\) and \(P_k\) due to the projective representation: \(P_k P_j = \omega(j,k) P_j P_k\). In the last equality, we commute \(P_k\) through \(P_i^\dag\) and \(P_j\) and include the phase \(\omega(i,k)^{-1}\omega(j,k)\). The above equation has to be satisfied term-by-term, implying that, if $\alpha_{ij} \ne 0$, then \(\omega(i,k)\omega(j,k)^{-1}=1, \forall k\). This restricts the non-zero terms to be those with \(i=j\), because otherwise \(P_iP_j^{-1}\) commutes with all \(P_k\). But since \(P_k\) forms a complete operator basis, it implies that \(P_iP_j^{-1} \propto I\), so \(P_i\) and \(P_j\) are the same up to a phase.
\end{proof}

As an example, the AKLT tensor can also be written in the following way:
\begin{align}
\begin{tikzpicture}[scale=0.4,baseline={([yshift=-0.65ex] current bounding box.center) }]
        \GFTensor{0,0}{1.2}{.6}{\small $Q$}{0};
\end{tikzpicture}
=3\begin{tikzpicture}[scale=0.45,baseline={([yshift=-0.65ex] current bounding box.center) }]
        \draw[very thick, draw=red] (-2,0) -- (0,0);
        \BTensor{0,1}{1}{.6}{\small $I$}{1};
        \BTensor{\singledx,1}{1}{.6}{\small $I$}{1};
        \draw[very thick, draw=red] (\singledx,0) -- (\singledx+2,0);
        \RedArrow{-1,0}{0.2}{1};
        \RedArrow{\singledx+1,0}{0.2}{1};
\end{tikzpicture}
-\begin{tikzpicture}[scale=0.45,baseline={([yshift=-0.65ex] current bounding box.center) }]
        \draw[very thick, draw=red] (-2,0) -- (0,0);
        \BTensor{0,1}{1}{.6}{\small $X$}{1};
        \BTensor{\singledx,1}{1}{.6}{\small $X$}{1};
        \draw[very thick, draw=red] (\singledx,0) -- (\singledx+2,0);
        \RedArrow{-1,0}{0.2}{1};
        \RedArrow{\singledx+1,0}{0.2}{1};
\end{tikzpicture} \nonumber \\
-\begin{tikzpicture}[scale=0.45,baseline={([yshift=-0.65ex] current bounding box.center) }]
        \draw[very thick, draw=red] (-2,0) -- (0,0);
        \BTensor{0,1}{1}{.6}{\small $Y$}{1};
        \BTensor{\singledx,1}{1}{.6}{\small $Y$}{1};
        \draw[very thick, draw=red] (\singledx,0) -- (\singledx+2,0);
        \RedArrow{-1,0}{0.2}{1};
        \RedArrow{\singledx+1,0}{0.2}{1};
\end{tikzpicture}
-\begin{tikzpicture}[scale=0.45,baseline={([yshift=-0.65ex] current bounding box.center) }]
        \draw[very thick, draw=red] (-2,0) -- (0,0);
        \BTensor{0,1}{1}{.6}{\small $Z$}{1};
        \BTensor{\singledx,1}{1}{.6}{\small $Z$}{1};
        \draw[very thick, draw=red] (\singledx,0) -- (\singledx+2,0);
        \RedArrow{-1,0}{0.2}{1};
        \RedArrow{\singledx+1,0}{0.2}{1};
\end{tikzpicture}
\end{align}
In fact, one can recognize this as local-unitary equivalent to the triplet projector which is known to generate the AKLT state. Lastly, we note that similar results were recently obtained in \cite{sahay2024classifying} in a different but equivalent basis.

\vtwo{However, even under the same SPT-like symmetry, not all MF MPS states are in the same SPT phase. For example, a \(D\)-level GHZ state can be written as (up to an isometry and normalization):}
\begin{align}
\begin{tikzpicture}[scale=0.5,baseline={([yshift=-0.65ex] current bounding box.center) }]
        \GFTensor{0,0}{1.2}{.6}{\small $Q$}{0};
        \draw (-1.5,0) node {\small $a$};
        \draw (\singledx+1.5,0) node {\small $d$};
        \draw (0,1.5) node {\small $b$};
        \draw (\singledx,1.5) node {\small $c$};
\end{tikzpicture}
\; =\sum_{i} \begin{tikzpicture}[scale=0.5,baseline={([yshift=-0.65ex] current bounding box.center) }]
        \draw[very thick, draw=red] (-2,0) -- (0,0);
        \BTensor{0,1}{1}{.6}{\small $Z^{-i}$}{1};
        \BTensor{\singledx,1}{1}{.6}{\small $Z^{i}$}{1};
        \draw[very thick, draw=red] (\singledx,0) -- (\singledx+2,0);
        \RedArrow{-1,0}{0.2}{1};
        \RedArrow{\singledx+1,0}{0.2}{1};
        \draw (-2.3,0) node {\small $a$};
        \draw (\singledx+2.3,0) node {\small $d$};
        \draw (0,1+1.3) node {\small $b$};
        \draw (\singledx,1+1.3) node {\small $c$};
\end{tikzpicture}
\;.
\end{align}
\vtwo{Here \(Z\) is the phase operator in the Weyl-Heisenberg group. One can verify that the above tensor is equal to \(\delta_{a,b}\delta_{b,c}\delta_{c,d}\), so the physical leg only has an rank of \(D\), and the tensor is exactly the GHZ tensor after eliminating the null space in the physical space. By setting $D=2$, we obtain a qubit GHZ state which is in a different phase then the AKLT state. Hence, having the same MF symmetry does not imply that the states are in the same phase. However, when we demand the MF MPS to be injective, then different MF MPS satisfying the same SPT-like symmetry are indeed in the same phase \cite{chen2011classification,schuch2011classifying}.}


\section{PEPS}\label{section_peps}
\subsection{Overview of the Protocol}
We now extend the discussion to the 2D square lattice and consider a similar protocol shown in Eq.~(\ref{peps_glue}). We start by preparing five-party states in $\mathbb C^d \otimes (\mathbb C^D)^{\otimes 4}$ that correspond to the PEPS tensor. We then measure the adjacent virtual leg qubits in the MF basis. Finally, we apply local operations to push the operator \(P_i\) to the boundary.
\begin{align}\label{peps_glue}
\begin{tikzpicture}[scale=0.4,baseline={([yshift=-0.65ex] current bounding box.center) }]
        \GPTensor{\singledx,\singledx}{1.2}{.6}{}{0};
        \GPTensor{-\singledx,\singledx}{1.2}{.6}{}{0};
        \GPTensor{-\singledx,-\singledx}{1.2}{.6}{}{0};
        \GPTensor{\singledx,-\singledx}{1.2}{.6}{}{0};
        \Measurement{-\singledx,0};
        \Measurement{\singledx,0};
        \Measurement{0,-\singledx};
        \Measurement{0,\singledx};
\end{tikzpicture}
=\begin{tikzpicture}[scale=0.4,baseline={([yshift=-0.65ex] current bounding box.center) }]
        \GPTensor{\singledx,\singledx}{1.}{.6}{}{0};
        \GPTensor{-\singledx,\singledx}{1.}{.6}{}{0};
        \GPTensor{-\singledx,-\singledx}{1.}{.6}{}{0};
        \GPTensor{\singledx,-\singledx}{1.}{.6}{}{0};
        \DTensor{\singledx,0}{1.}{.6}{}{1};
        \DTensor{-\singledx,0}{1.}{.6}{}{1};
        \DTensor{0,\singledx}{1.}{.6}{}{0};
        \DTensor{0,-\singledx}{1.}{.6}{}{0};
\end{tikzpicture}
\end{align}

However, the symmetry constraints to push the operator throught is more subtle in the case of PEPS than in the case of MPS because there are multiple virtual legs, so one can push the operator to multiple legs. Due to this, not all the possible types of symmetries allow to deterministically correct the measurement operators. For instance, observe the following symmetry condition:
\begin{align}\label{peps_bad_sym}
\begin{tikzpicture}[scale=0.5,baseline={([yshift=-2.65ex] current bounding box.center) }]
        \GPTensor{0,0}{1.3}{.6}{}{0};
        \DTensor{-\singledx,0}{1.}{.6}{\small $P$}{0};
        \begin{scope}[rotate around={45:(1.6,1.6)}]
        \UTensor{1.6,1.6}{1.3}{.6}{\small $U_{P}$}{2};
        \end{scope}
\end{tikzpicture}
=\; \begin{tikzpicture}[scale=0.5,baseline={([yshift=-0.65ex] current bounding box.center) }]
        \GPTensor{0,0}{1.3}{.6}{}{0};
        \DTensor{0,\singledx}{1.}{.6}{\small $P$}{1};
        \DTensor{\singledx,0}{1.}{.6}{\small $P$}{0};
        \DTensor{0,-\singledx}{1.}{.6}{\small $P$}{1};
\end{tikzpicture}
\;.
\end{align}
This symmetry constraint does not allow pushing the operator to the boundary. This is demonstrated in the following example:
\begin{align}\label{peps_bad_sym_ex}
\begin{tikzpicture}[scale=0.3,baseline={([yshift=-0.65ex] current bounding box.center) }]
        \draw[very thick, draw=red] (-2*\singledx-1,-2*\singledx) -- (2*\singledx+1,-2*\singledx);
        \draw[very thick, draw=red] (-2*\singledx-1,-0*\singledx) -- (2*\singledx+1,-0*\singledx);
        \draw[very thick, draw=red] (-2*\singledx-1,2*\singledx) -- (2*\singledx+1,2*\singledx);
        \draw[very thick, draw=red] (-2*\singledx,-2*\singledx-1) -- (-2*\singledx,2*\singledx+1);
        \draw[very thick, draw=red] (-0*\singledx,-2*\singledx-1) -- (-0*\singledx,2*\singledx+1);
        \draw[very thick, draw=red] (2*\singledx,-2*\singledx-1) -- (2*\singledx,2*\singledx+1);
        \GPTensor{-2*\singledx,-2*\singledx}{1.}{.6}{}{0};
        \GPTensor{-0*\singledx,-2*\singledx}{1.}{.6}{}{0};
        \GPTensor{2*\singledx,-2*\singledx}{1.}{.6}{}{0};
        \GPTensor{-2*\singledx,-0*\singledx}{1.}{.6}{}{0};
        \GPTensor{-0*\singledx,-0*\singledx}{1.}{.6}{}{0};
        \GPTensor{2*\singledx,-0*\singledx}{1.}{.6}{}{0};
        \GPTensor{-2*\singledx,2*\singledx}{1.}{.6}{}{0};
        \GPTensor{-0*\singledx,2*\singledx}{1.}{.6}{}{0};
        \GPTensor{2*\singledx,2*\singledx}{1.}{.6}{}{0};
        \RTensor{0,\singledx}{1.}{.6}{}{1};
        \begin{scope}[shift={(-2*\singledx,-2*\singledx)}, rotate around={45:(1.4,1.4)}]
        \UTensor{1.4,1.4}{1.1}{.6}{}{2};
        \end{scope}
        \begin{scope}[shift={(-2*\singledx,-0*\singledx)}, rotate around={45:(1.4,1.4)}]
        \UTensor{1.4,1.4}{1.1}{.6}{}{2};
        \end{scope}
        \begin{scope}[shift={(-0*\singledx,-2*\singledx)}, rotate around={45:(1.4,1.4)}]
        \UTensor{1.4,1.4}{1.1}{.6}{}{2};
        \end{scope}
        \begin{scope}[shift={(-0*\singledx,-0*\singledx)}, rotate around={45:(1.4,1.4)}]
        \UTensor{1.4,1.4}{1.1}{.6}{}{2};
        \end{scope}
        \begin{scope}[shift={(2*\singledx,-2*\singledx)}, rotate around={45:(1.4,1.4)}]
        \UTensor{1.4,1.4}{1.1}{.6}{}{2};
        \end{scope}
        \begin{scope}[shift={(2*\singledx,-0*\singledx)}, rotate around={45:(1.4,1.4)}]
        \UTensor{1.4,1.4}{1.1}{.6}{}{2};
        \end{scope}
\end{tikzpicture}
=\begin{tikzpicture}[scale=0.3,baseline={([yshift=1ex] current bounding box.center) }]
        \draw[very thick, draw=red] (-2*\singledx-1,-2*\singledx) -- (2*\singledx+1,-2*\singledx);
        \draw[very thick, draw=red] (-2*\singledx-1,-0*\singledx) -- (2*\singledx+1,-0*\singledx);
        \draw[very thick, draw=red] (-2*\singledx-1,2*\singledx) -- (2*\singledx+1,2*\singledx);
        \draw[very thick, draw=red] (-2*\singledx,-2*\singledx-1) -- (-2*\singledx,2*\singledx+1);
        \draw[very thick, draw=red] (-0*\singledx,-2*\singledx-1) -- (-0*\singledx,2*\singledx+1);
        \draw[very thick, draw=red] (2*\singledx,-2*\singledx-1) -- (2*\singledx,2*\singledx+1);
        \GPTensor{-2*\singledx,-2*\singledx}{1.}{.6}{}{0};
        \GPTensor{-0*\singledx,-2*\singledx}{1.}{.6}{}{0};
        \GPTensor{2*\singledx,-2*\singledx}{1.}{.6}{}{0};
        \GPTensor{-2*\singledx,-0*\singledx}{1.}{.6}{}{0};
        \GPTensor{-0*\singledx,-0*\singledx}{1.}{.6}{}{0};
        \GPTensor{2*\singledx,-0*\singledx}{1.}{.6}{}{0};
        \GPTensor{-2*\singledx,2*\singledx}{1.}{.6}{}{0};
        \GPTensor{-0*\singledx,2*\singledx}{1.}{.6}{}{0};
        \GPTensor{2*\singledx,2*\singledx}{1.}{.6}{}{0};
        \RTensor{-2*\singledx,\singledx}{1.}{.6}{}{1};
        \RTensor{2*\singledx,\singledx}{1.}{.6}{}{1};
        \DTensor{-3*\singledx,0}{1.}{.6}{}{0};
        \DTensor{3*\singledx,0}{1.}{.6}{}{0};
        \DTensor{-3*\singledx,-2*\singledx}{1.}{.6}{}{0};
        \DTensor{3*\singledx,-2*\singledx}{1.}{.6}{}{0};
        \DTensor{-2*\singledx,-3*\singledx}{1.}{.6}{}{1};
        \DTensor{-0*\singledx,-3*\singledx}{1.}{.6}{}{1};
        \DTensor{2*\singledx,-3*\singledx}{1.}{.6}{}{1};
\end{tikzpicture}
\,.
\end{align}
Here, attempting to push the purple operator to the boundary creates two more purple operators. This is in fact inevitable because of a topological constraint. The operator can be considered as a line connecting two plaquettes. The symmetry constraint only deform the string but it always ends at the same plaquettes. This is analogous to the Wilson line operator connecting two fluxes in the toric code. Since the plaquettes are in the bulk, any deformation will necessarily contain part of the string in the bulk, rendering the error non-correctable.

If we do not impose additional symmetry constraints to avoid the above situation, then the minimal symmetry constraint to enable correcting all operators is the following where we push the operators to the upper-right corner:
\begin{subequations} \label{peps_isosym}
\begin{align}\label{peps_isosym1}
\begin{tikzpicture}[scale=0.5,baseline={([yshift=-2.6ex] current bounding box.center) }]
        \GPTensor{0,0}{1.3}{.6}{$A$}{0};
        \RedArrow{-0.8,0}{0.2}{1}
        \RedArrow{1.0,0}{0.2}{1}
        \RedArrow{0,1.0}{0.2}{0}
        \RedArrow{0,-0.8}{0.2}{0}
        \DTensor{-\singledx,0}{1.}{.6}{\small $P_A$}{0};
        \begin{scope}[rotate around={45:(1.6,1.6)}]
        \UTensor{1.6,1.6}{1.3}{.6}{\small $U_{P_A}$}{2};
        \end{scope}
\end{tikzpicture}
&= \; \begin{tikzpicture}[scale=0.5,baseline={([yshift=-2.7ex] current bounding box.center) }]
        \GPTensor{0,0}{1.3}{.6}{$A$}{0};
        \RedArrow{-0.8,0}{0.2}{1}
        \RedArrow{1.0,0}{0.2}{1}
        \RedArrow{0,1.0}{0.2}{0}
        \RedArrow{0,-0.8}{0.2}{0}
        \DTensor{0,\singledx}{1.}{.6}{\small $P_{A1}$}{1};
        \DTensor{\singledx,0}{1.}{.6}{\small $P_{A2}$}{0};
\end{tikzpicture}\\\label{peps_isosym2}
\begin{tikzpicture}[scale=0.5,baseline={([yshift=0ex] current bounding box.center) }]
        \GPTensor{0,0}{1.3}{.6}{$A$}{0};
        \RedArrow{-0.8,0}{0.2}{1}
        \RedArrow{1.0,0}{0.2}{1}
        \RedArrow{0,1.0}{0.2}{0}
        \RedArrow{0,-0.8}{0.2}{0}
        \DTensor{0,-\singledx}{1.}{.6}{\small $P_B$}{1};
        \begin{scope}[rotate around={45:(1.6,1.6)}]
        \UTensor{1.6,1.6}{1.3}{.6}{\small $U_{P_B}$}{2};
        \end{scope}
\end{tikzpicture}
&= \; \begin{tikzpicture}[scale=0.5,baseline={([yshift=-2.7ex] current bounding box.center) }]
        \GPTensor{0,0}{1.3}{.6}{$A$}{0};
        \RedArrow{-0.8,0}{0.2}{1}
        \RedArrow{1.0,0}{0.2}{1}
        \RedArrow{0,1.0}{0.2}{0}
        \RedArrow{0,-0.8}{0.2}{0}
        \DTensor{0,\singledx}{1.}{.6}{\small $P_{B1}$}{1};
        \DTensor{\singledx,0}{1.}{.6}{\small $P_{B2}$}{0};
\end{tikzpicture}
\;.
\end{align}
\end{subequations}
We refer to tensors with such a symmetry as MF PEPS.

\subsection{Structural Theorem of MF PEPS}
We now derive the structural theorem for MF PEPS, analogous to the case of MF MPS. To begin with, the symmetry constraints Eqs.~(\ref{peps_isosym}) result in the PEPS being an isometric tensor network (isoTNS) \cite{zaletel2020isometric}.

\begin{lemma}
    Under the symmetry constraints Eqs.~(\ref{peps_isosym}),
    \begin{align}
    \begin{tikzpicture}[scale=0.5,baseline={([yshift=-0.65ex] current bounding box.center) }]
            \draw[very thick, draw=red] (0,0) -- (0,\singledx) -- (\singledx,\singledx);
            \draw[very thick, draw=red] (0,0) -- (\singledx,0) -- (\singledx,\singledx);
            \GPTensor{0,0}{1.3}{.6}{\small $A$}{0};
            \RedArrow{-0.8,0}{0.2}{1}
            \RedArrow{1.0,0}{0.2}{1}
            \RedArrow{0,1.0}{0.2}{0}
            \RedArrow{0,-0.8}{0.2}{0}
            \GPTensor{\singledx,\singledx}{1.3}{.6}{\small $A^\dag$}{2};
            \begin{scope}[shift={(\singledx,\singledx)}]
                \RedArrow{-0.8,0}{0.2}{1}
                \RedArrow{1.0,0}{0.2}{1}
                \RedArrow{0,1.0}{0.2}{0}
                \RedArrow{0,-0.8}{0.2}{0}
            \end{scope}
            \begin{scope}[rotate around={-45:(0,0)}]
                \Arrow{0,1.5}{0.2}{0}
            \end{scope}
    \end{tikzpicture}
    = \begin{tikzpicture}[scale=0.5,baseline={([yshift=-0.65ex] current bounding box.center) }]
            \draw[very thick, draw=red] (-1.2,0) -- (1.2, 0);
            \draw[very thick, draw=red] (0,-1.2) -- (0,1.2);
            \RedArrow{-0.8,0}{0.2}{1}
            \RedArrow{0,-0.8}{0.2}{0}
            \RedArrow{1.,0}{0.2}{1}
            \RedArrow{0,1.}{0.2}{0}
    \end{tikzpicture}
    \end{align}
\end{lemma}
\begin{proof}
    The proof is analogous to the case of MPS. First we apply the symmetry constraint to show that \(\forall P_A, P_B,\)
\begin{align}
    \begin{tikzpicture}[scale=0.5,baseline={([yshift=-0.65ex] current bounding box.center) }]
            \draw[very thick, draw=red] (0,0) -- (0,\singledx) -- (\singledx,\singledx);
            \draw[very thick, draw=red] (0,0) -- (\singledx,0) -- (\singledx,\singledx);
            \GPTensor{0,0}{1.3}{.6}{\small $A$}{0};
            \RedArrow{-0.8,0}{0.2}{1}
            \RedArrow{1.0,0}{0.2}{1}
            \RedArrow{0,1.0}{0.2}{0}
            \RedArrow{0,-0.8}{0.2}{0}
            \GPTensor{\singledx,\singledx}{1.3}{.6}{\small $A^\dag$}{2};
            \begin{scope}[shift={(\singledx,\singledx)}]
                \RedArrow{-0.8,0}{0.2}{1}
                \RedArrow{1.0,0}{0.2}{1}
                \RedArrow{0,1.0}{0.2}{0}
                \RedArrow{0,-0.8}{0.2}{0}
            \end{scope}
            \begin{scope}[rotate around={-45:(0,0)}]
                \Arrow{0,1.5}{0.2}{0}
            \end{scope}
    \end{tikzpicture}
    = \begin{tikzpicture}[scale=0.5,baseline={([yshift=-0.65ex] current bounding box.center) }]
            \draw[very thick, draw=red] (0,0) -- (0,\singledx) -- (\singledx,\singledx);
            \draw[very thick, draw=red] (0,0) -- (\singledx,0) -- (\singledx,\singledx);
            \GPTensor{0,0}{1.3}{.6}{\small $A$}{0};
            \RedArrow{-0.8,0}{0.2}{1}
            \RedArrow{1.0,0}{0.2}{1}
            \RedArrow{0,1.0}{0.2}{0}
            \RedArrow{0,-0.8}{0.2}{0}
            \GPTensor{\singledx,\singledx}{1.3}{.6}{\small $A^\dag$}{2};
            \begin{scope}[shift={(\singledx,\singledx)}]
                \RedArrow{-0.8,0}{0.2}{1}
                \RedArrow{1.0,0}{0.2}{1}
                \RedArrow{0,1.0}{0.2}{0}
                \RedArrow{0,-0.8}{0.2}{0}
            \end{scope}
            \begin{scope}[rotate around={-45:(0,0)}]
                \Arrow{0,1.5}{0.2}{0}
            \end{scope}
            \DTensor{-\singledx,0}{1.2}{.6}{\small $P_A$}{0}
            \DTensor{0,-\singledx}{1.2}{.6}{\small $P_B$}{1}
            \DTensor{2*\singledx,\singledx}{1.2}{.6}{\small $P_A^\dag$}{0}
            \DTensor{\singledx,2*\singledx}{1.2}{.6}{\small $P_B^\dag$}{1}
    \end{tikzpicture}
    \end{align}
Then Schur's lemma directly implies the isometry property.
\end{proof}

Note that, although every MPS can be expressed without loss of generality in terms of isometric tensors, the same is not true for a general PEPS. That is, isometric PEPS (of a fixed virtual dimension $D$), as the ones appearing in the Lemma above, form a strict subclass of PEPS with the same $D$.

Similar to the case of MPS, MF PEPS also admit a convenient Clifford-like structure after performing the polar decomposition.
\begin{theorem}\label{iso_structural_theorem}
    Under the conditions~(\ref{peps_isosym}), the polar decomposition of the PEPS tensor reads
\begin{align}
    \begin{tikzpicture}[scale=0.5,baseline={([yshift=-0.65ex] current bounding box.center) }]
        \GPTensor{0,0}{1.2}{.6}{}{0};
\end{tikzpicture}
=\begin{tikzpicture}[scale=0.5,baseline={([yshift=-0.65ex] current bounding box.center) }]
        \draw[very thick, draw=red] (-\singledx-1.2,0) -- (\singledx+1.2,0);
        \draw[very thick, draw=red] (0,-\singledx-1.2) -- (0,\singledx+1.2);
        \draw[thick, fill=tensorcolor, rounded corners=2pt] (-\singledx-0.6,-\singledx-0.6) rectangle (\singledx+0.6,\singledx+0.6);
        \draw[thick, fill=white, rounded corners=2pt] (-\singledx+0.6,-\singledx+0.6) rectangle (\singledx-0.6,\singledx-0.6);
        \UPTensor{0,0}{1.2}{.6}{\small $V$}{3};
        \RedArrow{-0.8-\singledx,0}{0.2}{1}
        \RedArrow{0.8+\singledx,0}{0.2}{3}
        \RedArrow{0,0.8+\singledx}{0.2}{2}
        \RedArrow{0,-0.8-\singledx}{0.2}{0}
        \RedArrow{-0.8,0}{0.2}{1}
        \RedArrow{0.8,0}{0.2}{3}
        \RedArrow{0,0.8}{0.2}{2}
        \RedArrow{0,-0.8}{0.2}{0}
        \draw (0,-\singledx) node {\small $Q$};
\end{tikzpicture}
\end{align}
where \(Q: \mathbb{C}^{D^4} \rightarrow \mathbb{C}^{D^4}\) is a positive-semidefinite matrix and \(V:  \mathbb{C}^{D^4} \rightarrow \mathbb{C}^{d}\) is an isometry in the range of \(Q\), satisfying:
\begin{enumerate}[(i)]
    \item \(V\) has the same null space as \(Q\), i.e.,
    \begin{align}
        \begin{tikzpicture}[scale=0.5,baseline={([yshift=-0.65ex] current bounding box.center) }]
        \GPTensor{0,0}{1.2}{.6}{}{0};
        \draw (-1.4,0) node {\small $i_v$};
        \draw (0,1.4) node {\small $j_v$};
        \draw (1.4,0) node {\small $k_v$};
        \draw (0,-1.4) node {\small $l_v$};
        \begin{scope}[shift={(\singledx,\singledx)},xscale=-1,yscale=-1]
            \UPTensor{0,0}{1.2}{.6}{\small $V^\dag$}{3};
            \draw (-1.4,0) node {\small $k_p$};
            \draw (0,1.4) node {\small $l_p$};
            \draw (1.4,0) node {\small $i_p$};
            \draw (0,-1.4) node {\small $j_p$};
        \end{scope}
        \end{tikzpicture}
        =\begin{tikzpicture}[scale=0.5,baseline={([yshift=-0.65ex] current bounding box.center) }]
        \draw[thick, fill=tensorcolor, rounded corners=2pt] (-\singledx-0.6,-\singledx-0.6) rectangle (\singledx+0.6,\singledx+0.6);
        \draw[thick, fill=white, rounded corners=2pt] (-\singledx+0.6,-\singledx+0.6) rectangle (\singledx-0.6,\singledx-0.6);
        \draw[very thick, draw=red] (-\singledx-1.2,0) -- (-\singledx-0.6,0);
        \draw[very thick, draw=red] (\singledx+0.6,0) -- (\singledx+1.2,0);
        \draw[very thick, draw=red] (0,-\singledx-1.2) -- (0,-\singledx-0.6);
        \draw[very thick, draw=red] (0,\singledx+0.6) -- (0,\singledx+1.2);
        \draw[very thick, draw=red] (-\singledx+1.2,0) -- (-\singledx+0.6,0);
        \draw[very thick, draw=red] (\singledx-0.6,0) -- (\singledx-1.2,0);
        \draw[very thick, draw=red] (0,-\singledx+1.2) -- (0,-\singledx+0.6);
        \draw[very thick, draw=red] (0,\singledx-0.6) -- (0,\singledx-1.2);
        \RedArrow{-0.8-\singledx,0}{0.2}{1}
        \RedArrow{0.8+\singledx,0}{0.2}{3}
        \RedArrow{0,0.8+\singledx}{0.2}{2}
        \RedArrow{0,-0.8-\singledx}{0.2}{0}
        \RedArrow{-0.8,0}{0.2}{1}
        \RedArrow{0.8,0}{0.2}{3}
        \RedArrow{0,0.8}{0.2}{2}
        \RedArrow{0,-0.8}{0.2}{0}
        \draw (-1.4-\singledx,0) node {\small $i_v$};
        \draw (0,1.4+\singledx) node {\small $j_v$};
        \draw (1.4+\singledx,0) node {\small $k_v$};
        \draw (0,-1.4-\singledx) node {\small $l_v$};
        \draw (-0.4,0) node {\small $i_p$};
        \draw (0,0.4) node {\small $j_p$};
        \draw (0.4,0) node {\small $k_p$};
        \draw (0,-0.4) node {\small $l_p$};
        \draw (0,-\singledx) node {\small $Q$};
\end{tikzpicture}
    \end{align}
    where the corresponding legs are labelled for clarity.
    \item \(P_A \otimes I \otimes P_{A1}^T \otimes P_{A2}^T\) and \(I \otimes P_B \otimes P_{B1}^T \otimes P_{B2}^T\) commute with \(Q\), \(\forall P_A, P_B\).
    \item When \(\{P_A\}, \{P_B\}\) are the \(D\)-dimensional Weyl-Heisenberg group, \(Q\) can be written in the following Clifford form:
    \begin{align}
    \begin{tikzpicture}[scale=0.5,baseline={([yshift=-0.65ex] current bounding box.center) }]
        \draw[thick, fill=tensorcolor, rounded corners=2pt] (-\singledx-0.6,-\singledx-0.6) rectangle (\singledx+0.6,\singledx+0.6);
        \draw[thick, fill=white, rounded corners=2pt] (-\singledx+0.6,-\singledx+0.6) rectangle (\singledx-0.6,\singledx-0.6);
        \draw[very thick, draw=red] (-\singledx-1.2,0) -- (-\singledx-0.6,0);
        \draw[very thick, draw=red] (\singledx+0.6,0) -- (\singledx+1.2,0);
        \draw[very thick, draw=red] (0,-\singledx-1.2) -- (0,-\singledx-0.6);
        \draw[very thick, draw=red] (0,\singledx+0.6) -- (0,\singledx+1.2);
        \draw[very thick, draw=red] (-\singledx+1.2,0) -- (-\singledx+0.6,0);
        \draw[very thick, draw=red] (\singledx-0.6,0) -- (\singledx-1.2,0);
        \draw[very thick, draw=red] (0,-\singledx+1.2) -- (0,-\singledx+0.6);
        \draw[very thick, draw=red] (0,\singledx-0.6) -- (0,\singledx-1.2);
        \RedArrow{-0.8-\singledx,0}{0.2}{1}
        \RedArrow{1.+\singledx,0}{0.2}{1}
        \RedArrow{0,1.+\singledx}{0.2}{0}
        \RedArrow{0,-0.8-\singledx}{0.2}{0}
        \RedArrow{-0.8,0}{0.2}{1}
        \RedArrow{0.8,0}{0.2}{3}
        \RedArrow{0,0.8}{0.2}{2}
        \RedArrow{0,-0.8}{0.2}{0}
        \draw (-1.4-\singledx,0) node {\small $i_v$};
        \draw (0,1.4+\singledx) node {\small $j_v$};
        \draw (1.4+\singledx,0) node {\small $k_v$};
        \draw (0,-1.4-\singledx) node {\small $l_v$};
        \draw (-0.4,0) node {\small $i_p$};
        \draw (0,0.4) node {\small $j_p$};
        \draw (0.4,0) node {\small $k_p$};
        \draw (0,-0.4) node {\small $l_p$};
        \draw (0,-\singledx) node {\small $Q$};
\end{tikzpicture}
    =\begin{tikzpicture}[scale=0.5,baseline={([yshift=-0.65ex] current bounding box.center) }]
        \begin{scope}[rotate around={-45:(0,0)}]
            \draw[very thick, draw=red] (-1.2, 1.2) -- (1.2, -1.2);
            \draw[very thick, draw=red] (1.2, 1.2) -- (-1.2, -1.2);
            \draw[very thick, draw=red] (0, 1.2) -- (0, -1.2);
            \draw[thick, fill=tensorcolor, rounded corners=2pt] (-0.5*\singledx-0.6,-0.6) rectangle (0.5*\singledx+0.6,+0.6);
            \HollowRedDot{0,.-1.2}
            \draw (0,0) node {\small $U_C$};
            \draw (-1.4, -1.4) node {\small $i_v$};
            \draw (-1.4, 1.4) node {\small $j_v$};
            \draw (1.4, 1.4) node {\small $k_v$};
            \draw (1.4, -1.4) node {\small $l_v$};
            \draw (0.4,2) node {\small $i_p, j_p, k_p, l_p$};
            \RedArrow{0,1.}{0.2}{0}
        \end{scope}
        \RedArrow{-1.2,0}{0.2}{1}
        \RedArrow{1.4,0}{0.2}{1}
        \RedArrow{0,1.4}{0.2}{0}
        \RedArrow{0,-1.2}{0.2}{0}
    \end{tikzpicture}
    \;.
    \end{align}
    Here we group the \(i_p,j_p,k_p,l_p\) legs into a single leg. \(U_C: \mathbb{C}^{D^6} \rightarrow \mathbb{C}^{D^6}\) is a Clifford unitary satisfying the same symmetry constrains as in (ii), and \(\ket{\psi}\) is a (possibly magic) pure state.
\end{enumerate}
\end{theorem}
The proof is entirely analogous to the case of MPS, and Corollary~\ref{Up_PP_relation}
carries to PEPS as well. Importantly, MF PEPS also exhibit the Clifford sequential generation structure as demonstrated in Eqs.~(\ref{mps_op_ex1},\ref{mps_op_ex2},\ref{mps_op_ex3}), and which allows for the efficient evaluation of the Weyl-Heisenberg observable. Since the evaluation of observables with large support in PEPS is computationally inefficient even with isoTNS~\cite{malz2024computational}, the Clifford sequential generation structure provides a significant computational simplification in two dimensions.

Lastly, the orthogonality center which can be considered as the ``starting point" of pushing operators need not be at the lower-left corner but can be located in the bulk.
For example, we can glue PEPS with isometry direction pointing at all four corners as shown below. The green lines partition the PEPS into regions with the intersection being the orthogonality center. A defect operator can be pushed to any one of the four corners, depending on which region it belongs to.
\begin{align}
    \begin{tikzpicture}[scale=0.5,baseline={([yshift=-0.65ex] current bounding box.center) }]
        \begin{scope}[shift={(0,0)}]
            \GPTensor{0,0}{1.6}{.6}{A}{0};
            \RedArrow{-1,0}{0.2}{3}
            \RedArrow{0.8,0}{0.2}{3}
            \RedArrow{0,0.8}{0.2}{2}
            \RedArrow{0,-1}{0.2}{2}
        \end{scope}
        \begin{scope}[shift={(\singledx+0.8,0)}]
            \GPTensor{0,0}{1.6}{.6}{A}{0};
            \RedArrow{-0.8,0}{0.2}{1}
            \RedArrow{1,0}{0.2}{1}
            \RedArrow{0,0.8}{0.2}{2}
            \RedArrow{0,-1}{0.2}{2}
        \end{scope}
        \begin{scope}[shift={(2*\singledx+1.6,0)}]
            \GPTensor{0,0}{1.6}{.6}{A}{0};
            \RedArrow{-0.8,0}{0.2}{1}
            \RedArrow{1,0}{0.2}{1}
            \RedArrow{0,0.8}{0.2}{2}
            \RedArrow{0,-1}{0.2}{2}
        \end{scope}
        \begin{scope}[shift={(0,\singledx+0.8)}]
            \GPTensor{0,0}{1.6}{.6}{A}{0};
            \RedArrow{-1,0}{0.2}{3}
            \RedArrow{0.8,0}{0.2}{3}
            \RedArrow{0,1.0}{0.2}{0}
            \RedArrow{0,-0.8}{0.2}{0}
        \end{scope}
        \begin{scope}[shift={(0,2*\singledx+1.6)}]
            \GPTensor{0,0}{1.6}{.6}{A}{0};
            \RedArrow{-1,0}{0.2}{3}
            \RedArrow{0.8,0}{0.2}{3}
            \RedArrow{0,1.0}{0.2}{0}
            \RedArrow{0,-0.8}{0.2}{0}
        \end{scope}
        \begin{scope}[shift={(\singledx+0.8,\singledx+0.8)}]
            \GPTensor{0,0}{1.6}{.6}{A}{0};
            \RedArrow{-0.8,0}{0.2}{1}
            \RedArrow{1,0}{0.2}{1}
            \RedArrow{0,1}{0.2}{0}
            \RedArrow{0,-0.8}{0.2}{0}
        \end{scope}
        \begin{scope}[shift={(2*\singledx+1.6,\singledx+0.8)}]
            \GPTensor{0,0}{1.6}{.6}{A}{0};
            \RedArrow{-0.8,0}{0.2}{1}
            \RedArrow{1,0}{0.2}{1}
            \RedArrow{0,1}{0.2}{0}
            \RedArrow{0,-0.8}{0.2}{0}
        \end{scope}
        \begin{scope}[shift={(\singledx+0.8,2*\singledx+1.6)}]
            \GPTensor{0,0}{1.6}{.6}{A}{0};
            \RedArrow{-0.8,0}{0.2}{1}
            \RedArrow{1,0}{0.2}{1}
            \RedArrow{0,1}{0.2}{0}
            \RedArrow{0,-0.8}{0.2}{0}
        \end{scope}
        \begin{scope}[shift={(2*\singledx+1.6,2*\singledx+1.6)}]
            \GPTensor{0,0}{1.6}{.6}{A}{0};
            \RedArrow{-0.8,0}{0.2}{1}
            \RedArrow{1,0}{0.2}{1}
            \RedArrow{0,1}{0.2}{0}
            \RedArrow{0,-0.8}{0.2}{0}
        \end{scope}
        \draw[very thick, dashed, draw=gtensorcolor] (0.5*\singledx+0.4, -1.6) -- (0.5*\singledx+0.4, 3*\singledx+1.6);
        \draw[very thick, dashed, draw=gtensorcolor] (-1.6, 0.5*\singledx+0.4) -- (3*\singledx+1.6, 0.5*\singledx+0.4);
\end{tikzpicture}
\end{align}


\subsection{Topological Order}
An important class of MF PEPS are those with topological order. This is one instance where MF has a unique advantage because it is known that topological order cannot be prepared by a constant-depth unitary circuits~\cite{kitaev2003fault,bravyi2006lieb,aharonov2018quantum}, but can be using constant-depth MF.
Therefore, in this section we begin by imposing an MF symmetry that resemble the topological order. This symmetry is insufficient to necessarily imply topological states, so we then impose an additional symmetry which further constrains the state but allows to connect with topological order. We probe the latter via the spectrum of the transfer matrix and provide analytical solutions to understand its degeneracy.


\subsubsection{Symmetry constraints}
Our first step is to impose the following stronger symmetry constraint where we have a single MF basis $\mathcal B = \{ P_i \}$ and we allow the defect operators to move in any direction:
\begin{subequations}\label{topo_sym}
\begin{align}\label{topo_sym1}
\begin{tikzpicture}[scale=0.5,baseline={([yshift=-0.65ex] current bounding box.center) }]
        \DTensor{-\singledx,0}{1.}{.6}{\small $P_i$}{0};
        \GPTensor{0,0}{1.3}{.6}{\small $A$}{0};
        \RedArrow{-0.8,0}{0.2}{1}
        \RedArrow{1.0,0}{0.2}{1}
        \RedArrow{0,1.0}{0.2}{0}
        \RedArrow{0,-0.8}{0.2}{0}
\end{tikzpicture}
=\begin{tikzpicture}[scale=0.5,baseline={([yshift=-3ex] current bounding box.center) }]
        \GPTensor{0,0}{1.3}{.6}{\small $A$}{0};
        \DTensor{0,\singledx}{1.}{.6}{\small $P_i$}{1};
        \begin{scope}[rotate around={45:(1.6,1.6)}]
        \UTensor{1.6,1.6}{1.3}{.6}{\small $U_{1}$}{2};
        \end{scope}
        \RedArrow{-0.8,0}{0.2}{1}
        \RedArrow{1.0,0}{0.2}{1}
        \RedArrow{0,1.0}{0.2}{0}
        \RedArrow{0,-0.8}{0.2}{0}
\end{tikzpicture}\\ \label{topo_sym2}
=\begin{tikzpicture}[scale=0.5,baseline={([yshift=-2.65ex] current bounding box.center) }]
        \GPTensor{0,0}{1.3}{.6}{\small $A$}{0};
        \DTensor{\singledx,0}{1.}{.6}{\small $P_i$}{0};
        \begin{scope}[rotate around={45:(1.6,1.6)}]
        \UTensor{1.6,1.6}{1.3}{.6}{\small $U_{2}$}{2};
        \end{scope}
        \RedArrow{-0.8,0}{0.2}{1}
        \RedArrow{1.0,0}{0.2}{1}
        \RedArrow{0,1.0}{0.2}{0}
        \RedArrow{0,-0.8}{0.2}{0}
\end{tikzpicture}
=\begin{tikzpicture}[scale=0.5,baseline={([yshift=-0ex] current bounding box.center) }]
        \GPTensor{0,0}{1.3}{.6}{\small $A$}{0};
        \DTensor{0,-\singledx}{1.}{.6}{\small $P_i$}{1};
        \begin{scope}[rotate around={45:(1.6,1.6)}]
        \UTensor{1.6,1.6}{1.3}{.6}{\small $U_{3}$}{2};
        \end{scope}
        \RedArrow{-0.8,0}{0.2}{1}
        \RedArrow{1.0,0}{0.2}{1}
        \RedArrow{0,1.0}{0.2}{0}
        \RedArrow{0,-0.8}{0.2}{0}
\end{tikzpicture}
\end{align}
\end{subequations}
Here we again label the virtual legs with arrows to specify the order of matrix multiplication. 
However, the above symmetry is insufficient to necessarily imply topological order. We discuss why this is the case in Appendix~\ref{topo_insufficiency}. Following Ref.~\cite{schuch2010peps}, we consider a subgroup \(\{M\}\) of the MF group and impose the following constraint.
\begin{align}\label{topo_more_sym}
\begin{tikzpicture}[scale=0.5,baseline={([yshift=-0.65ex] current bounding box.center) }]
        \GPTensor{0,0}{1.2}{.6}{\small $A$}{0};
        \RedArrow{-0.8,0}{0.2}{1}
        \RedArrow{1.0,0}{0.2}{1}
        \RedArrow{0,1.0}{0.2}{0}
        \RedArrow{0,-0.8}{0.2}{0}
\end{tikzpicture}
=e^{i\phi}\begin{tikzpicture}[scale=0.5,baseline={([yshift=-0.65ex] current bounding box.center) }]
        \DTensor{-\singledx,0}{1.2}{.6}{\small $M$}{0};
        \DTensor{\singledx,0}{1.2}{.6}{\small $M^\dag$}{0};
        \DTensor{0,-\singledx}{1.2}{.6}{\small $M$}{1};
        \DTensor{0,\singledx}{1.2}{.6}{\small $M^\dag$}{1};
        \GPTensor{0,0}{1.2}{.6}{\small $A$}{0};
        \RedArrow{-0.8,0}{0.2}{1}
        \RedArrow{1.0,0}{0.2}{1}
        \RedArrow{0,1.0}{0.2}{0}
        \RedArrow{0,-0.8}{0.2}{0}
\end{tikzpicture}
\;, \forall M
\end{align}
Here \(\phi\) is some phase. Note that there is no error correction operator acting on the physical leg so it is stricter than MF symmetry. Also the above symmetry is satisfied by \(Q\) as well because \(V^\dag A = Q\).

Similar to the case of MPS with the SPT-type symmetry [Eq.~\eqref{full_push_thru_sym}], we provide an analytical solution to the states satisfying the above symmetry when the MF group is Abelian.
\begin{theorem}\label{topo_structural_theorem}
Let \(\{P_i\}\) be a projective representation of an Abelian MF group, then under Eqs.~(\ref{topo_sym}),
\begin{align}\label{topo_structural_theorem_eq}
Q
=\sum_{i} \alpha_i \;
\begin{tikzpicture}[scale=0.5,baseline={([yshift=-0.65ex] current bounding box.center) }]
        \PTensor{-\singledx,0}{1.5}{.6}{\small $P_i^\dag$}{3};
        \PTensor{\singledx,0}{1.5}{.6}{\small $P_i$}{1};
        \PTensor{0,-\singledx}{1.5}{.6}{\small $P_i^\dag$}{2};
        \PTensor{0,\singledx}{1.5}{.6}{\small $P_i$}{0};
        \RedArrow{-0.8-\singledx,0}{0.2}{1}
        \RedArrow{1.0+\singledx,0}{0.2}{1}
        \RedArrow{0,1.0+\singledx}{0.2}{0}
        \RedArrow{0,-0.8-\singledx}{0.2}{0}
\end{tikzpicture}
\;.
\end{align}
In addition, under Eqs. (\ref{topo_more_sym}), \(\alpha_{P_i}=e^{i\phi}\alpha_{M P_i}\) in Eq.~(\ref{topo_structural_theorem_eq}) for any \(M\) in the subgroup and \(n \phi = 0  \mod  2\pi\) where \(n\) is the order of \(M\).
\end{theorem}
\begin{proof}
    The proof is entirely analogous to the case of MPS. We expand \(Q\) in the MF basis:
    \begin{align}
    Q
    =\sum_{ijkl} \alpha_{ijkl} \;
    \begin{tikzpicture}[scale=0.5,baseline={([yshift=-0.65ex] current bounding box.center) }]
            \PTensor{-\singledx,0}{1.5}{.6}{\small $P_i^\dag$}{3};
            \PTensor{\singledx,0}{1.5}{.6}{\small $P_k$}{1};
            \PTensor{0,-\singledx}{1.5}{.6}{\small $P_l^\dag$}{2};
            \PTensor{0,\singledx}{1.5}{.6}{\small $P_j$}{0};
            \RedArrow{-0.8-\singledx,0}{0.2}{1}
            \RedArrow{1.0+\singledx,0}{0.2}{1}
            \RedArrow{0,1.0+\singledx}{0.2}{0}
            \RedArrow{0,-0.8-\singledx}{0.2}{0}
    \end{tikzpicture}
    \;.
\end{align}
We now apply the symmetry constraint on legs labeled by \(i\) and \(j\).
\begin{widetext}
\begin{align}
    Q
    =\sum_{ijkl} \alpha_{ijkl} \begin{tikzpicture}[scale=0.38,baseline={([yshift=-4.65ex] current bounding box.center) }]
             \PTensor{-2*\singledx,0}{1.}{.6}{\small $P_i^\dag$}{3};
            \UTensor{-\singledx,0}{1.}{.6}{\small $P_m^\dag$}{0};
            \DTensor{-3*\singledx,0}{1.2}{.6}{\small $P_m$}{0};
            \PTensor{\singledx,0}{1.5}{.6}{\small $P_k$}{1};
            \PTensor{0,-\singledx}{1.5}{.6}{\small $P_l^\dag$}{2};
            \PTensor{0,2*\singledx}{1.}{.6}{\small $P_j$}{0};
            \UTensor{0,\singledx}{1.}{.6}{\small $P_m$}{1};
            \DTensor{0,3*\singledx}{1.2}{.6}{\small $P_m^\dag$}{1};
            \RedArrow{-0.8-3*\singledx,0}{0.2}{1}
            \RedArrow{1.0+\singledx,0}{0.2}{1}
            \RedArrow{0,1.0+3*\singledx}{0.2}{0}
            \RedArrow{0,-0.8-\singledx}{0.2}{0}
    \end{tikzpicture}
    =\sum_{ijkl} \frac{\omega(j,m)}{\omega(i,m)} \alpha_{ijkl} \begin{tikzpicture}[scale=0.38,baseline={([yshift=-4.65ex] current bounding box.center) }]
            \PTensor{-1*\singledx,0}{1.}{.6}{\small $P_i^\dag$}{3};
            \UTensor{-2*\singledx,0}{1.}{.6}{\small $P_m^\dag$}{0};
            \DTensor{-3*\singledx,0}{1.2}{.6}{\small $P_m$}{0};
            \PTensor{\singledx,0}{1.5}{.6}{\small $P_k$}{1};
            \PTensor{0,-\singledx}{1.5}{.6}{\small $P_l^\dag$}{2};
            \PTensor{0,1*\singledx}{1.}{.6}{\small $P_j$}{0};
            \UTensor{0,2*\singledx}{1.}{.6}{\small $P_m$}{1};
            \DTensor{0,3*\singledx}{1.2}{.6}{\small $P_m^\dag$}{1};
            \RedArrow{-0.8-3*\singledx,0}{0.2}{1}
            \RedArrow{1.0+\singledx,0}{0.2}{1}
            \RedArrow{0,1.0+3*\singledx}{0.2}{0}
            \RedArrow{0,-0.8-\singledx}{0.2}{0}
    \end{tikzpicture}
    =\sum_{ijkl} \frac{\omega(j,m)}{\omega(i,m)} \alpha_{ijkl} \begin{tikzpicture}[scale=0.38,baseline={([yshift=-0.65ex] current bounding box.center) }]
            \PTensor{-\singledx,0}{1.5}{.6}{\small $P_i^\dag$}{3};
            \PTensor{\singledx,0}{1.5}{.6}{\small $P_k$}{1};
            \PTensor{0,-\singledx}{1.5}{.6}{\small $P_l^\dag$}{2};
            \PTensor{0,\singledx}{1.5}{.6}{\small $P_j$}{0};
            \RedArrow{-0.8-\singledx,0}{0.2}{1}
            \RedArrow{1.0+\singledx,0}{0.2}{1}
            \RedArrow{0,1.0+\singledx}{0.2}{0}
            \RedArrow{0,-0.8-\singledx}{0.2}{0}
            \end{tikzpicture}
\end{align}
Again we the equation has to be satisfied term-by-term, so \(\omega(i,m)^{-1}\omega(j,m)=1, \forall m\), leaving the only choice to be \(i=j\). One can repeat the same calculation for legs \(i\) and \(k\) to get \(i=k\). Finally, we repeat the same calculation for legs \(i\) and \(l\).
%
\begin{align}
    Q
    =\sum_{ijkl} \alpha_{ijkl} \begin{tikzpicture}[scale=0.38,baseline={([yshift=-0.65ex] current bounding box.center) }]
             \PTensor{-2*\singledx,0}{1.}{.6}{\small $P_i^\dag$}{3};
            \UTensor{-\singledx,0}{1.}{.6}{\small $P_m^\dag$}{0};
            \DTensor{-3*\singledx,0}{1.2}{.6}{\small $P_m$}{0};
            \PTensor{\singledx,0}{1.5}{.6}{\small $P_k$}{1};
            \PTensor{0,-2*\singledx}{1.5}{.6}{\small $P_l^\dag$}{2};
            \UTensor{0,-\singledx}{1.}{.6}{\small $P_m$}{1};
            \DTensor{0,-3*\singledx}{1.2}{.6}{\small $P_m^\dag$}{1};
            \PTensor{0,\singledx}{1.5}{.6}{\small $P_j$}{0};
            \RedArrow{-0.8-3*\singledx,0}{0.2}{1}
            \RedArrow{1.0+\singledx,0}{0.2}{1}
            \RedArrow{0,1.0+\singledx}{0.2}{0}
            \RedArrow{0,-0.8-3*\singledx}{0.2}{0}
    \end{tikzpicture}
    =\sum_{ijkl} \alpha_{ijkl}
    \begin{tikzpicture}[scale=0.38,baseline={([yshift=-0.65ex] current bounding box.center) }]
             \PTensor{-\singledx,0}{1.}{.6}{\small $P_i^\dag$}{3};
            \UTensor{-2*\singledx,0}{1.}{.6}{\small $P_m^\dag$}{0};
            \DTensor{-3*\singledx,0}{1.2}{.6}{\small $P_m$}{0};
            \PTensor{\singledx,0}{1.5}{.6}{\small $P_k$}{1};
            \PTensor{0,-\singledx}{1.5}{.6}{\small $P_l^\dag$}{2};
            \UTensor{0,-2*\singledx}{1.}{.6}{\small $P_m$}{1};
            \DTensor{0,-3*\singledx}{1.2}{.6}{\small $P_m^\dag$}{1};
            \PTensor{0,\singledx}{1.5}{.6}{\small $P_j$}{0};
            \RedArrow{-0.8-3*\singledx,0}{0.2}{1}
            \RedArrow{1.0+\singledx,0}{0.2}{1}
            \RedArrow{0,1.0+\singledx}{0.2}{0}
            \RedArrow{0,-0.8-3*\singledx}{0.2}{0}
    \end{tikzpicture}
    =\sum_{ijkl} \frac{\omega(l,m)}{\omega(i,m)} \alpha_{ijkl} \begin{tikzpicture}[scale=0.38,baseline={([yshift=-0.65ex] current bounding box.center) }]
            \PTensor{-\singledx,0}{1.5}{.6}{\small $P_i^\dag$}{1};
            \PTensor{\singledx,0}{1.5}{.6}{\small $P_k$}{1};
            \PTensor{0,-\singledx}{1.5}{.6}{\small $P_l^\dag$}{0};
            \PTensor{0,\singledx}{1.5}{.6}{\small $P_j$}{0};
            \RedArrow{-0.8-\singledx,0}{0.2}{1}
            \RedArrow{1.0+\singledx,0}{0.2}{1}
            \RedArrow{0,1.0+\singledx}{0.2}{0}
            \RedArrow{0,-0.8-\singledx}{0.2}{0}
            \end{tikzpicture}
\end{align}
\end{widetext}
Again, satisfying the equation requires that \(\omega(i,m)^{-1}\omega(l,m)=1, \forall m\), leaving the only choice to be \(i=m\). Hence, the only non-trivial terms are those with \(i=j=k=l\).
\end{proof}

We discuss some examples. 
Consider the tensor
\begin{align}\label{toric_code_ex}
    Q
    =\sum_{i=0}^{D-1} \begin{tikzpicture}[scale=0.5,baseline={([yshift=-0.65ex] current bounding box.center) }]
            \PTensor{-\singledx,0}{1.5}{.6}{\small $X^{-i}$}{3};
            \PTensor{\singledx,0}{1.5}{.6}{\small $X^{i}$}{1};
            \PTensor{0,-\singledx}{1.5}{.6}{\small $X^{-i}$}{2};
            \PTensor{0,\singledx}{1.5}{.6}{\small $X^{i}$}{0};
            \RedArrow{-0.8-\singledx,0}{0.2}{1}
            \RedArrow{1.0+\singledx,0}{0.2}{1}
            \RedArrow{0,1.0+\singledx}{0.2}{0}
            \RedArrow{0,-0.8-\singledx}{0.2}{0}
    \end{tikzpicture}
\end{align}
corresponding to a $\mathbb Z_D$ quantum double model, up to an isometry in the physical space. As a modification, one can introduce a phase to each term in the summation,
\begin{align}
    Q
    =\sum_{i=0}^{D-1} \omega_k^i \begin{tikzpicture}[scale=0.5,baseline={([yshift=-0.65ex] current bounding box.center) }]
            \PTensor{-\singledx,0}{1.5}{.6}{\small $X^{\shortminus i}$}{3};
            \PTensor{\singledx,0}{1.5}{.6}{\small $X^{i}$}{1};
            \PTensor{0,-\singledx}{1.5}{.6}{\small $X^{\shortminus i}$}{2};
            \PTensor{0,\singledx}{1.5}{.6}{\small $X^{i}$}{0};
            \RedArrow{-0.8-\singledx,0}{0.2}{1}
            \RedArrow{1.0+\singledx,0}{0.2}{1}
            \RedArrow{0,1.0+\singledx}{0.2}{0}
            \RedArrow{0,-0.8-\singledx}{0.2}{0}
    \end{tikzpicture}
    \;,
\end{align}
where \(\omega_k=e^{i 2 k \pi / D }, k=0,1,...,D-1\). This corresponds to a \(\mathbb Z_D\) quantum double model with an electric charge \(k\). To see that, we glue two tensors with charge \(k\) and \(-k\) and apply a local unitary \(Z^k \otimes Z^{-k}\) to annihilate the charge:
\begin{align}
&\sum_{i,j=0}^{D-1}  \omega_k^i \omega_{-k}^{j}  \begin{tikzpicture}[scale=0.45,baseline={([yshift=-0.65ex] current bounding box.center) }]
    \begin{scope}[shift={(0,0)}]
            \PTensor{-2*\singledx,0}{1.}{.6}{\small $X^{\shortminus i}$}{3};
            \PTensor{\singledx,0}{1.5}{.6}{\small $X^{i}$}{1};
            \PTensor{0,-\singledx}{1.5}{.6}{\small $X^{\shortminus i}$}{2};
            \PTensor{0,\singledx}{1.5}{.6}{\small $X^{i}$}{0};
            \UTensor{-\singledx,0}{1.}{.6}{\small $Z^k$}{0}
            \RedArrow{-0.8-\singledx,0}{0.2}{1}
            \RedArrow{1.0+\singledx,0}{0.2}{1}
            \RedArrow{0,1.0+\singledx}{0.2}{0}
            \RedArrow{0,-0.8-\singledx}{0.2}{0}
    \end{scope}
    \begin{scope}[shift={(-5*\singledx,0)}]
            \PTensor{-\singledx,0}{1.5}{.6}{\small $X^{\shortminus j}$}{3};
            \PTensor{2*\singledx,0}{1.}{.6}{\small $X^{j}$}{1};
            \PTensor{0,-\singledx}{1.5}{.6}{\small $X^{\shortminus j}$}{2};
            \PTensor{0,\singledx}{1.5}{.6}{\small $X^{j}$}{0};
            \UTensor{1*\singledx,0}{1.}{.6}{\small $Z^{\shortminus k}$}{0}
            \RedArrow{-0.8-\singledx,0}{0.2}{1}
            \RedArrow{1.0+\singledx,0}{0.2}{1}
            \RedArrow{0,1.0+\singledx}{0.2}{0}
            \RedArrow{0,-0.8-\singledx}{0.2}{0}
    \end{scope}
\end{tikzpicture} \\
&=\sum_{i,j=0}^{D-1} \begin{tikzpicture}[scale=0.45,baseline={([yshift=-0.65ex] current bounding box.center) }]
    \begin{scope}[shift={(0,0)}]
            \PTensor{-1*\singledx,0}{1.}{.6}{\small $X^{\shortminus i}$}{3};
            \PTensor{\singledx,0}{1.5}{.6}{\small $X^{i}$}{1};
            \PTensor{0,-\singledx}{1.5}{.6}{\small $X^{\shortminus i}$}{2};
            \PTensor{0,\singledx}{1.5}{.6}{\small $X^{i}$}{0};
            \UTensor{-2*\singledx,0}{1.}{.6}{\small $Z^{k}$}{0}
            \RedArrow{-0.8-\singledx,0}{0.2}{1}
            \RedArrow{1.0+\singledx,0}{0.2}{1}
            \RedArrow{0,1.0+\singledx}{0.2}{0}
            \RedArrow{0,-0.8-\singledx}{0.2}{0}
    \end{scope}
    \begin{scope}[shift={(-5*\singledx,0)}]
            \PTensor{-\singledx,0}{1.5}{.6}{\small $X^{\shortminus j}$}{3};
            \PTensor{\singledx,0}{1.}{.6}{\small $X^{j}$}{1};
            \PTensor{0,-\singledx}{1.5}{.6}{\small $X^{\shortminus j}$}{2};
            \PTensor{0,\singledx}{1.5}{.6}{\small $X^{j}$}{0};
            \UTensor{2*\singledx,0}{1.}{.6}{\small $Z^{\shortminus k}$}{0}
            \RedArrow{-0.8-\singledx,0}{0.2}{1}
            \RedArrow{1.0+\singledx,0}{0.2}{1}
            \RedArrow{0,1.0+\singledx}{0.2}{0}
            \RedArrow{0,-0.8-\singledx}{0.2}{0}
    \end{scope}
\end{tikzpicture}\\
&=\sum_{i,j=0}^{D-1} \begin{tikzpicture}[scale=0.45,baseline={([yshift=-0.65ex] current bounding box.center) }]
    \begin{scope}[shift={(0,0)}]
            \PTensor{-1*\singledx,0}{1.}{.6}{\small $X^{\shortminus i}$}{3};
            \PTensor{\singledx,0}{1.5}{.6}{\small $X^{i}$}{1};
            \PTensor{0,-\singledx}{1.5}{.6}{\small $X^{\shortminus i}$}{2};
            \PTensor{0,\singledx}{1.5}{.6}{\small $X^{i}$}{0};
            \RedArrow{-0.8-\singledx,0}{0.2}{1}
            \RedArrow{1.0+\singledx,0}{0.2}{1}
            \RedArrow{0,1.0+\singledx}{0.2}{0}
            \RedArrow{0,-0.8-\singledx}{0.2}{0}
    \end{scope}
    \begin{scope}[shift={(-3*\singledx,0)}]
            \PTensor{-\singledx,0}{1.5}{.6}{\small $X^{\shortminus j}$}{3};
            \PTensor{\singledx,0}{1.}{.6}{\small $X^{j}$}{1};
            \PTensor{0,-\singledx}{1.5}{.6}{\small $X^{\shortminus j}$}{2};
            \PTensor{0,\singledx}{1.5}{.6}{\small $X^{j}$}{0};
            \RedArrow{-0.8-\singledx,0}{0.2}{1}
            \RedArrow{1.0+\singledx,0}{0.2}{1}
            \RedArrow{0,1.0+\singledx}{0.2}{0}
            \RedArrow{0,-0.8-\singledx}{0.2}{0}
    \end{scope}
\end{tikzpicture}
\;.
\end{align}
In the second line above we commute \(Z^k\) and \(Z^{-k}\) through \(X^{-i}\) and \(X^i\) to cancel out the phase, and in the third line we annihilate \(Z^k\) and \(Z^{-k}\) on the virtual legs.

The two above examples also satisfy the symmetry (\ref{topo_more_sym}) with \(M=Z\). In the charge-free case, \(\phi=0\), while in the case with charge \(k\), \(\phi=2 \pi k / D\). This is the symmetry that allows the Wilson line to fluctuate freely. As another example, if we take \(\{M\}\) to be the entire MF group and set \(\phi=0\), then the resulting state becomes
\begin{align}\label{full_sym_ex}
Q
=\sum_{i} \begin{tikzpicture}[scale=0.5,baseline={([yshift=-0.65ex] current bounding box.center) }]
        \PTensor{-\singledx,0}{1.5}{.6}{\small $P_i^\dag$}{3};
        \PTensor{\singledx,0}{1.5}{.6}{\small $P_i$}{1};
        \PTensor{0,-\singledx}{1.5}{.6}{\small $P_i^\dag$}{2};
        \PTensor{0,\singledx}{1.5}{.6}{\small $P_i$}{0};
        \RedArrow{-0.8-\singledx,0}{0.2}{1}
        \RedArrow{1.0+\singledx,0}{0.2}{1}
        \RedArrow{0,1.0+\singledx}{0.2}{0}
        \RedArrow{0,-0.8-\singledx}{0.2}{0}
\end{tikzpicture}
\end{align}
In other words, we have created a tensor which takes the equal sum of all the terms in Eq.~(\ref{topo_structural_theorem_eq}).

\subsubsection{Spectrum of transfer matrix}
We now try to probe the possible topological order under the symmetries Eqs.~(\ref{topo_sym},\ref{topo_more_sym}). Unfortunately, currently there does not exist a universal method to determine the topological order in PEPS \cite{schuch2013topological}. Therefore, we seek for signatures of the topological order instead. One way to probe the topological order of PEPS is to calculate the spectrum of its transfer matrix \cite{cirac2011entanglement,schuch2013topological}. We define \(E=A^\dag A\), or pictorially
\begin{align}
    \begin{tikzpicture}[scale=0.5,baseline={([yshift=-0.65ex] current bounding box.center) }]
            \EPTensor{0,0}{1.3}{.6}{\small $E$}{0};
            \RedArrow{-0.8,-0.2}{0.2}{1}
            \RedArrow{1.0,-0.2}{0.2}{1}
            \RedArrow{-1.,0.2}{0.2}{3}
            \RedArrow{0.8,0.2}{0.2}{3}
            \RedArrow{-0.2,1.0}{0.2}{0}
            \RedArrow{-0.2,-0.8}{0.2}{0}
            \RedArrow{0.2,0.8}{0.2}{2}
            \RedArrow{0.2,-1.}{0.2}{2}
    \end{tikzpicture}
    =\begin{tikzpicture}[scale=0.5,baseline={([yshift=-0.65ex] current bounding box.center) }]
            \GPTensor{0,0}{1.3}{.6}{\small $A$}{0};
            \RedArrow{-0.8,0}{0.2}{1}
            \RedArrow{1.0,0}{0.2}{1}
            \RedArrow{0,1.0}{0.2}{0}
            \RedArrow{0,-0.8}{0.2}{0}
            \GPTensor{\singledx,\singledx}{1.3}{.6}{\small $A^\dag$}{2};
            \begin{scope}[shift={(\singledx,\singledx)}]
                \RedArrow{-1,0}{0.2}{3}
                \RedArrow{0.8,0}{0.2}{3}
                \RedArrow{0,0.8}{0.2}{2}
                \RedArrow{0,-1}{0.2}{2}
            \end{scope}
    \end{tikzpicture}
    \;.
\end{align}
The transfer matrix \(\mathbb{T}\) is defined as a chain of \(E\) contracted sideways. We take the periodic boundary condition and take the large system limit.
\begin{align}\label{transfer_matrix}
\mathbb{T}=\textcolor{red}{...}
\begin{tikzpicture}[scale=0.5,baseline={([yshift=-0.65ex] current bounding box.center) }]
\begin{scope}[shift={(0,0)}]
        \EPTensor{0,0}{1.2}{.6}{\small $E$}{0};
        \RedArrow{-0.8,-0.2}{0.2}{1}
        \RedArrow{1.0,-0.2}{0.2}{1}
        \RedArrow{-1.,0.2}{0.2}{3}
        \RedArrow{0.8,0.2}{0.2}{3}
        \RedArrow{-0.2,1.0}{0.2}{0}
        \RedArrow{-0.2,-0.8}{0.2}{0}
        \RedArrow{0.2,0.8}{0.2}{2}
        \RedArrow{0.2,-1.}{0.2}{2}
\end{scope}
\begin{scope}[shift={(\singledx,0)}]
        \EPTensor{0,0}{1.2}{.6}{\small $E$}{0};
        \RedArrow{-0.8,-0.2}{0.2}{1}
        \RedArrow{1.0,-0.2}{0.2}{1}
        \RedArrow{-1.,0.2}{0.2}{3}
        \RedArrow{0.8,0.2}{0.2}{3}
        \RedArrow{-0.2,1.0}{0.2}{0}
        \RedArrow{-0.2,-0.8}{0.2}{0}
        \RedArrow{0.2,0.8}{0.2}{2}
        \RedArrow{0.2,-1.}{0.2}{2}
\end{scope}
\begin{scope}[shift={(2*\singledx,0)}]
        \EPTensor{0,0}{1.2}{.6}{\small $E$}{0};
        \RedArrow{-0.8,-0.2}{0.2}{1}
        \RedArrow{1.0,-0.2}{0.2}{1}
        \RedArrow{-1.,0.2}{0.2}{3}
        \RedArrow{0.8,0.2}{0.2}{3}
        \RedArrow{-0.2,1.0}{0.2}{0}
        \RedArrow{-0.2,-0.8}{0.2}{0}
        \RedArrow{0.2,0.8}{0.2}{2}
        \RedArrow{0.2,-1.}{0.2}{2}
\end{scope}
\begin{scope}[shift={(3*\singledx,0)}]
        \EPTensor{0,0}{1.2}{.6}{\small $E$}{0};
        \RedArrow{-0.8,-0.2}{0.2}{1}
        \RedArrow{1.0,-0.2}{0.2}{1}
        \RedArrow{-1.,0.2}{0.2}{3}
        \RedArrow{0.8,0.2}{0.2}{3}
        \RedArrow{-0.2,1.0}{0.2}{0}
        \RedArrow{-0.2,-0.8}{0.2}{0}
        \RedArrow{0.2,0.8}{0.2}{2}
        \RedArrow{0.2,-1.}{0.2}{2}
\end{scope}
\end{tikzpicture}
\textcolor{red}{...}
\end{align}
It is shown in Ref.~\cite{schuch2013topological} that the degeneracy in the largest eigenvalue of \(\mathbb{T}\) constitutes a signature that the PEPS is in a non-trivial phase (symmetry-broken or topological). Using Theorem ~\ref{topo_structural_theorem}, we can easily evaluate the transfer matrix.
\begin{theorem}
Under the symmetry constraints Eqs.~(\ref{topo_sym}), the transfer matrix \(\mathbb{T}\) can be written as
\begin{align}\label{transfer_matrix_structure1}
    \mathbb{T} &= \sum_i t_{P_i} \begin{tikzpicture}[scale=0.5,baseline={([yshift=-0.65ex] current bounding box.center) }]
        \PTensor{0,-\singledx}{1.2}{.6}{\small $P_i^\dag$}{2};
        \PTensor{0,\singledx}{1.2}{.6}{\small $P_i$}{0};
        \PTensor{\singledx,-\singledx}{1.2}{.6}{\small $P_i^\dag$}{2};
        \PTensor{\singledx,\singledx}{1.2}{.6}{\small $P_i$}{0};
        \PTensor{2*\singledx,-\singledx}{1.2}{.6}{\small $P_i^\dag$}{2};
        \PTensor{2*\singledx,\singledx}{1.2}{.6}{\small $P_i$}{0};
        \RedArrow{0,1.0+\singledx}{0.2}{0}
        \RedArrow{0,-0.8+\singledx}{0.2}{0}
        \RedArrow{0,1.0-\singledx}{0.2}{0}
        \RedArrow{0,-0.8-\singledx}{0.2}{0}
        \RedArrow{\singledx,1.0+\singledx}{0.2}{0}
        \RedArrow{\singledx,-0.8+\singledx}{0.2}{0}
        \RedArrow{\singledx,1.0-\singledx}{0.2}{0}
        \RedArrow{\singledx,-0.8-\singledx}{0.2}{0}
        \RedArrow{2*\singledx,1.0+\singledx}{0.2}{0}
        \RedArrow{2*\singledx,-0.8+\singledx}{0.2}{0}
        \RedArrow{2*\singledx,1.0-\singledx}{0.2}{0}
        \RedArrow{2*\singledx,-0.8-\singledx}{0.2}{0}
        \begin{scope}[shift={(0,0)}]
            \draw[very thick, draw=red] (0,-\singledx+1.2) -- (0.8,-\singledx+1.2) --  (0.8,-\singledx-0.8)  -- (0.4,-\singledx-0.8) -- (0.4,-\singledx-1.2);
        \end{scope}
        \begin{scope}[shift={(\singledx,0)}]
            \draw[very thick, draw=red] (0,-\singledx+1.2) -- (0.8,-\singledx+1.2) --  (0.8,-\singledx-0.8)  -- (0.4,-\singledx-0.8) -- (0.4,-\singledx-1.2);
        \end{scope}
        \begin{scope}[shift={(2*\singledx,0)}]
            \draw[very thick, draw=red] (0,-\singledx+1.2) -- (0.8,-\singledx+1.2) --  (0.8,-\singledx-0.8)  -- (0.4,-\singledx-0.8) -- (0.4,-\singledx-1.2);
        \end{scope}
        \begin{scope}[xscale=-1,yscale=-1,shift={(0,0)}]
            \draw[very thick, draw=red] (0,-\singledx+1.2) -- (0.8,-\singledx+1.2) --  (0.8,-\singledx-0.8)  -- (0.4,-\singledx-0.8) -- (0.4,-\singledx-1.2);
        \end{scope}
        \begin{scope}[xscale=-1,yscale=-1,shift={(-\singledx,0)}]
            \draw[very thick, draw=red] (0,-\singledx+1.2) -- (0.8,-\singledx+1.2) --  (0.8,-\singledx-0.8)  -- (0.4,-\singledx-0.8) -- (0.4,-\singledx-1.2);
        \end{scope}
        \begin{scope}[xscale=-1,yscale=-1,shift={(-2*\singledx,0)}]
            \draw[very thick, draw=red] (0,-\singledx+1.2) -- (0.8,-\singledx+1.2) --  (0.8,-\singledx-0.8)  -- (0.4,-\singledx-0.8) -- (0.4,-\singledx-1.2);
        \end{scope}
        \draw (-\singledx+0.5,-\singledx) node {\textcolor{blue}{...}};
        \draw (3*\singledx-0.5,-\singledx) node {\textcolor{blue}{...}};
        \draw (-\singledx+0.5,\singledx) node {\textcolor{blue}{...}};
        \draw (3*\singledx-0.5,\singledx) node {\textcolor{blue}{...}};
\end{tikzpicture}\\ \label{transfer_matrix_structure2}
t_{P_i} &= (\sum_j \alpha_{P_j} \alpha_{P_i^\dag P_j}^*)^L
\end{align}
where \(L\) is the length of the transfer matrix.
\end{theorem}

\begin{proof}
    To begin with, we use Theorem~\ref{iso_structural_theorem} to find
    \begin{align}
        E=A^\dag A = Q^\dag V^\dag V Q = Q^\dag Q
    \end{align}
    where we use the fact that \(V\) and \(Q\) share the same null space to have \(V^\dag V Q = Q\). By using the analytic solution (\ref{topo_structural_theorem_eq}) of \(Q\), we can explicitly calculate \(E\).
    \begin{align}
    \begin{tikzpicture}[scale=0.5,baseline={([yshift=-0.65ex] current bounding box.center) }]
\begin{scope}[shift={(0,0)}]
        \EPTensor{0,0}{1.2}{.6}{\small $E$}{0};
        \RedArrow{-0.8,-0.2}{0.2}{1}
        \RedArrow{1.0,-0.2}{0.2}{1}
        \RedArrow{-1.,0.2}{0.2}{3}
        \RedArrow{0.8,0.2}{0.2}{3}
        \RedArrow{-0.2,1.0}{0.2}{0}
        \RedArrow{-0.2,-0.8}{0.2}{0}
        \RedArrow{0.2,0.8}{0.2}{2}
        \RedArrow{0.2,-1.}{0.2}{2}
\end{scope}
\end{tikzpicture}
        &=\sum_i e_{P_i} \begin{tikzpicture}[scale=0.5,baseline={([yshift=-0.65ex] current bounding box.center) }]
        \PTensor{-\singledx,0}{1.2}{.6}{\small $P_i^\dag$}{3};
        \PTensor{0,-\singledx}{1.2}{.6}{\small $P_i^\dag$}{2};
        \PTensor{0,\singledx}{1.2}{.6}{\small $P_i$}{0};
        \PTensor{\singledx,0}{1.2}{.6}{\small $P_i$}{1};
        \RedArrow{-0.8-\singledx,0}{0.2}{1}
        \RedArrow{1-\singledx,0}{0.2}{1}
        \RedArrow{-0.8+\singledx,0}{0.2}{1}
        \RedArrow{1.0+\singledx,0}{0.2}{1}
        \RedArrow{0,1.0+\singledx}{0.2}{0}
        \RedArrow{0,-0.8+\singledx}{0.2}{0}
        \RedArrow{0,1.0-\singledx}{0.2}{0}
        \RedArrow{0,-0.8-\singledx}{0.2}{0}
        \begin{scope}[xscale=1,yscale=1,shift={(0,0)}]
            \draw[very thick, draw=red] (0,-\singledx+1.2) -- (0.8,-\singledx+1.2) --  (0.8,-\singledx-0.8)  -- (0.4,-\singledx-0.8) -- (0.4,-\singledx-1.2);
        \end{scope}
        \begin{scope}[xscale=-1,yscale=-1,shift={(0,0)}]
            \draw[very thick, draw=red] (0,-\singledx+1.2) -- (0.8,-\singledx+1.2) --  (0.8,-\singledx-0.8)  -- (0.4,-\singledx-0.8) -- (0.4,-\singledx-1.2);
        \end{scope}
        \begin{scope}[xscale=1,yscale=1,shift={(0,0)}, rotate around={90:(0,0)}]
            \draw[very thick, draw=red] (0,-\singledx+1.2) -- (0.8,-\singledx+1.2) --  (0.8,-\singledx-0.8)  -- (0.4,-\singledx-0.8) -- (0.4,-\singledx-1.2);
        \end{scope}
        \begin{scope}[xscale=-1,yscale=-1,shift={(0,0)}, rotate around={90:(0,0)}]
            \draw[very thick, draw=red] (0,-\singledx+1.2) -- (0.8,-\singledx+1.2) --  (0.8,-\singledx-0.8)  -- (0.4,-\singledx-0.8) -- (0.4,-\singledx-1.2);
        \end{scope}
\end{tikzpicture}\\
e_{P_i} &= \sum_j \alpha_{P_j} \alpha_{P_i^\dag P_j}^*
    \end{align}
    We now multiply two \(E\) sideways.
    \begin{align}
\begin{tikzpicture}[scale=0.4,baseline={([yshift=-0.65ex] current bounding box.center) }]
\begin{scope}[shift={(0,0)}]
        \EPTensor{0,0}{1.2}{.6}{\small $E$}{0};
        \RedArrow{-0.8,-0.2}{0.2}{1}
        \RedArrow{1.0,-0.2}{0.2}{1}
        \RedArrow{-1.,0.2}{0.2}{3}
        \RedArrow{0.8,0.2}{0.2}{3}
        \RedArrow{-0.2,1.0}{0.2}{0}
        \RedArrow{-0.2,-0.8}{0.2}{0}
        \RedArrow{0.2,0.8}{0.2}{2}
        \RedArrow{0.2,-1.}{0.2}{2}
\end{scope}
\begin{scope}[shift={(\singledx,0)}]
        \EPTensor{0,0}{1.2}{.6}{\small $E$}{0};
        \RedArrow{-0.8,-0.2}{0.2}{1}
        \RedArrow{1.0,-0.2}{0.2}{1}
        \RedArrow{-1.,0.2}{0.2}{3}
        \RedArrow{0.8,0.2}{0.2}{3}
        \RedArrow{-0.2,1.0}{0.2}{0}
        \RedArrow{-0.2,-0.8}{0.2}{0}
        \RedArrow{0.2,0.8}{0.2}{2}
        \RedArrow{0.2,-1.}{0.2}{2}
\end{scope}
\end{tikzpicture}
= \sum_{i,j} e_{P_i} e_{P_j} \begin{tikzpicture}[scale=0.4,baseline={([yshift=-0.65ex] current bounding box.center) }]
        \begin{scope}[shift={(0,0)}]
        \PTensor{-\singledx,0}{1.2}{.6}{\small $P_i^\dag$}{3};
        \PTensor{0,-\singledx}{1.2}{.6}{\small $P_i^\dag$}{2};
        \PTensor{0,\singledx}{1.2}{.6}{\small $P_i$}{0};
        \PTensor{\singledx,0}{1.2}{.6}{\small $P_i$}{1};
        \RedArrow{-0.8-\singledx,0}{0.2}{1}
        \RedArrow{1-\singledx,0}{0.2}{1}
        \RedArrow{-0.8+\singledx,0}{0.2}{1}
        \RedArrow{1.0+\singledx,0}{0.2}{1}
        \RedArrow{0,1.0+\singledx}{0.2}{0}
        \RedArrow{0,-0.8+\singledx}{0.2}{0}
        \RedArrow{0,1.0-\singledx}{0.2}{0}
        \RedArrow{0,-0.8-\singledx}{0.2}{0}
        \end{scope}
        \begin{scope}[shift={(3*\singledx,0)}]
        \PTensor{-\singledx,0}{1.2}{.6}{\small $P_j^\dag$}{3};
        \PTensor{0,-\singledx}{1.2}{.6}{\small $P_j^\dag$}{2};
        \PTensor{0,\singledx}{1.2}{.6}{\small $P_j$}{0};
        \PTensor{\singledx,0}{1.2}{.6}{\small $P_j$}{1};
        \RedArrow{-0.8-\singledx,0}{0.2}{1}
        \RedArrow{1-\singledx,0}{0.2}{1}
        \RedArrow{-0.8+\singledx,0}{0.2}{1}
        \RedArrow{1.0+\singledx,0}{0.2}{1}
        \RedArrow{0,1.0+\singledx}{0.2}{0}
        \RedArrow{0,-0.8+\singledx}{0.2}{0}
        \RedArrow{0,1.0-\singledx}{0.2}{0}
        \RedArrow{0,-0.8-\singledx}{0.2}{0}
        \end{scope}
        \draw[very thick, draw=red] (\singledx-1.2,0) -- (\singledx-1.2,0.8) -- (2*\singledx+1.2,0.8) -- (2*\singledx+1.2,0);
        \begin{scope}[xscale=1,yscale=1,shift={(0,0)}]
            \draw[very thick, draw=red] (0,-\singledx+1.2) -- (0.8,-\singledx+1.2) --  (0.8,-\singledx-0.8)  -- (0.4,-\singledx-0.8) -- (0.4,-\singledx-1.2);
        \end{scope}
        \begin{scope}[xscale=-1,yscale=-1,shift={(0,0)}]
            \draw[very thick, draw=red] (0,-\singledx+1.2) -- (0.8,-\singledx+1.2) --  (0.8,-\singledx-0.8)  -- (0.4,-\singledx-0.8) -- (0.4,-\singledx-1.2);
        \end{scope}
        \begin{scope}[xscale=1,yscale=1,shift={(3*\singledx,0)}]
            \draw[very thick, draw=red] (0,-\singledx+1.2) -- (0.8,-\singledx+1.2) --  (0.8,-\singledx-0.8)  -- (0.4,-\singledx-0.8) -- (0.4,-\singledx-1.2);
        \end{scope}
        \begin{scope}[xscale=-1,yscale=-1,shift={(-3*\singledx,0)}]
            \draw[very thick, draw=red] (0,-\singledx+1.2) -- (0.8,-\singledx+1.2) --  (0.8,-\singledx-0.8)  -- (0.4,-\singledx-0.8) -- (0.4,-\singledx-1.2);
        \end{scope}
        \begin{scope}[xscale=1,yscale=1,shift={(3*\singledx,0)}, rotate around={90:(0,0)}]
            \draw[very thick, draw=red] (0,-\singledx+1.2) -- (0.8,-\singledx+1.2) --  (0.8,-\singledx-0.8)  -- (0.4,-\singledx-0.8) -- (0.4,-\singledx-1.2);
        \end{scope}
        \begin{scope}[xscale=-1,yscale=-1,shift={(0,0)}, rotate around={90:(0,0)}]
            \draw[very thick, draw=red] (0,-\singledx+1.2) -- (0.8,-\singledx+1.2) --  (0.8,-\singledx-0.8)  -- (0.4,-\singledx-0.8) -- (0.4,-\singledx-1.2);
        \end{scope}
\end{tikzpicture}
\end{align}
Here, the loop is non-zero only when \(i=j\) by the orthogonality of the MF basis, so the only nontrivial terms are those with \(i=j\) with coefficients \(e_{P_i}^2=(\sum_j \alpha_{P_j} \alpha_{P_i^\dag P_j}^*)^2\). Similarly, we can keep multiplying \(E\) from the right and only the on-diagonal terms will survive with coefficients \(e_{P_i}^L\). Lastly, taking the periodic boundary condition annihilates the two operators at the left and right.
\end{proof}

With the analytic expression of the transfer matrix, we now include the additional symmetry constraint (\ref{topo_more_sym}) which directly implies degeneracy in the largest eigenvalue, when we choose some appropriate \(L\) to annihilate the phase \(\phi\).
\begin{corollary}
In addition to Eqs.~(\ref{topo_sym}), under Eqs.~(\ref{topo_more_sym}) with a subgroup \(\{M\}\), the transfer matrix \(\mathbb T\) with length \(L\) being an integer multiple of \(m=|\{M\}|\) has at least an \(m\)-fold degeneracy in the largest eigenvalue.
\end{corollary}

\begin{proof}
Eq.~(\ref{topo_more_sym}) implies that 
\begin{align}
    e_{M} = \sum_j |\alpha_{P_j}|^2 e^{i \phi}, \, \forall M
\end{align}
Now we rewrite Eq.~(\ref{transfer_matrix_structure1}) in the following way.
\begin{align}
    \mathbb T = \sum_i e_{P_i}^L [(P_i)_{ab}(P_i^\dag)_{cd}\ket{c}\ket{d}\bra{a}\bra{b}]^{\otimes L}
\end{align}
In this way, it becomes appearant that \(\mathbb T\) is already diagonalized, so \(e_{P_i}^L\) are the non-trivial eigenvalues. By choosing \(L\) to be integer multiples of \(m\) (in other word we annihilate the charges associated with the tensor),
\begin{align}
    e_{M}^L = (\sum_j |\alpha_{P_j}|^2)^L, \, \forall M
\end{align}
Since there are \(m\) elements in total, the eigenvalues \(e_{M}\) is \(m\)-fold degenerate, so it remains to show that \(e_{M}\) is the biggest eigenvalue. This is shown via the Cauchy-Schwartz inequality.
\begin{align}
e_{P_i} &= \sum_j \alpha_{P_j} \alpha_{P_i^\dag P_j}^* \\
&\le (\sum_j |\alpha_{P_j}|^2)^{\frac{1}{2}}(\sum_j |\alpha_{P_i^\dag P_j}|^2)^{\frac{1}{2}}\\
&= \sum_j |\alpha_{P_j}|^2 = e_{M}
\end{align}
\end{proof}

As an example, we consider the qubit Pauli group as the MF basis and observe the following state.
\begin{align}
Q=   \begin{tikzpicture}[scale=0.5,baseline={([yshift=-0.65ex] current bounding box.center) }]
    \PTensor{-\singledx,0}{1.5}{.6}{\small $I$}{3};
    \PTensor{\singledx,0}{1.5}{.6}{\small $I$}{1};
    \PTensor{0,-\singledx}{1.5}{.6}{\small $I$}{2};
    \PTensor{0,\singledx}{1.5}{.6}{\small $I$}{0};
    \RedArrow{-0.8-\singledx,0}{0.2}{1}
    \RedArrow{1.0+\singledx,0}{0.2}{1}
    \RedArrow{0,1.0+\singledx}{0.2}{0}
    \RedArrow{0,-0.8-\singledx}{0.2}{0}
    \end{tikzpicture}
    +\begin{tikzpicture}[scale=0.5,baseline={([yshift=-0.65ex] current bounding box.center) }]
    \PTensor{-\singledx,0}{1.5}{.6}{\small $X$}{3};
    \PTensor{\singledx,0}{1.5}{.6}{\small $X$}{1};
    \PTensor{0,-\singledx}{1.5}{.6}{\small $X$}{2};
    \PTensor{0,\singledx}{1.5}{.6}{\small $X$}{0};
    \RedArrow{-0.8-\singledx,0}{0.2}{1}
    \RedArrow{1.0+\singledx,0}{0.2}{1}
    \RedArrow{0,1.0+\singledx}{0.2}{0}
    \RedArrow{0,-0.8-\singledx}{0.2}{0}
    \end{tikzpicture} \nonumber\\
    +\alpha \big(\, \begin{tikzpicture}[scale=0.5,baseline={([yshift=-0.65ex] current bounding box.center) }]
    \PTensor{-\singledx,0}{1.5}{.6}{\small $Y$}{3};
    \PTensor{\singledx,0}{1.5}{.6}{\small $Y$}{1};
    \PTensor{0,-\singledx}{1.5}{.6}{\small $Y$}{2};
    \PTensor{0,\singledx}{1.5}{.6}{\small $Y$}{0};
    \RedArrow{-0.8-\singledx,0}{0.2}{1}
    \RedArrow{1.0+\singledx,0}{0.2}{1}
    \RedArrow{0,1.0+\singledx}{0.2}{0}
    \RedArrow{0,-0.8-\singledx}{0.2}{0}
    \end{tikzpicture}
    +\begin{tikzpicture}[scale=0.5,baseline={([yshift=-0.65ex] current bounding box.center) }]
    \PTensor{-\singledx,0}{1.5}{.6}{\small $Z$}{3};
    \PTensor{\singledx,0}{1.5}{.6}{\small $Z$}{1};
    \PTensor{0,-\singledx}{1.5}{.6}{\small $Z$}{2};
    \PTensor{0,\singledx}{1.5}{.6}{\small $Z$}{0};
    \RedArrow{-0.8-\singledx,0}{0.2}{1}
    \RedArrow{1.0+\singledx,0}{0.2}{1}
    \RedArrow{0,1.0+\singledx}{0.2}{0}
    \RedArrow{0,-0.8-\singledx}{0.2}{0}
    \end{tikzpicture} \,\big)
\end{align}
The above state can be interpreted as an interpolation between the toric code state (\ref{toric_code_ex}) when \(\alpha=0\) and the equal sum state (\ref{full_sym_ex}) when \(\alpha=1\). The entire family of states possesses the symmetry of type (\ref{topo_more_sym}) with \(M=X\), so one would expect a two-fold degeneracy in the largest eigenvalue. Indeed, we have \(e_{I}=e_{X}=2+2\alpha^2\) and \(e_{Y}=e_{Z}=4\alpha\). When \(\alpha=1\), we have an additional symmetry of the type (\ref{topo_more_sym}) with \(M=Z\), and the degeneracy is promoted to four-fold. 
%
%
In summary, this is a simple example of a continuous family of states, with signatures of topological order, preparable with MF in constant depth.

The above discussion characterizes the degeneracy in the spectrum of the transfer matrix when the MF PEPS satisfies the topological-type symmetries Eqs.~(\ref{topo_sym}) and (\ref{topo_more_sym}). Degeneracy in the spectrum does not directly imply topological order, since a symmetry-broken phase can also possess such degeneracy. Nevertheless, our results provide a systematic characterization of MF PEPS with topological-type symmetries and supports the existence of the topological order.

\section{Implementing Unitaries with MF}\label{section_mpu}

Here we extend the idea of MF from state preparations to applying matrix-product operators (MPO) which are unitary or, more generally, isometries. Many, but not all, of the existing constructions are based on Clifford teleportation~\cite{gottesman1999demonstrating,piroli2021quantum,aasen2022adiabatic,tantivasadakarn2023hierarchy,buhrman2023state,lootens2023low,piroli2024approximating,sukeno2024quantum} which has a discrete structure. Here we will show that the considered protocol, although not much stronger than Clifford teleportation, allows for continuous families.


We consider the following protocol to implement a MPO via MF and derive a structural theorem.
A MPO can be represented as the following operator
\begin{align}
    \textcolor{red}{...}
    \begin{tikzpicture}[scale=0.5,baseline={([yshift=-0.65ex] current bounding box.center) }]
        \GTensor{0,0}{1.}{.6}{\small $O_1$}{5};
        \GTensor{\singledx,0}{1.}{.6}{\small $O_2$}{5};
        \GTensor{2*\singledx,0}{1.}{.6}{\small $O_3$}{5};
\end{tikzpicture}
\textcolor{red}{...}
\;,
\end{align}
where the boundary conditions will be specified later. In general, MPO can create correlations at arbitrary distances, so they cannot be implemented with constant-depth unitaries. This justifies the use of MF to implement the MPO. The protocol begins with the application of individual MPO tensors to the (arbitrary) input state.
\begin{align}
    \begin{tikzpicture}[scale=0.5,baseline={([yshift=-0.65ex] current bounding box.center) }]
        \GTensor{0,0}{1.2}{.6}{\small $O_1$}{5};
        \GTensor{2*\singledx,0}{1.2}{.6}{\small $O_2$}{5};
        \GTensor{4*\singledx,0}{1.2}{.6}{\small $O_3$}{5};
        \begin{scope}[shift={(0,0)}]
        \draw[very thick, draw=red] (-1.2,0) -- (-1.2,1.2);
        \draw[very thick, draw=red] (1.2,0) -- (1.2,1.2);
        \end{scope}
        \begin{scope}[shift={(2*\singledx,0)}]
        \draw[very thick, draw=red] (-1.2,0) -- (-1.2,1.2);
        \draw[very thick, draw=red] (1.2,0) -- (1.2,1.2);
        \end{scope}
        \begin{scope}[shift={(4*\singledx,0)}]
        \draw[very thick, draw=red] (-1.2,0) -- (-1.2,1.2);
        \draw[very thick, draw=red] (1.2,0) -- (1.2,1.2);
        \end{scope}
        \draw[thick, fill=white, rounded corners=2pt] (-0.6,-\singledx-0.6) rectangle (4*\singledx+0.6,-\singledx+0.6);
        \draw (2*\singledx,-\singledx) node {\small $\ket{\phi}$};
\end{tikzpicture}
\end{align}
For this to be a physical process without post-selection, we require each MPO tensor \(O_i\) to be proportional to an isometry from bottom to top
\begin{align}\label{mpu_orthogonality}
\begin{tikzpicture}[scale=0.5,baseline={([yshift=-0.65ex] current bounding box.center) }]
        \GTensor{0,0}{1.}{.6}{\small $O$}{5};
        \draw (0,-1.3) node {\small $a$};
        \GTensor{0,\singledx}{1.}{.6}{\small $O^\dag$}{5};
        \draw (0,\singledx+1.3) node {\small $b$};
        \begin{scope}[shift={(0,0)}]
        \draw[very thick, draw=red] (-1.,0) -- (-1.,\singledx);
        \draw[very thick, draw=red] (1.,0) -- (1.,\singledx);
        \end{scope}
\end{tikzpicture}
= D \, \delta_{ab} \;,
\end{align}
where \(D\) is the virtual bond dimension; the reason for this choice will become clear later. Next, we glue the adjacent virtual legs using MF:
\begin{align}
    \begin{tikzpicture}[scale=0.5,baseline={([yshift=-0.65ex] current bounding box.center) }]
        \GTensor{0,0}{1.2}{.6}{\small $O_1$}{5};
        \GTensor{2*\singledx,0}{1.2}{.6}{\small $O_2$}{5};
        \GTensor{4*\singledx,0}{1.2}{.6}{\small $O_3$}{5};
        \DTensor{1*\singledx,0}{1.2}{.6}{\small $P_1$}{0};
        \DTensor{3*\singledx,0}{1.2}{.6}{\small $P_2$}{0};
        \draw[thick, fill=white, rounded corners=2pt] (-0.6,-\singledx-0.6) rectangle (4*\singledx+0.6,-\singledx+0.6);
        \draw (2*\singledx,-\singledx) node {\small $\ket{\phi}$};
        \begin{scope}[shift={(0,0)}]
        \draw[very thick, draw=red] (-1.2,0) -- (-1.2,1.2);
        \end{scope}
        \begin{scope}[shift={(4*\singledx,0)}]
        \draw[very thick, draw=red] (1.2,0) -- (1.2,1.2);
        \end{scope}
\end{tikzpicture}
\end{align}
After that, one would like to push the defect operators through by applying local error corrections \(U_{P_i}\) and annihilate them at the open boundary. Alternatively, one can impose a periodic boundary condition and try to push all the defect operators to the same site. This results in a probabilistic protocol but still with an \(O(1)\) success rate. MF imposes the following symmetry to the MPO tensor.
\begin{align}\label{mpo_post_sel_sym}
\begin{tikzpicture}[scale=0.5,baseline={([yshift=-2.6ex] current bounding box.center) }]
        \DTensor{-\singledx,0}{1.}{.6}{\small $P_i$}{0};
        \GTensor{0,0}{1.}{.6}{\small $O$}{5};
        \UTensor{0,\singledx}{1.}{.6}{\small $U_{P_i}$}{3};
        \draw (0,-1.3) node {\small $a$};
\end{tikzpicture}
=\begin{tikzpicture}[scale=0.5,baseline={([yshift=0.4ex] current bounding box.center) }]
        \DTensor{\singledx,0}{1.}{.6}{\small $P'_i$}{0};
        \GTensor{0,0}{1.}{.6}{\small $O$}{5};
        \draw (0,-1.3) node {\small $a$};
\end{tikzpicture}
\; , \quad  \forall P_i, a.
\end{align}
%
The MPO can be understood as a collection of MPS labelled by \(a\) satisfying the same symmetry, connecting to our previous results on MF MPS.
%
We now show that this results to certain orthogonality conditions.
\begin{lemma}\label{mpo_post_sel_lemma}
Under Eqs.~(\ref{mpu_orthogonality},\ref{mpo_post_sel_sym}), 
%
\begin{align}
    \begin{tikzpicture}[scale=0.5,baseline={([yshift=.65ex] current bounding box.center) }]
        \GTensor{0,0}{1.}{.6}{\small $O$}{5};
        \draw[very thick,black] (0,1) -- (1,1);
        \draw (0,-1.2) node {\small $a$};
\end{tikzpicture}
\; =\;\;
    \begin{tikzpicture}[scale=0.5,baseline={([yshift=-3] current bounding box.center) }]
        \Unitary{0,0}{1}{.6}{\small $V_a$}{1};
\end{tikzpicture}
\; , \quad \forall a
\end{align}
where \(V_a\) are isometries possessing the MF symmetries, i.e.,
\begin{align}\label{mpo_iso_sym}
\begin{tikzpicture}[scale=0.5,baseline={([yshift=-0.65ex] current bounding box.center) }]
        \Unitary{0,0}{1}{.6}{\small $V_a$}{1};
        \DTensor{-\singledx,0}{1.}{.6}{\small $P_i$}{0};
\end{tikzpicture}
=\begin{tikzpicture}[scale=0.5,baseline={([yshift=-0.65ex] current bounding box.center) }]
        \Unitary{0,0}{1}{.6}{\small $V_a$}{1};
        \DTensor{\singledx,0}{1.}{.5}{\small $P'_i$}{0};
        \UTensor{\singledx,1.2}{1.}{.5}{\small $U_{P_i}^\dag$}{2};
\end{tikzpicture}
\; , \quad \forall P_i, V_a \;.
\end{align}
In addition, different $V_a$ satisfy the following orthogonality relation.
\begin{align}\label{mpu_iso_orthogonality}
    \begin{tikzpicture}[scale=0.5,baseline={([yshift=-1.65ex] current bounding box.center) }]
        \Unitary{0,0}{1}{.6}{\small $V_a$}{1};
        \begin{scope}[shift={(\singledx,0)},xscale=-1,yscale=1]
            \Unitary{0,0}{1}{.6}{\small $V_b^\dag$}{1};
        \end{scope}
\end{tikzpicture}
\; = \; \delta_{ab} \;
\begin{tikzpicture}[scale=0.5,baseline={([yshift=0.65ex] current bounding box.center) }]
        \draw[very thick, draw=red](-1.2,0)--(1.2,0);
\end{tikzpicture}
\end{align}
\end{lemma}
\begin{proof}
    The statement that \(V_a\) are isometries possessing the Clifford-like symmetries follows trivially from the MPS result Lemma \ref{lemma_cf}. To show the orthogonality relation and to fix the constant multiplier, we apply the MF symmetries to the left hand side of Eq.~\eqref{mpu_iso_orthogonality}.
    \begin{align}
    \begin{tikzpicture}[scale=0.5,baseline={([yshift=-0.65ex] current bounding box.center) }]
        \Unitary{0,0}{1}{.6}{\small $V_a$}{1};
        \begin{scope}[shift={(\singledx,0)},xscale=-1,yscale=1]
            \Unitary{0,0}{1}{.6}{\small $V_b^\dag$}{1};
        \end{scope}
\end{tikzpicture}
= \begin{tikzpicture}[scale=0.5,baseline={([yshift=-0.65ex] current bounding box.center) }]
        \DTensor{-\singledx,0}{1.2}{.6}{\small $P_i$}{0};
        \DTensor{2*\singledx,0}{1.2}{.6}{\small $P_i^\dag$}{0};
        \Unitary{0,0}{1}{.6}{\small $V_a$}{1};
        \begin{scope}[shift={(\singledx,0)},xscale=-1,yscale=1]
            \Unitary{0,0}{1}{.6}{\small $V_b^\dag$}{1};
        \end{scope}
\end{tikzpicture}
\; ,\; \forall P_i .
\end{align}
Therefore, by Schur's lemma, \(V_a V_b^\dag = k_{ab} I\) where \(k_{ab}\) is a constant multiplier depending on \(a\) and \(b\). Additionally, the orthogonality requirement (\ref{mpu_orthogonality}) demands that
\begin{align}
    \frac{1}{D}\;
    \begin{tikzpicture}[scale=0.5,baseline={([yshift=-0.65ex] current bounding box.center) }]
        \Unitary{0,0}{1.2}{.6}{\small $V_a$}{1};
        \begin{scope}[shift={(\singledx,0)},xscale=-1,yscale=1]
            \Unitary{0,0}{1.2}{.6}{\small $V_b^\dag$}{1};
        \end{scope}
        \draw[very thick, draw=red](-1.2,0)--(-1.2,-1.2)--(\singledx+1.2,-1.2)--(\singledx+1.2,0);
\end{tikzpicture}
\;=\;\frac{k_{ab}}{D} \; \begin{tikzpicture}[scale=0.5,baseline={([yshift=-0.65ex] current bounding box.center) }]
        \draw[very thick, draw=red](-1.2,0)--(-1.2,-1.2)--(\singledx+1.2,-1.2)--(\singledx+1.2,0)--(-1.2,0);
\end{tikzpicture}
\;=\;\delta_{ab}
\end{align}
The above equation restricts \(k_{ab} = \delta_{ab}\), and Eqs. (\ref{mpu_iso_orthogonality}) naturally follow.
\end{proof}

With the above lemma decomposing the MPO into orthogonal isometries, we are now ready to state our structural theorem of MPO which connects it to the Clifford-like teleportation.
\begin{theorem}
Under Eqs.~(\ref{mpu_orthogonality},\ref{mpo_post_sel_sym}),
\begin{align} \label{mpu_form}
\begin{tikzpicture}[scale=0.5,baseline={([yshift=-0.65ex] current bounding box.center) }]
        \GTensor{0,0}{1.2}{.6}{\small $O$}{5};
\end{tikzpicture}
=\begin{tikzpicture}[scale=0.5,baseline={([yshift=0.3ex] current bounding box.center) }]
        \draw[very thick, draw=red] (-1.6,0) -- (1.6,0);
        \Unitary{0,0}{1.2}{.6}{\small $U$}{0};
        \draw[very thick] (0+1.2,1.2) -- (0+1.2,2.4);
        \draw[very thick] (0-1.2,1.2) -- (0-1.2,-\singledx);
        \GTensor{0-1.2,-\singledx}{1.2}{0.6}{\small $\tilde{U}$}{3};
    \end{tikzpicture}
\end{align}
where \(U\) is a unitary satisfying the MF symmetry constraints
\begin{align}\label{mpo_purification_sym}
\begin{tikzpicture}[scale=0.5,baseline={([yshift=-0.65ex] current bounding box.center) }]
        \Unitary{0,0}{1}{.6}{\small $U$}{0};
        \DTensor{-\singledx,0}{1.}{.6}{\small $P_i$}{0};
\end{tikzpicture}
\; = \; \begin{tikzpicture}[scale=0.5,baseline={([yshift=-0.65ex] current bounding box.center) }]
        \Unitary{0,0}{1}{.6}{\small $U$}{0};
        \DTensor{\singledx,0}{1.}{.5}{\small $P'_i$}{0};
        \UTensor{\singledx,1.2}{1.}{.5}{\small $U_{P_i}^\dag$}{2};
\end{tikzpicture}
\;, \quad \forall P_i
\end{align}
and \(\tilde{U}\) is a unitary determined by \(O\). Moreover, one can choose a fixed \(U\) satisfying Eqs.~(\ref{mpo_purification_sym}) and represent all the possible \(O\) under the same MF symmetry by varying \(\tilde{U}\).
\end{theorem}

\begin{proof}
    We will construct \(U\) and \(\tilde{U}\) explicitly. We take any \(O\) satisfying the MF symmetry Eq. (\ref{mpo_post_sel_sym}). We then define \(U\) as
\begin{align}
\begin{tikzpicture}[scale=0.5,baseline={([yshift=-0.65ex] current bounding box.center) }]
        \Unitary{0,0}{1}{.6}{\small $U$}{0};
        \draw (-1.5,1.2) node {\small $a$};
\end{tikzpicture}
=\begin{tikzpicture}[scale=0.5,baseline={([yshift=-0.65ex] current bounding box.center) }]
        \Unitary{0,0}{1}{.6}{\small $V_a$}{1};
\end{tikzpicture}
=\begin{tikzpicture}[scale=0.5,baseline={([yshift=0ex] current bounding box.center) }]
        \GTensor{0,0}{1.}{.6}{\small $O$}{5};
        \draw[very thick,black] (0,1) -- (1,1);
        \draw (0,-1.2) node {\small $a$};
\end{tikzpicture}
\end{align}
Where we introduce a basis labeled by \(a\) in the top-left leg, and \(V_a\) are the isometries as defined in Lemma~\ref{mpo_post_sel_lemma}. Naturally \(U\) satisfies Eqs.~(\ref{mpo_purification_sym}) because all \(V_a\) satisfy the same symmetry. Hence, it remains to show that \(U\) is a unitary, which is true since
\begin{align}
\begin{tikzpicture}[scale=0.5,baseline={([yshift=-0.65ex] current bounding box.center) }]
        \Unitary{0,0}{1.2}{.6}{\small $U$}{0};
        \Unitary{\singledx,0}{1.2}{.6}{\small $U^\dag$}{0};
        \draw (-1.5,1.2) node {\small $a$};
        \draw (\singledx+1.5,1.2) node {\small $b$};
\end{tikzpicture}
    =\begin{tikzpicture}[scale=0.5,baseline={([yshift=-0.65ex] current bounding box.center) }]
        \Unitary{0,0}{1}{.6}{\small $V_a$}{1};
        \begin{scope}[shift={(\singledx,0)},xscale=-1,yscale=1]
            \Unitary{0,0}{1}{.6}{\small $V_b^\dag$}{1};
        \end{scope}
\end{tikzpicture}
= \delta_{ab} \; \begin{tikzpicture}[scale=0.5,baseline={([yshift=-0.65ex] current bounding box.center) }]
        \draw[very thick, draw=red](-1.2,0)--(1.2,0);
\end{tikzpicture}
.
\end{align}
Now we show that any other MPO \(O'\) satisfying the same MF symmetry can be written as a combination of \(U\) and \(\tilde{U}\). Given any other \(O'\) satisfying the same MF symmetry, we carry out the same procedure and construct another unitary \(U'\). In other words,
\begin{align}
\begin{tikzpicture}[scale=0.5,baseline={([yshift=-0.65ex] current bounding box.center) }]
        \Unitary{0,0}{1}{.6}{\small $U'$}{0};
        \draw (-1.5,1.2) node {\small $a$};
\end{tikzpicture}
=\begin{tikzpicture}[scale=0.5,baseline={([yshift=0.2ex] current bounding box.center) }]
        \GTensor{0,0}{1.}{.6}{\small $O'$}{5};
        \draw[very thick,black] (0,1) -- (1,1);
        \draw (0,-1.2) node {\small $a$};
\end{tikzpicture}
.
\end{align}
Because \(O'\) satisfies the same MF symmetry, \(U'\) satisfies Eqs. (\ref{mpo_purification_sym}) as well. Now we contract \(U'U^\dag\) and apply the MF symmetry.
\begin{align}
    \begin{tikzpicture}[scale=0.5,baseline={([yshift=-0.65ex] current bounding box.center) }]
        \Unitary{0,0}{1.2}{.6}{\small $U'$}{0};
        \Unitary{\singledx,0}{1.2}{.6}{\small $U^\dag$}{0};
\end{tikzpicture}
= \begin{tikzpicture}[scale=0.5,baseline={([yshift=-0.65ex] current bounding box.center) }]
        \DTensor{-\singledx,0}{1.2}{.6}{\small $P_i$}{0};
        \DTensor{2*\singledx,0}{1.2}{.6}{\small $P_i^\dag$}{0};
        \Unitary{0,0}{1.2}{.6}{\small $U'$}{0};
        \Unitary{\singledx,0}{1.2}{.6}{\small $U^\dag$}{0};
\end{tikzpicture}
,\; \forall P_i.
\end{align}
Therefore, by Schur's lemma,
\begin{align}
    \begin{tikzpicture}[scale=0.5,baseline={([yshift=-0.65ex] current bounding box.center) }]
        \Unitary{0,0}{1.2}{.6}{\small $U'$}{0};
        \Unitary{\singledx,0}{1.2}{.6}{\small $U^\dag$}{0};
\end{tikzpicture}
= \begin{tikzpicture}[scale=0.5,baseline={([yshift=-0.65ex] current bounding box.center) }]
        \GTensor{0,0}{1.2}{.6}{\small $\tilde{U}$}{4};
        \draw[very thick, draw=red](-1.2,-1.2)--(1.2,-1.2);
\end{tikzpicture}
\end{align}
Where the constant multiplier has to be 1 to agree with the unitarity of \(U' U^\dag\). As a result, \(O'\) takes the following form~\eqref{mpu_form}.
\end{proof}

The above theorem implies that the MF symmetry determines a class of \(U\) possessing the Clifford-like symmetry~(\ref{mpo_purification_sym}). Moreover, while there exists a continuous family of MF MPOs, they are related by a local unitary \(\tilde{U}\). Hence, the MPO can be written as one layer of arbitrary local rotation and Clifford-like teleportation.
\begin{align}
&\begin{tikzpicture}[scale=0.5,baseline={([yshift=-0.65ex] current bounding box.center) }]
        \GTensor{0,0}{1.}{.6}{\small $O_1$}{5};
        \GTensor{\singledx,0}{1.}{.6}{\small $O_2$}{5};
        \GTensor{2*\singledx,0}{1.}{.6}{\small $O_3$}{5};
\end{tikzpicture}\\
    &=\begin{tikzpicture}[scale=0.5,baseline={([yshift=-0.65ex] current bounding box.center) }]
        \Unitary{0,0}{1.2}{.6}{\small $U_{1}$}{0};
        \Unitary{2*\singledx,0}{1.2}{.6}{\small $U_{2}$}{0};
        \Unitary{4*\singledx,0}{1.2}{.6}{\small $U_{3}$}{0};
        \draw[very thick, draw=red] (0+1,0) -- (2*\singledx-1,0);
        \draw[very thick, draw=red] (2*\singledx+1,0) -- (4*\singledx-1,0);
        \draw[very thick] (0+1.2,1.2) -- (0+1.2,2.4);
        \draw[very thick] (2*\singledx+1.2,1.2) -- (2*\singledx+1.2,2.4);
        \draw[very thick] (4*\singledx+1.2,1.2) -- (4*\singledx+1.2,2.4);
        \draw[very thick, draw=red] (4*\singledx+1,0) -- (4*\singledx+2.,0);
        \draw[very thick, draw=red] (-1,0) -- (-2.,0);
        \draw[very thick] (0-1.2,1.2) -- (0-1.2,-\singledx);
        \draw[very thick] (2*\singledx-1.2,1.2) -- (2*\singledx-1.2,-\singledx);
        \draw[very thick] (4*\singledx-1.2,1.2) -- (4*\singledx-1.2,-\singledx);
        \GTensor{0-1.2,-\singledx}{1.2}{0.6}{\small $\tilde{U}_1$}{3};
        \GTensor{2*\singledx-1.2,-\singledx}{1.2}{0.6}{\small $\tilde{U}_2$}{3};
        \GTensor{4*\singledx-1.2,-\singledx}{1.2}{0.6}{\small $\tilde{U}_3$}{3};
    \end{tikzpicture}
\end{align}
The above structure also implies that the MPO has to be an isometry under the open-boundary condition. If the leftmost virtual legs at the boundary is considered as input and the rightmost as output, the resulting MPO is a unitary staircase. In the case of periodic boundary conditions, one have to perform post-selection, and hence the MPO need not be unitary. 

Lastly, we note that the above discussion can be extended to 2D as well. Rather than having an MPO, in 2D we can implement projected entangled pair operators (PEPOs) using MF. By the same analogy, the PEPOs are equivalent to one layer of local unitaries and then a layer of Clifford-like teleportation on a 2D grid.

\section{Discussion}
In this work, we systematically investigate the properties of states preparable using MF. In the language of tensor networks, preparability under MF translates to symmetry constraints of the MPS tensor in one dimension and PEPS tensor in two dimensions. We characterize the universal structures of such tensors and present the underlying Clifford-like structure and the source of magic. In addition, we provide analytic solutions to states with physically-motivated MF symmetries and discuss their connection to SPT order in one dimension and topological order in two dimensions. Lastly, we extend the formalism to implementing operators using MF and provide a similar structural theorem to analyze the Clifford-like structure and the magic.

This work opens up new windows to investigate the power of MF.
First, we have provided analytic solutions to some MF states with physically-motivated symmetries which form continuous families, so it would be interesting to investigate these families further. Some questions include what phases can the family belong to, and if one can find a phase transition in the MF family.
Second, it is known that performing multiple rounds of MF can prepare additional states such as the nonabelian topological order characterized by solvable groups \cite{tantivasadakarn2023hierarchy,li2023symmetry}. One would thus wonder if there are other type of states preparable using multiple rounds of MF, and if so what are their properties. From the perspective of symmetries of the tensor network, preparing a state using multiple rounds of MF amounts to including additional symmetry constraints, so we expect similar technique to be applicable even in the case of performing multiple rounds of MF. 
%

\emph{Note added}.---During the completion of this manuscript, four related works appeared \cite{smith2024constant,sahay2024classifying,sahay2024finite,stephen2024preparing} which discuss state preparations using MF from various perspectives. Our results agree with them where they overlap.

\begin{acknowledgments}
We would like to thank Tsung-Cheng Peter Lu, Grace Sommers, Marten Folkertsma, Norbert Schuch, and J.~Ignacio Cirac for useful discussions. S.G. and Y.Z. acknowledge support from NSF QuSEC-TAQS OSI 2326767.
\end{acknowledgments}

\appendix

\section{Existence of Clifford Purification}\label{existance_clifford_purification}
In this section we prove Lemma \ref{lemma_clifford_purification}, in other words discuss why any Clifford-like isometry in Theorem \ref{structural_theorem_2} admits a Clifford unitary purification. We borrow the techniques from \cite{hostens2005stabilizer}. Consider a system with \(n\) number of \(d\)-dimensional qudits and define the following notations.
\begin{align}
    a &=(v_1, v_2,...,v_n,w_1,w_2,...,w_n) \in  \mathbb{Z}_d^{2n}\\
    XZ(a)&=X_1^{v_1}Z_1^{w_1} \otimes X_2^{v_2}Z_2^{w_2} \otimes ... \otimes X_n^{v_n}Z_n^{w_n}
\end{align}
Where \(X_k\) and \(Z_k\) are the Weyl-Heisenberg shift and clock operators of qudit \(a\). In this way, the vector \(a\) uniquely labels an element up to a phase. 

A multiqudit Clifford unitary is specified by its action on the generator of the multiqudit Weyl-Heisenberg group. 
\begin{align}
    U_C X_k U_C^\dag &= e^{i h_k \pi /2d}XZ(C_k)\\
    U_C Z_l U_C^\dag &= e^{i h_l \pi /2d}XZ(C_l)
\end{align}
Not that the phase is the \(2d\)-th root of unity in general. We let \(C_k, C_l\) be the column vectors of a matrix \(C \in \mathbb{Z}_d^{2n \times 2n}\) and let \(h_k, h_l\) be the elements of a vector \(h \in \mathbb{Z}_{2d}^{2n}\). In \(C\) and \(h\) completely specifies the action of a linear map, but it does not necessarily have to be a unitary. Crucially, Ref. \cite{hostens2005stabilizer} proves that \(C\) and \(h\) unique specifies a Clifford unitary if they satisfy certain conditions.
\begin{theorem}
(Reproduced from \cite{hostens2005stabilizer} section II, III) Let the linear map defined by \(C\) and \(h\) to be \(\mathcal{E}\). If it satisfies the following properties, then there exist a Clifford unitary \(U_C\) such that \(U_C^\dag \mathcal{E}(XZ(a)) U_C = XZ(a), \forall a\).
\begin{enumerate}
    \item \(\mathcal{E}\) preserves the commutation relation between Weyl-Heisenberg group elements. Specifically,
    \begin{align}
    \mathcal{E}(Z_a)\mathcal{E}(Z_b) &= \mathcal{E}(Z_b)\mathcal{E}(Z_a), \forall a, b \\
    \mathcal{E}(X_a)\mathcal{E}(X_b) &= \mathcal{E}(X_b)\mathcal{E}(X_a), \forall a, b \\
    \mathcal{E}(Z_a)\mathcal{E}(X_b) &= \mathcal{E}(Z_b)\mathcal{E}(X_a), \forall a \neq b \\
    \mathcal{E}(Z_a)\mathcal{E}(X_a) &= e^{i 2\pi /d}\mathcal{E}(Z_a)\mathcal{E}(Z_a), \forall a
    \end{align}
    \item \(\mathcal{E}\) preserves the order of elements. Specifically, given that \(Z_a^d=I\), \(X_a^d=I\), it suffices to have
    \begin{align}
        \mathcal{E}(Z_a)^d&=I, \forall a\\
        \mathcal{E}(X_a)^d&=I, \forall a
    \end{align}
\end{enumerate}
\end{theorem}
The proof is based on an explicit algorithm that sequentially applies single-qubit and two-qubit gates which act as row operations on \(C\) to diagonalize \(C\). Then, a final layer of Weyl-Heisenberg unitaries is applied to set \(h=0\). Reversing the unitaries generated by the algorithm give rises to \(U_C\). 

Now, the isometry with Clifford-like symmetries in Theorem~\ref{structural_theorem_2} defines a map from \(P_i\) to \((P_i^\dag \otimes P_i'^T \otimes P_i'^\dag)\). This map 
can be considered as two columns of \(C\) and two elements of \(h\) which specifies the action of \(\mathcal{E}(Z_1)\) and \(\mathcal{E}(X_1)\). 
Importantly, the above map satisfies the two constraints. To verify the first constraint, take \(P_a\) and \(P_b\) to be \(X\) or \(Z\). We apply the symmetry of \(P_a\) first and then apply the \(P_b\):
\begin{align}
\begin{split}
    \begin{tikzpicture}[scale=0.5,baseline={([yshift=-0.65ex] current bounding box.center) }]
        \FUnitary{0,0}{1.2}{.6}{\small $V_Q$}{1};
\end{tikzpicture}
=\begin{tikzpicture}[scale=0.5,baseline={([yshift=-0.65ex] current bounding box.center) }]
        \FUnitary{0,0}{1.2}{.6}{\small $V_Q$}{1};
        \DTensor{-\singledx,0}{1.2}{.6}{\small $P_a$}{0};
        \DTensor{\singledx,0}{1.2}{.5}{\small $P_a'^\dag$}{0};
        \DTensor{\singledx,1.2}{1.2}{.5}{\small $P_a'^T$}{0};
        \DTensor{\singledx,2.4}{1.2}{.5}{\small $P_a^\dag$}{0};
\end{tikzpicture}\\
=\begin{tikzpicture}[scale=0.5,baseline={([yshift=-0.65ex] current bounding box.center) }]
        \FUnitary{0,0}{1.2}{.6}{\small $V_Q$}{1};
        \DTensor{-2*\singledx,0}{1.2}{.6}{\small $P_a$}{0};
        \DTensor{-\singledx,0}{1.2}{.6}{\small $P_b$}{0};
        \DTensor{2*\singledx,0}{1.2}{.5}{\small $P_a'^\dag$}{0};
        \DTensor{2*\singledx,1.2}{1.2}{.5}{\small $P_a'^T$}{0};
        \DTensor{2*\singledx,2.4}{1.2}{.5}{\small $P_a^\dag$}{0};
        \DTensor{\singledx,0}{1.2}{.5}{\small $P_b'^\dag$}{0};
        \DTensor{\singledx,1.2}{1.2}{.5}{\small $P_b'^T$}{0};
        \DTensor{\singledx,2.4}{1.2}{.5}{\small $P_b^\dag$}{0};
\end{tikzpicture}
\end{split}
\end{align}
we can also apply the symmetry of \(P_b\) first and then apply the \(P_a\) and the results should be equal. Also
\begin{align}
\begin{split}
    \begin{tikzpicture}[scale=0.5,baseline={([yshift=-0.65ex] current bounding box.center) }]
        \FUnitary{0,0}{1.2}{.6}{\small $V_Q$}{1};
\end{tikzpicture}
=\begin{tikzpicture}[scale=0.5,baseline={([yshift=-0.65ex] current bounding box.center) }]
        \FUnitary{0,0}{1.2}{.6}{\small $V_Q$}{1};
        \DTensor{-\singledx,0}{1.2}{.6}{\small $P_b$}{0};
        \DTensor{\singledx,0}{1.2}{.5}{\small $P_b'^\dag$}{0};
        \DTensor{\singledx,1.2}{1.2}{.5}{\small $P_b'^T$}{0};
        \DTensor{\singledx,2.4}{1.2}{.5}{\small $P_b^\dag$}{0};
\end{tikzpicture}\\
=\begin{tikzpicture}[scale=0.5,baseline={([yshift=-0.65ex] current bounding box.center) }]
        \FUnitary{0,0}{1.2}{.6}{\small $V_Q$}{1};
        \DTensor{-2*\singledx,0}{1.2}{.6}{\small $P_b$}{0};
        \DTensor{-\singledx,0}{1.2}{.6}{\small $P_a$}{0};
        \DTensor{2*\singledx,0}{1.2}{.5}{\small $P_b'^\dag$}{0};
        \DTensor{2*\singledx,1.2}{1.2}{.5}{\small $P_b'^T$}{0};
        \DTensor{2*\singledx,2.4}{1.2}{.5}{\small $P_b^\dag$}{0};
        \DTensor{\singledx,0}{1.2}{.5}{\small $P_a'^\dag$}{0};
        \DTensor{\singledx,1.2}{1.2}{.5}{\small $P_a'^T$}{0};
        \DTensor{\singledx,2.4}{1.2}{.5}{\small $P_a^\dag$}{0};
\end{tikzpicture}
\end{split}
\end{align}
From there, one can immediately conclude that if \(P_a P_b = \omega P_b P_a\), and \((P_b^\dag \otimes P_b'^T \otimes P_b'^\dag)(P_a^\dag \otimes P_a'^T \otimes P_a'^\dag) = \omega' (P_a^\dag \otimes P_a'^T \otimes P_a'^\dag)(P_b^\dag \otimes P_b'^T \otimes P_b'^\dag)\), then \(V = \omega \omega' V\), so \(\omega' = \omega^*\) otherwise \(V\) is trivial. To verify the second constraint, we take \(P_a\) to be \(X\) or \(Z\) and apply the symmetry \(d\) times.
\begin{align}
\begin{split}
    \begin{tikzpicture}[scale=0.65,baseline={([yshift=-0.65ex] current bounding box.center) }]
        \FUnitary{0,0}{1.2}{.6}{\small $V_Q$}{1};=
\end{tikzpicture}
=\begin{tikzpicture}[scale=0.65,baseline={([yshift=-0.65ex] current bounding box.center) }]
        \FUnitary{0,0}{1.2}{.6}{\small $V_Q$}{1};
        \DTensor{-\singledx,0}{1.2}{.6}{\small $P_a^d$}{0};
        \DTensor{\singledx,0}{1.2}{.5}{\small $P_a'^{\shortminus d}$}{0};
        \DTensor{\singledx,1.2}{1.2}{.5}{\small $P_a'^{Td}$}{0};
        \DTensor{\singledx,2.4}{1.2}{.5}{\small $P_a^{\shortminus d}$}{0};
\end{tikzpicture}
\end{split}
\end{align}
Since we take \(P_a\) to be \(X\) or \(Z\), \(P_a^d = I\). Because any Pauli group elements has at most an order \(d\), \((P_a^\dag \otimes P_a'^T \otimes P_a'^\dag) = \omega I\) where \(\omega\) is some phase. From the above equation it follows that \(V=\omega V\), so \(\omega=1\) or \(V\) is zero. Hence, the map from \(P_i\) to \((P_i^\dag \otimes P_i'^T \otimes P_i'^\dag)\) satisfies the two constraints.


Returning to the tableau formalism, we only have to diagonalize two columns of \(C\) and set two elements of \(h\) to zero. This is strictly easier than the original problem so the algorithm can always succeed. Hence, one can always find a \(U_C\) that undoes the action of the isometry. In fact \(U_C\) is non-unique because we only have to diagonalize two columns of \(C\) and set two elements of \(h\) to zero, and the rest of \(C\) and \(h\) can be arbitrarily specified.

Lastly, if the above theorem can be generalized to arbitrary MF group, then one can prove the existence of unitaries that do not generate operator entanglement in the MF operator basis, analogous to the Clifford unitary. We leave this to the future work.

\section{Examples of MF bases}\label{mf_basis_ex}

In this section, we discuss some examples of the MF basis. We begin with discussing the case where the MF basis forms a group. We have seen the case of Weyl-Heisenberg group in \(d\) dimensions. When the dimension \(d\) is composite, for example \(d=d_1d_2\), then the MF basis can be chosen as the tensor product of two Weyl-Heisenberg groups of dimension \(d_1\) and \(d_2\), respectively. Specifically, denote the generator of the two Weyl-Heisenberg groups to be \(X_{d_1}\), \(Z_{d_1}\), \(X_{d_2}\), \(Z_{d_2}\). Then the MF group is generated by
\begin{align}
    X_{d_1} \otimes I,  Z_{d_1} \otimes I,  I \otimes X_{d_2}, I \otimes Z_{d_2}
\end{align}
The above example can be considered as a tensor product of a \(d_1\)- and \(d_2\)-qudit, and ``gluing" in the MF basis is equivalent to ``gluing" pairs of \(d_1\)- and \(d_2\)-qudits individually. Nonetheless, this is not the only option. Another option is the following set of the MF group generator.
\begin{align}
    X_{d_1} \otimes I, I \otimes X_{d_2}, I \otimes Z_{d}
\end{align}
where the \(X\) type generators are inherited from the \(d_1\)- and \(d_2\)-qudit, but the \(Z\) type generator is the clock operation \(Z_d\) of the \(d\) dimensional Weyl-Heisenberg group. 

Lastly, we point out that the MF basis need not form a group. Ref. \cite{werner2001all} gives an explicit family of MF basis, inspired by the clock and the shift operator. The MF operator is written as \(U_{ij}\) where \(i\) and \(j\) label the operator takes values from 0 to \(d-1\). \(U_{ij}\) is defined as
\begin{align}
    U_{ij} \ket{k} = H^j_{ik} \ket{\lambda(j,k)} \;.
\end{align}
Here, \(H^j\) are Hadamard matrices which means \(|H^j_{ik}|=1, \forall i,j,k\) and \(H^jH^{j \dag} = d I, \forall j\). \(\lambda(j,k)\) defines a Latin square, namely the maps \(k \mapsto \lambda(j,k)\) and \(k \mapsto \lambda(k,j)\) are injective for every \(k\). When the MF basis forms a group, the Latin square is exactly the Cayley table of the shift-like operators, but in general the operations defined in the Latin square are not closed under multiplication. That means there exist MF basis that does not form a group. We leave the investigation preparing MPS and PEPS using non-group MF basis to the future work.

\section{Insufficiency of MF symmetry to generate topological order}\label{topo_insufficiency}
In this section, we discuss why Eqs. (\ref{topo_sym}) alone does not guarantee the topological order. As a simple counterexample, consider the following state:
\begin{align}
    Q
    =\begin{tikzpicture}[scale=0.5,baseline={([yshift=-0.65ex] current bounding box.center) }]
            \PTensor{-\singledx,0}{1.5}{.6}{\small $I$}{3};
            \PTensor{\singledx,0}{1.5}{.6}{\small $I$}{1};
            \PTensor{0,-\singledx}{1.5}{.6}{\small $I$}{2};
            \PTensor{0,\singledx}{1.5}{.6}{\small $I$}{0};
            \RedArrow{-0.8-\singledx,0}{0.2}{1}
            \RedArrow{1.0+\singledx,0}{0.2}{1}
            \RedArrow{0,1.0+\singledx}{0.2}{0}
            \RedArrow{0,-0.8-\singledx}{0.2}{0}
    \end{tikzpicture}
\end{align}
This tensor results in products of Bell pairs when contracted into a two-dimensional array, which obviously does not exhibit topological order. In fact, it is known that an injective PEPS cannot support any topological order \cite{schuch2010peps,cirac2021matrix}. Without Eqs. (\ref{topo_more_sym}), the family of solution in Theorem \ref{topo_structural_theorem} is mostly injective except for a measure-zero subset. Hence, one would need to impose additional constraints to ensure non-injectivity. Indeed, imposing Eqs. (\ref{topo_more_sym}) is sufficient to ensure non-injectivity.
\begin{observation}
Under Eqs. (\ref{topo_more_sym}), with a non-trivial subgroup \(\{M\}\), \(A\) cannot be injective.
\end{observation}
\begin{proof}
We prove by contradiction. Suppose \(A\) is injective, then there exists an inverse \(A^{-1}\) (when read from the four virtual legs to the physical leg). Eqs. (\ref{topo_more_sym}) can be written as
\begin{align}
    A = e^{i \phi} A (M \otimes M \otimes M^* \otimes M^*)
\end{align}
We now apply \(A^{-1}\) to the left to obtain
\begin{align}
    I = e^{i \phi} (M \otimes M \otimes M^* \otimes M^*)
\end{align}
That leaves the only choice to be \(M=I\) and \(\phi=0\), violating our assumption that \(\{M\}\) is non-trivial.
\end{proof}



\bibliography{apssamp}

\end{document}